%% file: full_version.tex
\documentclass{article}
\newif\ifcomments
\commentsfalse
\newif\iffullv
\fullvtrue

\iffullv
\newcommand{\fullv}[1]{#1} 
\newcommand{\shortv}[1]{} 
\else
\newcommand{\fullv}[1]{}
\newcommand{\shortv}[1]{#1}
\fi 

\usepackage{xurl}
\usepackage[normalem]{ulem}
\usepackage{svg}
\usepackage{xr}
\usepackage{graphicx}
\usepackage{dcolumn}
\usepackage{tabularx}
\usepackage{xcolor}
\usepackage{bm}
\usepackage{braket}
\usepackage{mathrsfs}
\usepackage{amsmath}
\usepackage{amsfonts}

\usepackage[colorlinks]{hyperref}
\hypersetup{linkcolor=red,filecolor=blue,citecolor=magenta,urlcolor=blue}
\usepackage{amsthm}
\usepackage{cleveref}
\usepackage[T1]{fontenc}

\newcommand{\highlight}[1]{{\color{red}#1}}

\newtheorem{theorem}{Theorem}
\newtheorem{lemma}{Lemma}
\newtheorem{claim}{Claim}

\newtheorem{corollary}{Corollary}
\newtheorem{remark}{Remark}
\newtheorem{definition}{Definition}

\newtheorem{construction}[theorem]{Construction}
\newtheorem{conjecture}{Conjecture}

\newcommand{\R}{\mathbb{R}}
\newcommand{\secparam}{\lambda}\newcommand{\negl}{\mathsf{negl}}
\newcommand{\poly}{\mathsf{poly}}
\newcommand{\distr}{\cD}
\newcommand{\bit}{\{0,1\}}
\newcommand{\from}{\leftarrow}

\newcommand{\reg}[1]{{\color{orange}#1}} 

\newcommand{\pr}[1]{\Pr \bracS{#1} }
\newcommand{\ch}{\mathsf{ch}}
\newcommand{\ans}{\mathsf{ans}}
\newcommand{\tracedist}[2]{\mathsf{TD}\brackets{#1, #2}}
\newcommand{\eps}{\varepsilon}
\newcommand{\abs}[1]{\left|#1\right|}

\newcommand{\floor}[1]{\left\lfloor #1 \right\rfloor}

\newcommand{\alice}{\cA}
\newcommand{\bob}{\cB}

\newcommand{\alicetild}{\widetilde{\alice}}
\newcommand{\bobtild}{\widetilde{\bob}}
\newcommand{\eve}{Q}

\newcommand{\xhog}{\mathsf{XHOG}}

\newcommand{\gen}{\mathsf{Gen}}

\newcommand{\ver}{\mathsf{Ver}}
\newcommand{\poq}{\mathscr{P}}
\newcommand{\brackets}[1]{\left( #1 \right)}
\newcommand{\bracketsSquare}[1]{\left[ #1 \right]}
\newcommand{\bracketsCurly}[1]{\left\{ #1 \right\}}
\newcommand{\brac}{\brackets}
\newcommand{\bracC}{\bracketsCurly}
\newcommand{\bracS}{\bracketsSquare}

\newcommand{\sbrac}{\bracketsSquare}

\newcommand{\cA}{\mathcal{A}}
\newcommand{\cB}{\mathcal{B}}

\newcommand{\cD}{\mathcal{D}}

\newcommand{\cX}{\mathcal{X}}
\newcommand{\cY}{\mathcal{Y}}

\newcommand{\hash}{G}

\DeclareMathOperator*{\argmax}{arg\,max} 

\usepackage{dsfont}
\usepackage{authblk}

\shortv{
    \newcommand{\supref}[1]{{\Cref*{sup-#1} (supplement)}}
    \newcommand{\refcite}[1]{
    \ifnum\pdfmatch{,}{#1}=1 {Refs \cite{#1}}
    \else {Ref \cite{#1}}
    \fi
    }
}
\fullv{
    \newcommand{\supref}[1]{\Cref{#1}}
    \newcommand{\refcite}[1]{\cite{#1}}
}

\newif \ifemail
\emailfalse

\begin{document}

\title{On the Equivalence between Classical Position Verification and Certified Randomness}
\ifemail
\else
\author[1,3]{Fatih Kaleoglu}
\author[1]{Minzhao Liu}
\author[1]{Kaushik Chakraborty}
\author[1,2]{David Cui}
\author[1]{Omar Amer}
\author[1]{Marco Pistoia}
\author[1]{Charles Lim}
\fi

\affil[1]{Global Technology Applied Research, JPMorganChase}
\affil[2]{Massachusetts Institute of Technology}
\affil[3]{UC Santa Barbara}

\affil[1]{Global Technology Applied Research, JPMorganChase}
\affil[2]{Massachusetts Institute of Technology}
\affil[3]{UC Santa Barbara}

\date{}
\maketitle

\input{Files/abstract}

\newpage
\tableofcontents
\newpage

\input{Files/main_body}
\fullv{\subsection{Paper Organization}
\input{Files/intro}
\section*{Acknowledgments}
We thank Dakshita Khurana and Kabir Tomer for pointing out an inaccuracy in our citation of recent work. We thank James Bartusek for valuable discussions regarding the implications of our results.

\section*{Disclaimer}
This paper was prepared for informational purposes with contributions from the Global Technology Applied Research center of JPMorgan Chase \& Co. This paper is not a product of the Research Department of JPMorgan Chase \& Co. or its affiliates. Neither JPMorgan Chase \& Co. nor any of its affiliates makes any explicit or implied representation or warranty and none of them accept any liability in connection with this paper, including, without limitation, with respect to the completeness, accuracy, or reliability of the information contained herein and the potential legal, compliance, tax, or accounting effects thereof. This document is not intended as investment research or investment advice, or as a recommendation, offer, or solicitation for the purchase or sale of any security, financial instrument, financial product or service, or to be used in any way for evaluating the merits of participating in any transaction.
}
\input{Files/main_body2}
\newpage
\input{Files/prelims}
\input{Files/rcs_rom}
\input{Files/computational_entropy}
\input{Files/crea}
\input{Files/instantiation}

\bibliographystyle{alpha}
\bibliography{refs}

\end{document}

%% file: Files/abstract.tex
\begin{abstract}

Gate-based quantum computers hold enormous potential to accelerate classically intractable computational tasks. Random circuit sampling (RCS) is the only known task that has been able to be experimentally demonstrated using current-day NISQ devices. However, for a long time, it remained challenging to demonstrate the quantum utility of RCS on practical problems. Recently, leveraging RCS, an interactive protocol generating \emph{certified randomness} was demonstrated using a trapped ion quantum computer \cite{jpmc_cr}, advancing the practical utility of  near-term gate-based quantum computers. In this work, we establish a strong connection between certified randomness and another quantum computation classical communication primitive, \emph{classically verifiable position verification} (CVPV), which circumvents the practical challenges that may arise from long-distance quantum communications. We provide a new generic compiler that can convert any single-round proof of quantumness based certified randomness protocol into a secure classical communication-based position verification scheme. Later, we extend our compiler to different types of multi-round protocols. Notably, our compiler can be applied to any multi-round certified randomness protocol that can be analyzed using the entropy accumulation theorem \cite{DFR20}, making its applicability very general. Moreover, we show that CVPV is equivalent to a relaxed variant of certified randomness that we define. We instantiate each of our compilers using existing certified randomness protocols. In particular, building on the work of Aaronson and Hung \cite{AH23}, we give a NISQ-friendly instantiation based on RCS, which was experimentally demonstrated by Liu et al. \cite{jpmc_cr}. Hence, we show that CVPV is another application within reach of NISQ devices.
\end{abstract}

%% file: Files/main_body.tex
\maketitle
\renewcommand{\theequation}{\arabic{equation}}

\section{Introduction}

Significant effort and investment have been devoted to the study and development of quantum computers due to their potential to accelerate a range of computational tasks that are classically intractable. Researchers have explored a myriad of applications, assessing the potential advantages \cite{shor1994algorithms, harrow2009quantum, Shaydulin2024, liu2021rigorous, berry2007Efficient}. Several problems \cite{shor1994algorithms, jordan2024optimization, babbush2023exponential} have been theoretically shown to achieve exponential quantum speedup over the best-known classical algorithms. More recently, researchers have made progress in showing the utility of quantum computational advantage, either by trying to solve problems of practical relevance \cite{Kim2023} or demonstrating prototypes of quantum error correction \cite{bluvstein2024logical,Paetznick2024demonstration,willow} and scalable architectures \cite{aghaeerad2025scaling}. Furthermore, there exist cryptographic problems that are impossible to solve using classical methods alone. Two such primitives are \emph{certified randomness} and \emph{position verification} \cite{buhrman2014position}. 

\par The goal of certified randomness is to classically validate the output of an inherently random quantum process. It has potential applications to many areas including cryptography, blockchain, and finance \cite{ACC+25}. Traditionally, constructions for certified randomness have relied on the violation of Bell inequalities \cite{acin2016certified, colbeck2011private, pironio2010random, shalm2015strong}. However, these protocols require distant parties to share entanglement via quantum communications, which is susceptible to channel loss and noise. 
To circumvent this issue, a new research direction, termed \emph{quantum computing and classical communication} (QCCC), aims to replace quantum communication with quantum computers that use classical communication. Within the QCCC framework, several constructions for certified randomness have been developed \cite{AH23, BCMVV21, mahadev2022efficient, YZ24}.

\par In the near future, it is likely that most quantum computers will be available as a cloud service. Consider a client who wishes to outsource a computational task using an external quantum cloud provider. Suppose that the computational task is to be performed on sensitive data subject to strict regional regulations, whereby the data is forbidden to leave a certain region. A notable example is the General Data Protection Regulation \cite{gdpr} which stipulates that personal data of the European Union cannot be stored outside that region. In such a scenario, the client would want to ensure that the cloud provider operates in the region which it claims to be operate in, lest the client should be liable in case of a cyberattack. 

\par Position-based cryptography offers a solution to this problem by using geographic location as a credential to offer enhanced security in multiparty protocols. This concept, introduced in \refcite{chandran2009position}, has been studied in multiple settings, such as wireless and quantum communications \cite{brands1993distance,bussard2004trust,capkun2005secure,singelee2005location,sastry2003secure,vora2006secure,zhang2006secure}. Position verification is the most fundamental position-based cryptographic functionality, where the goal is to allow multiple parties to verify the geographic location of a prover. 
\par Position verification using only classical resources is impossible via a generic attack \cite{chandran2009position} which targets two weaknesses: (1) clonability of messages and (2) the ability to \emph{pull out} randomness from a classical algorithm. The latter refers to the fact that without loss of generality, any randomized classical algorithm first generates a uniformly random seed and then performs a deterministic computation on its input combined with the random seed.
A typical quantum position verification would circumvent this attack by avoiding the first weakness using quantum communication, thanks to the no-cloning theorem. However, quantum communication creates unique challenges for position verification \cite{qi2015loss}. As an alternative to quantum communication, Liu et al. \cite{LLQ22} showed how to circumvent the impossibility result using computational assumptions, namely \emph{proof of quantumness} (PoQ), to avoid the second weakness. This initiated the study of position verification in the QCCC framework, hence \emph{classically verifiable position verification} (CVPV). In the language of the cloud scenario above, CVPV has the benefit over regular quantum position verification that the client need not possess any quantum resources.

In the QCCC framework, most of the existing certified randomness protocols \cite{BCMVV21,YZ24}, as well as the CVPV construction of \refcite{LLQ22} would require a fault-tolerant quantum computer, therefore those constructions are beyond the reach of noisy intermediate-scale quantum (NISQ) devices.

Early works to demonstrate quantum computational advantage with NISQ devices relied on sampling problems, such as random circuit sampling (RCS) \cite{Arute2019,Wu2021,Zhu2022,morvan2023phase,decross2024computational}. 

Aaronson and Hung \cite{AH23} proposed a certified randomness protocol with inefficient verification based on RCS. This is both the only known application of RCS to a useful task, and the only certified randomness protocol that is within the reach of NISQ devices. 

\par In this spirit, a recent work \cite{jpmc_cr} experimentally demonstrated a practical certified randomness protocol leveraging the result of Aaronson and Hung \cite{AH23}. Through an interactive protocol, a trapped ion quantum computer was used to generate certified randomness. In this protocol, a classical verifier sends randomly generated quantum circuits to a prover with a quantum computer. The prover performs RCS on these circuits and sends the measured bit-strings back to the verifier. If the bit-strings pass a certain statistical test, they are guaranteed to have some entropy, yielding certified randomness. Currently, this approach stands out from the other certified randomness protocols as it does not require fault-tolerant quantum devices.

\par Position verification does not fall short of certified randomness with regards to benefiting from the QCCC framework and NISQ-friendly construction. The protocol proposed in \refcite{LLQ22} relies on the hardness of learning with errors (LWE) assumption and requires a fault-tolerant quantum computer out of reach for NISQ devices.
\fullv{In this paper, we pose the following question and answer it in the affirmative:

\vspace{0.1in}
\textit{Q1: Is it possible to design a classical position verification protocol that can be implemented using a NISQ device?} 
\vspace{0.1in}
}
\par \shortv{In this work, we propose a CVPV protocol that can be implemented using a NISQ device. }To achieve this, we leverage the NISQ-friendly task of RCS. We certify the position of the prover by making it generate random circuit samples at the prescribed time and location. This is ascertained by designing the protocol such that the information required to generate the samples, i.e. the circuit description, simultaneously reaches the prover from multiple directions. The prover is then asked to send identical samples of high statistical score to all directions. Our construction paves the way for another experimental demonstration of quantum advantage following \refcite{jpmc_cr}.

\fullv{In this paper, we answer the following question:

\vspace{0.1in}
\textit{Q2: What is the fundamental property needed for designing a secure position verification scheme with classical verifiers?}
\vspace{0.1in}
}
\par \shortv{Along the way, we study the fundamental property needed to design a secure position verification scheme by establishing a strong connection with another seemingly unrelated primitive: certified randomness. }We show that a general class of PoQ protocols with the \emph{certified randomness} property can be used to construct secure classical position verification. As a corollary, we show that CVPV exists in the random oracle model assuming one of the following: (1) classical hardness of random circuit sampling (with inefficient verifiers), (2) quantum hardness of learning with errors (LWE) problem, (3) or the Aaronson-Ambainis Conjecture \cite{AA14}. Furthermore, in the other direction, we show that any CVPV protocol can be compiled into a certified randomness protocol.
\par Therefore, we show that certified randomness is necessary and sufficient for constructing a position verification protocol within the QCCC framework.
This connection enables the development of new classes of position verification protocols derived from certified randomness protocols, and vice versa.

\iffullv \subsection{Our Results} \else
\section{Results} \fi
We give several compilers to achieve CVPV using certified randomness. We show that certified randomness is also necessary for CVPV, hence a fundamental property required to construct it. Then, we show how to instantiate our compiler using a NISQ-friendly RCS-based construction, as well as theoretically significant constructions based on error-correcting codes and LWE.
\subsection{Compilers} \label{sec:compilers}
\fullv{To address the second question (Q2), we }\shortv{We }prove that a single-round protocol that achieves certified randomness with classical verifiers is sufficient for the secure position verification scheme with classical verifiers in the quantum random oracle model. 
\begin{theorem}[Single-Round Compiler - Informal] 
\label{thm:single_comp_inf}Suppose there exists a single-round certified randomness protocol with a classical verifier and computationally bounded quantum prover, then there exists a position verification protocol with classical verifiers secure in the quantum random oracle model.
\end{theorem}

We first consider any single-round PoQ-based certified randomness protocol $\poq = (P,V)$, where $P$ is the prover and $V$ is the verifier. It is modeled as a two-message interactive protocol. In the first message, the verifier sends a random challenge $\ch$ to the prover. In the second message, it gets back $\ans$ as a response. Later, the verifier performs a verification process $\ver$ to accept or reject the response. 

\begin{figure*}[t]
    \centering
    \includegraphics[width=\textwidth]{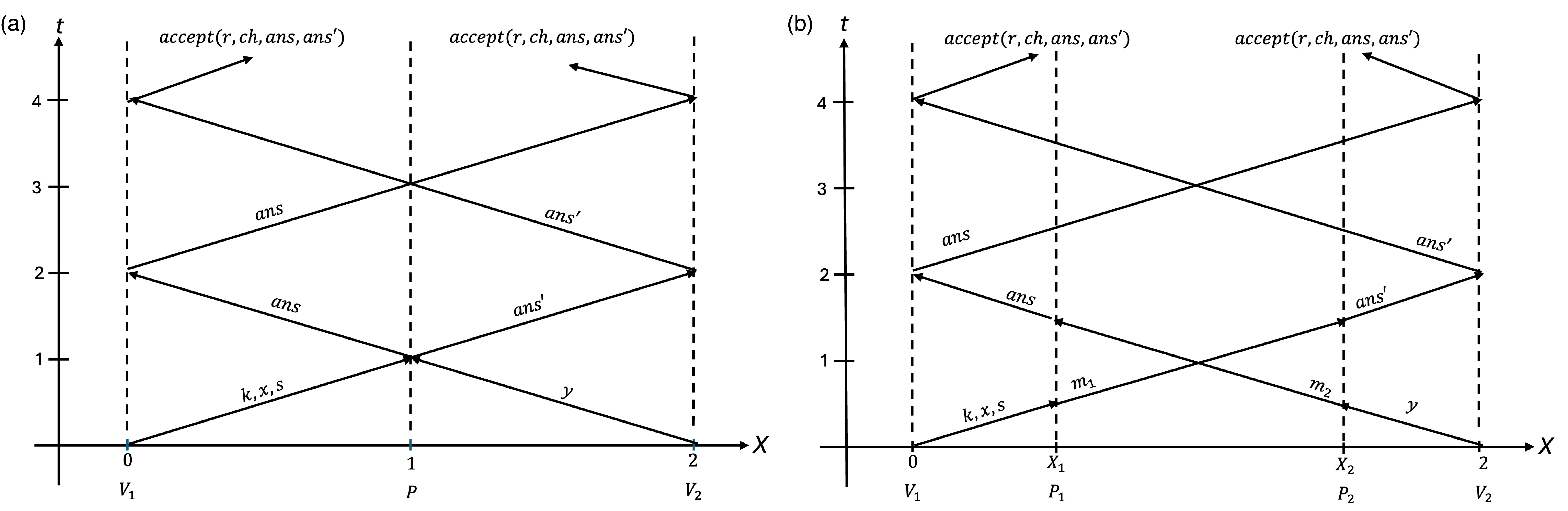}
    \caption{(a) A spacetime diagram describing the honest interaction between two verifiers at locations $X=0$ and $X=2$ and an honest prover located at the claimed position (for simplicity, located at $X=1$ in the diagram). The verifiers pre-share random values $x,y,k,r$ prior to the start of the protocol, and set $s = ch\oplus G_k(x\oplus y)$ for some family of hash functions $\{G_k\}$ and $ch = \mathsf{Gen}(1^n; r)$ for $\mathsf{Gen}$ the verifier functionality of some certified randomness verifier and where $accept(\cdot)$ determines a boolean function that determines whether the verifiers accept the prover's claim. In all spacetime figures we provide, we implicitly require that the proper timing constraints are observed in order for the verifiers to accept. (b) A spacetime diagram describing the honest interaction between two verifiers at locations $X=0$ and $X=2$ and two dishonest prover located at arbitrary positions between $V_1$ and $V_2$ (for simplicity, located at $X=0.5$ and $X=1.5$ in the diagram). The provers are allowed to communicate during the protocol but are non-signaling. Notation and behavior is otherwise analogous to (a). In particular, verifiers expect to receive messages at times determined by honest prover's position.}
\label{fig:sing_round_dis}
\end{figure*}

\par \textbf{Construction:} In our compiler, we consider that the two verifiers (say $V_1,V_2$) in the CVPV protocol have access to a family of cryptographic hash functions $\bracC{\hash_k}_{k \in \bit^\secparam}:\bit^{m} \to \bit^n$. In addition, they share a random hash key $k$ and two random inputs $x,y \leftarrow \{0,1\}^m$. At time $t = 0$, $V_1$ sends $s := G_k(x\oplus y) \oplus \ch$, $x,k$ to the prover $P$ at a claimed location (say $X = 1$ in \Cref{fig:sing_round_dis}) and $V_2$ sends $y$ to $P$. Upon receiving $x,s,k,y$ from the verifiers, the prover can recover $\ch$ and send back the response $\ans$ to the verifiers. The verifiers verify the position of the prover if the following conditions are satisfied. 

\begin{enumerate}
\item \textbf{Timing Requirement:} Both verifiers receive at time $t=2$.
\item \textbf{Consistency Requirement:} The received answers are identical.
\item \textbf{Certified Randomness Verification:} The verification process $\ver$ corresponding to the certified randomness protocol accepts the response.
\end{enumerate} 
We refer to \Cref{fig:sing_round_dis} for the schematic diagram of the compiler. For clarity, we also depict general adversarial behavior in \Cref{fig:sing_round_dis}.

\par Note that \refcite{unruh2014quantum} also uses the idea to hide a secret by hashing the XOR of two messages that arrive from opposite directions. The difference is that \refcite{unruh2014quantum} hides the measurement basis in a QPV protocol while sending the quantum state in the clear, whereas we hide the certified randomness challenge which is the entire input of the prover in our CVPV protocol.

\par In the literature, most of the existing certified randomness protocols are multi-round and it is not immediately clear whether we can draw a conclusion similar to \Cref{thm:single_comp_inf} for the multi-round protocols. Below, we discuss several multi-round protocols, and we leave the detailed descriptions of these protocols in the Supplementary Material.

\par Our first multi-round compiler is a natural generalization of the single-round compiler. For each round of the original certified randomness protocol, we use the interactive construction just like in the single-round compiler. However, this compiler works only for the multi-round protocols that achieve a stronger notion of certified randomness, namely those with the so called \textit{Sequential Decomposition} property.

\begin{definition}[Sequential Decomposition Property - Informal]\label{def:seq_decomp_inf}
A multi-round PoQ protocol is said to have the sequential decomposition property if the following holds: Consider an unbounded guesser trying to guess the prover's answer. In each round, the guesser is only allowed to communicate with the prover after the answer and the guess are sent to the verifiers. The guesser must fail at least one round with high probability.
\end{definition}

\begin{theorem}[Multi-Round Sequential Compiler - Informal] 
\label{thm:mult_seq_comp_inf} Suppose there exists a multi-round certified randomness protocol, with a classical verifier and a computationally-bounded quantum prover, that satisfies the sequential decomposition property. Then, there exists a position verification protocol with classical verifiers secure in the quantum random oracle model.
\end{theorem}

It is not clear whether these existing protocols like \refcite{BCMVV21,AH23} would satisfy the sequential decomposition property that we need for our multi-round compiler. Although these protocols provide a lower bound on the smooth min-entropy conditioned on some side information, our setting allows communication with a guesser that is allowed to receive a copy of the classical outputs, resulting in zero entropy.

\par The key enabling insight of this work is as follows: to be helpful in assisting the guesses, the side information for the $i$th round must be present in the $i$th round, not before or after. Although the conditional entropy cannot be lower bounded due to communication, it should not affect the guessing probability. Here, we show that protocols in \refcite{BCMVV21,AH23} indeed satisfy the sequential decomposition property by proving the following theorem.

\begin{theorem}[Existence of Multi-Round Certified Randomness Protocol with Sequential Decomposition Property - Informal] 
\label{thm:non-sig_seq_inf}
If a multi-round certified randomness protocol, with a classical verifier and a computationally bounded quantum prover, has non-zero single-round von Neumann entropy ``on average'' when it does not abort (e.g. in \refcite{AH23,BCMVV21}), then that certified randomness protocol also satisfies the sequential decomposition property.
\end{theorem}
In other words, any PoQ certified randomness protocol that can use the entropy accumulation theorem (EAT) \cite{DFR20} for the multi-round security proof satisfies this property. EAT is a powerful framework for reducing multi-round entropy bounds to single-round entropy bounds even for an adaptive adversary in the device independent setting, and has significantly reduced the complexity of theoretical analysis of many cryptographic protocols \cite{AH23,BCMVV21,ADFRV18,Pirandola:20}. This definition is highly permissive and allows potential future certified randomness protocols to be compiled into a position verification protocol.

\par In addition to the aforementioned sequential compiler, we consider a \emph{rapid-fire} compiler. Here, the verifiers send all the challenges that are related to the multi-round certified randomness protocol sequentially with a small predetermined time gap. The difference with the previous approach is that here, the verifiers do not wait for the answers to arrive from the prover before sending the next challenge. This protocol achieves the desired security in the QROM from the multi-round certified randomness protocols without the sequential decomposition property like in \refcite{BCMVV21,AH23}. We get the following result in this direction.

\begin{theorem}[Multi-Round Rapid-Fire Compiler - Informal] 
\label{thm:mult_rapid_comp_inf}Suppose there exists a multi-round certified randomness protocol with a classical verifier and a computationally-bounded quantum prover, then under some restricted communication assumptions, there exists a position verification protocol with classical verifiers secure in the quantum random oracle model.
\end{theorem}

Finally, we combine both previous approaches and provide a new compiler called the \textit{Sequential Rapid-Fire Compiler} that shows prospects in overcoming the limitations of the previous multi-round compilers.

\begin{theorem}[Multi-Round Sequential-Rapid-Fire Compiler - Informal] 
\label{thm:mult_seq_rapid_comp_inf}Consider the multi-round position verification protocol that exists due to \Cref{thm:mult_rapid_comp_inf}. A sequential compilation of this protocol with communication restrictions within each repetition but not between repetitions is also a position verification protocol with classical verifiers secure in the quantum random oracle model.
\end{theorem}

\subsection{Instantiations}
Later, in \supref{sec:inst}, we instantiate our compilers with a random circuit sampling (RCS)--based multi-round certified randomness protocol proposed by Aaronson and Hung \cite{AH23}. Our construction achieves two advantages over the existing protocol in \refcite{LLQ22}. First, it is NISQ-friendly\fullv{, thereby answering the first question (Q1) in the affirmative}. Second, it achieves CVPV under the hardness assumption of RCS, which is does not imply the LWE-hardness assumption used by \refcite{LLQ22}. 

\par Our single-round compiler can also be instantiated using the certified randomness protocol of Yamakawa and Zhandry \cite{YZ24}, which shows that there exists a CVPV scheme secure in the random oracle model assuming the Aaronson-Ambainis Conjecture \cite{AA14}; the details can be found in \supref{sec:inst_YZ}.

\par As another example, our compiler can be instantiated using the work of Brakerski et al. \cite{BCMVV21}, which yields the result of \refcite{LLQ22} that CVPV exists in the random oracle model assuming the hardness of LWE, albeit with a different construction.

\fullv{\input{Files/prior_work}}

\shortv{\section{Discussion}}
\fullv{\subsection{Discussion}}

Unlike position verification based on quantum communication, our RCS-based protocol uses classical communication and can be much more resilient against environmental noise encountered during communication. However, the use of a quantum computer introduces new opportunities for noise to degrade the quality of the protocol. In the RCS-based instantiation, this effect manifests as a reduction in the achievable certified randomness XEB score, which reduces the entropy guarantee.

\par We also highlight the remaining problems related to classical verifier position verification. In this paper, we prove the soundness of the compilers in quantum random oracle model. However, it is crucial to have the security analysis in the plain model for a practical implementation of this protocol. Additionally, this paper requires proof of quantumness-based certified randomness to construct the CVPV protocol. However, it is not clear whether this is the most minimal assumption. We leave the minimum sufficient assumption required to construct a CVPV protocol as an open problem.

\par Practical implementation of the RCS-based protocol also faces additional challenges. Thus far, the security model considered in the experimental demonstration of RCS-based certified randomness assumes a family of finite-size adversaries using existing state-of-the-art techniques instead of the more general asymptotic adversaries considered in \refcite{AH23}. This restriction was a result of the challenging conditions imposed on the experiment due to a combination of limiting factors, including the achievable quantum computing fidelity, latency, and the limited classical computational budget of the verifier to complete the exponentially expensive verification task. Additional work remains to close the gap between experimental feasibility and security under more general adversary models.

\par Further, there are additional challenges associated with position verification. First, one needs to establish line-of-sight links between the verifiers and the prover, or at least whatever routing cannot substantially increase the total communication time. Second, the timing requirement is very stringent, which is especially difficult for current quantum computing devices. If we assume that an adversary can execute the quantum circuit instantaneously, the positional uncertainty is at least the speed of light times the time it takes for an honest quantum computer to execute the quantum circuit. Further experimental progress on latency reduction is needed to make position verification practical.

%% file: Files/prior_work.tex
\fullv{\subsection{Prior Work}}
\shortv{\section{Prior work}} \label{sec:prior_work}
In \cite{chandran2009position}, Chandran et al. proved the impossibility of a classical position verification scheme in the standard model. In \cite{kent2011quantum}, Kent et al. first proposed the idea of quantum communication-based position verification under quantum tagging. Later, in \cite{buhrman2014position}, Buhrman et al. re-initiated the study of this topic and proved the impossibility of designing an information-theoretically secure quantum communication-based position verification protocol. The authors provide a generic strategy to attack any position verification protocol using the instantaneous non-local computation technique proposed by Vaidman \cite{vaidman2003instantaneous}. However, for this attack, the malicious provers need to share a doubly exponential (in the security parameter) number of EPR pairs. Later, in \cite{beigi2011simplified}, the author reduced the entanglement requirement to exponential (in the security parameter) using port-based teleportation. On the other hand, in \cite{buhrman2014position}, Buhrman et al. showed if the adversaries do not have access to EPR pairs, then it is possible to design a secure quantum communication-based position verification scheme. Many protocols in the literature achieve security against non-entangled adversaries, but they become vulnerable to adversaries with an exponential amount of entanglement \cite{chakraborty2015practical,unruh2014quantum,junge2022geometry,qi2015loss,allerstorfer2023security}. In the standard model, the security of a quantum communication-based position verification scheme against adversaries with a polynomially bounded amount of entanglement remains open. However, in the random oracle model, Unruh \cite{unruh2014quantum} proves the security of a quantum communication-based position verification protocol against adversaries with unbounded shared entanglement. Our compilers use share techniques of query-extraction with Unruh \cite{unruh2014quantum}. While \cite{unruh2014quantum} makes a reduction to a \emph{monogamy-of-entanglement} game in the quantum-communication setting, we make a reduction to \emph{certified randomness} in the classical-communication setting.

Recently, in \cite{bluhm2022single}, Bluhm et al. proposed a noise-robust protocol that uses only a single qubit quantum resources and some classical communication. Interestingly, the adversary's quantum resource for any attack strategy increases with the classical communication resources in the protocol. Later, inspired by this result in \cite{bluhm2022single}, Allerfoster et al. propose loss-tolerant and noise-tolerant protocols that are within reach of today's quantum communication technology. However, all these protocols still suffer from the distance limitation of the quantum communication \cite{allerstorfer2023making}.

In \cite{LLQ22}, Liu et al. initiate the study of designing position verification protocols that are based on proof of quantumness and classical communication. The authors show that the proofs of quantumness is necessary to design a secure classical verifier position verification protocol. Moreover, the authors prove that if the prover has access to a fault-tolerant quantum computer, then under the LWE hardness assumption, one can design a secure classical communication-based position verification protocol, hence beating the impossibility result proposed by Chandran et al. \cite{chandran2009position}.

%% file: Files/intro.tex
\fullv{We give an overview of our techniques in \Cref{methods}.}
\shortv{We discuss prior work in \Cref{sec:prior_work}.} We recall some relevant preliminary concepts in \Cref{sec:prelims}. In \Cref{sec:const_rom}, we provide our compiler for CVPV from single-round certified randomness and prove the security of the construction, our first major result. In \Cref{sec:mr} we provide a number of methods of generalizing this compiler to allow for multi-round certified randomness protocols. Specifically, in \Cref{sec:mr_seq} we show that certified randomness protocols with a property we refer to as sequential decomposability suffice to allow for a natural generalization of our compiler for multi-round settings, and in \Cref{sec:mr_rep}  we prove (with a full proof in \Cref{sec:seq_crea}) that a class of natural certified randomness protocols satisfy the necessary property. In \Cref{sec:mr_rf}, and \Cref{sec:seq_rf} we prove the security of an alternate compilation method, based on additional timing constraints, which has advantages in the idealized model at the cost of practical robustness. In \Cref{sec:inst}, we show that the well known near-term proposal for certified randomness due to Aaronson and Hung suffices to instantiate our compiler. The rest of \Cref{sec:inst_all} shows other instantiations of our compiler that are not NISQ-friendly but theoretically relevant.

%% file: Files/main_body2.tex
\iffullv \section{Technical Overview} \else
\section{Methods} \fi \label{methods}
\subsection{Single Round Compiler}

The informal construction and schematics of the single round compiler can be found in the main text, so we omit them here.

\textbf{Soundness Proof Sketch (Proof Sketch of \Cref{thm:single_comp_inf}).}
We reduce the soundness of the CVPV protocol above to the certified randomenss property of $\poq$. We first make the connection to certified randomness by imagining an ideal game, where $P_1$ and $P_2$ are simultaneously given $\ch$ and asked to output the \emph{same} correct answer $\ans$. If no communication is allowed, then beating this game violates certified randomness due to no-signalling. Intuitively, certified randomness dictates that $\ans$ cannot be computed deterministically from $\ch$, which forces $P_1$ and $P_2$ to perform their own local random computation to output a valid answer.
\par In the real security game, the 
provers $P_1,P_2$ effectively receive secret shares $x,y$ of $\ch$, where $x \oplus y = \ch$, and then get to perform one round of simultaneous communication. If they simply forward their shares, then this would be equivalent to the ideal game. However, we need to argue that this \emph{challenge-forwarding} adversary is optimal. To do so, we rely on the fact that $x \oplus y$ is information-theoretically hidden from the provers before they communicate. 
\par One would hope that this in turn hides any \emph{useful} information about $\ch$ during the same timeframe. Nonetheless, it is not clear how to argue this directly. Namely, one needs to rule out \emph{homomorphic} attacks, where $P_1$ performs a quantum computation on input $x$, $P_2$ on input $y$, and then they can each deterministically recover the same output $\ans$, which could be obtained by running the honest prover of $\poq$ on input $\ch$. This issue was previously faced by Unruh \cite{unruh2014quantum} in the context of quantum position verification, where the classical basis information had to be hidden from the adversary in this strong sense.
\par To circumvent this issue, we follow \refcite{unruh2014quantum} and use a cryptographic hash function $\hash$ and encrypt $\ch$ with a one-time-pad using $\hash(x \oplus y)$. The property we need from $\hash$ is \emph{query-extractability}, and accordingly we show security in the quantum random oracle model. This technique was also used by Liu et al. \cite{LLQ22} to achieve security against unbounded entanglement.

\subsection{Multi-Round Sequential Compiler}

For this compiler, we first start with a multi-round (say $\ell$-round) certified randomness protocol. Similar to the single-round certified randomness protocol, we can formulate any $\ell$-round certified randomness protocol to a $2\ell$-communication round interactive protocol, where at round $i \in [\ell]$, the verifier sends a random challenge $\ch_i$ to the prover and gets back $\ans_i$ from it. The next round starts after the verifier receives the answer from the prover. 

\textbf{Construction:} The compiler corresponding to such an $\ell$-round certified randomness protocol is a sequential repetition of the interactive portion of the single-round compiler, followed by the necessary testing of the entire transcript. Before the beginning of the protocol, i.e., at $t = -\infty$ the verifiers ($V_1,V_2$) share $\ell$ random hash keys $\{k_i\}_{i \in [\ell]}$, $\ell$ random input pairs $\{(x_i,y_i)\}_{i \in [\ell]}$, and $\ell$ random challenges $\{\ch_i\}_{i \in [\ell]}$ corresponding to the challenges of the certified randomness protocol. On the $i$-th ($i \in [\ell]$) round, $V_1$ sends $s_i := G_{k_i}(x_i\oplus y_i) \oplus \ch_i$ and $V_2$ sends $y_i$ to the prover $P$ at a claimed location $X = 1$ (see \Cref{fig:seq_hon_comp} for reference). Upon receiving $k_i,x_i,s_i,y_i$, the prover computes $\ch_i$, and sends back the answer $\ans_i$ corresponding to the challenge $\ch_i$ to the verifiers. Suppose, the verifiers send the challenges at time $t^{\text{send}}_i$, and receive the answers at time $t^{\text{rec}}_i$. At the end of the protocol, the verifiers accept the claimed location of the prover if the answers satisfy the timing constraint, i.e. $t^{\text{rec}}_{i} - t^{\text{send}}_i = 2$ for all $i\in [\ell]$, the consistency check, i.e. $\ans_i = \ans'_i$ for all $i \in [\ell]$, and the certified randomness constraint. We refer to \Cref{fig:seq_hon_comp} for the schematic diagram of the protocol and \Cref{fig:seq_dis} for the behavior of cheating provers. 

\begin{figure}[t]
    \centering
    \includegraphics[scale= 0.5]{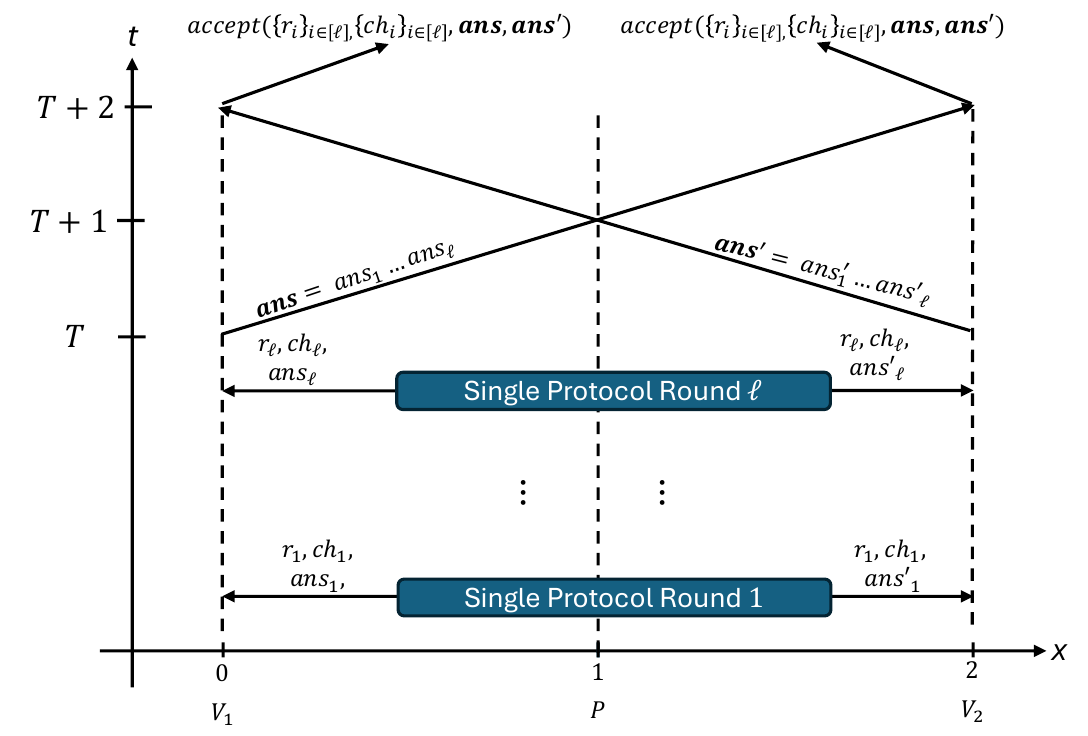}
    \caption{A spacetime diagram describing the honest interaction between two verifiers at locations $X=0$ and $X=1$ and an honest prover located at the claimed position (for simplicity, located at $X=1$ in the diagram). For brevity, we model the execution of the challenge-response portion of each round of the multi-round protocol as a black box that provides the relevant randomness, challenge, and answers from each round, though this abstraction is not used in our proof. Following some $\ell$ rounds of the protocols, the verifiers engage in a final interaction to accept or reject the protocol. As in \Cref{fig:sing_round_dis} and elsewhere, we omit depicting the details of the expected in-round timing constraints, which are detailed in the text.}
\label{fig:seq_hon_comp}

\end{figure}
\begin{figure}[ht]
    \centering
    \includegraphics[scale= 0.5]{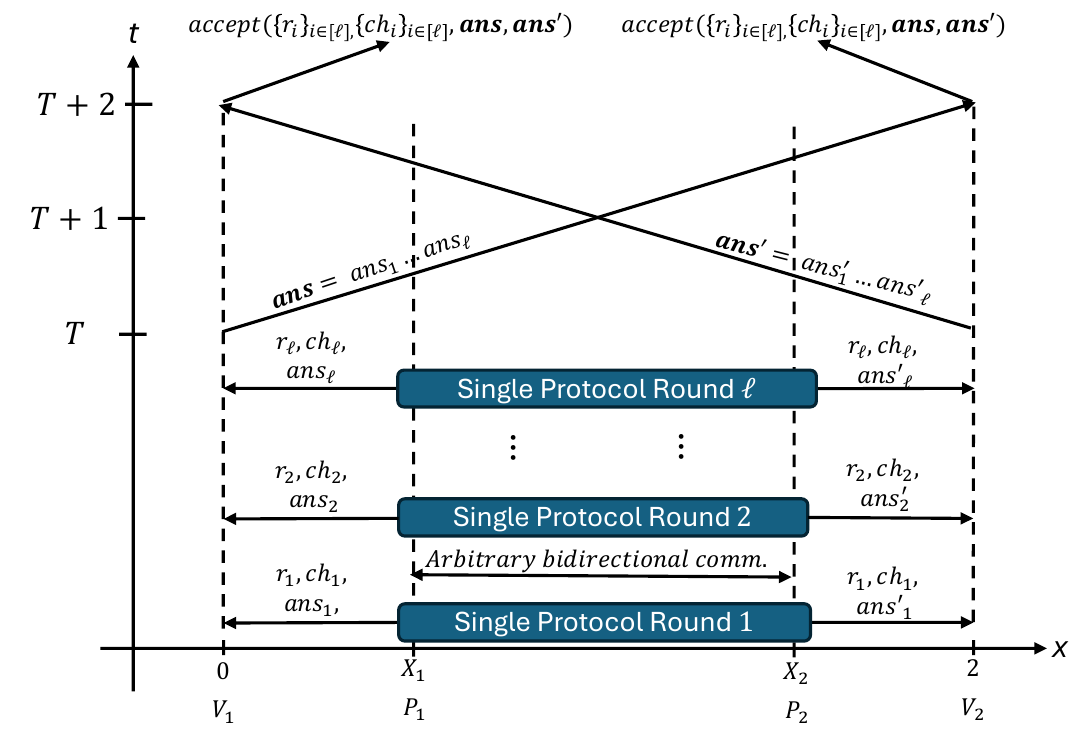}
    \caption{A spacetime diagram describing the interaction between two verifiers at locations $X=0$ and $X=2$ and two dishonest prover located at arbitrary positions between $V_1$ and $V_2$ (for simplicity, located at $X=.5$ and $X=1.5$ in the diagram). The provers are allowed arbitrary setup between protocol rounds, as in \supref{def:cr_multiround}. Notation and behavior is otherwise analogous to \Cref{fig:seq_hon_comp}.}
\label{fig:seq_dis}
\end{figure}

\textbf{Difficulty of the Soundness Proof.} The adversarial model of all the existing multi-round certified randomness protocols \cite{AH23,BCMVV21} use entropy accumulation theorem (EAT) to calculate a lower bound on the min-entropy of the produced outcomes for all the rounds. More precisely, the multi-round certified randomness protocol in \refcite{AH23} models the entire adversary channel as an entropy accumulation channel (EAT) \cite{DFR20}. Although the original security analysis in \refcite{BCMVV21} did not directly use EAT, later, Merkulov and Arnon-Friedman show in \refcite{merkulov2023entropy} that the adversarial channel of \refcite{BCMVV21} can indeed be modeled as an EAT channel. The EAT channel that is proposed in \refcite{DFR20} does not allow the prover to communicate its private registers to the external adversary during the runtime of the protocol. However, in our sequential compiler, after the $i$-th round the malicious provers can communicate with each other, and exchange their internal registers as well as their answers for the $i$-th round. We refer to \Cref{fig:seq_dis} for an example. This stops us from applying the security analysis from \refcite{DFR20,BCMVV21,AH23} directly. As a way out, we require that our certified randomness protocols satisfy a stronger notion of security, namely sequential decomposition that we define informally in \Cref{def:seq_decomp_inf} (formally stated in \supref{def:cr_multiround}).
\textbf{Soundness Proof Sketch (\Cref{thm:mult_seq_comp_inf}).} Similar to the soundness proof of the single round protocol, here we also reduce the soundness of the multi-round CVPV protocol to the multi-round certified randomness protocol. Here, the only difference is that due to the multi-round nature of the compiler, we need to consider a multi-round ideal guessing game. In this ideal guessing game on the $i$-th round $P_1$ and $P_2$ are simultaneously given $\ch_i$ and asked to output the \emph{same} correct answer $\ans_i$. Note that, the usual definition of the certified randomness do not provide any guarantee on the winning probability of this guessing game. Therefore, we need to add the sequential decomposition property. Indeed, if no communication is allowed then due to the sequential decomposition property of the multi-round certified randomness protocol the winning probability of this guessing game will be negligible. The rest of the reduction to the real security game is similar to the proof sketch of \Cref{thm:single_comp_inf}. We refer to \supref{thm:mr_seq_sec} for a more detailed analysis.

\subsection{Sequential Decomposition from Repetition}
We strengthen our result on multi-round compilers by showing the existence of a family of certified randomness protocols that are secure under our notion of sequential decomposition (\Cref{def:seq_decomp_inf} or \supref{def:cr_multiround} for a more formal version). In \Cref{thm:non-sig_seq_inf}, we prove that certified randomness protocols based on repetition with single round entropy guarantees, including the well known protocols in \refcite{AH23,BCMVV21}, satisfy \supref{def:cr_multiround}.

\textbf{Proof Sketch (\Cref{thm:non-sig_seq_inf}).}  
We prove \Cref{thm:non-sig_seq_inf} by using the key observation that the prover's answers from the previous rounds do not help the guesser after the guess is already committed for each round. First, we consider an optimal adversary (maximum probability of succeeding the proof of quantumness test of the certified randomness protocol and the consistency check of the sequential decomposition property) where the prover may send information about previous answers to the guesser. We then consider a slightly modified adversary that has the same optimal success probability, but only the guesser is allowed to send information to the prover. For this modified adversary, the guesser simply assumes that all guesses are correct thus far. If at least one guess is incorrect, the protocol has already failed and the adversary strategy from then on does not matter. If the all guesses are correct, then pretending the answers are always the same as the guesses results in the correct behavior.

Specifically, we allow the guesser to be unbounded and prepare arbitrary quantum memory for both the guesser and the prover, and the quantum memory state is exactly that of the original optimal adversary (conditioned on classical outcomes on the answers and guesses agreeing with the answers). This adversary must have the same success probability as the original adversary and is therefore optimal. Further, the prover is no longer allowed to communicate with the guesser, and one can lower bound the entropy conditioned on the guesser side information and upper bound the protocol success probability.

\subsection{Rapid Fire Compiler}
In the CVPV protocols that are based on our sequential compiler, the verifiers need to wait for the answer to arrive from the prover before starting the next round. This may introduce some unwanted delay and make the CVPV protocol time consuming. There may also be certified randomness protocols that do not satisfy \Cref{def:seq_decomp_inf}. We show that one can overcome this drawback by sending the challenges without waiting for the responses from the provers. We refer to this compiler as the rapid fire compiler. 

\textbf{Construction:} Similar to the sequential compiler here, the verifiers share $\{k_i,x_i,y_i,\ch_i\}_{i \in [\ell]}$. Moreover, the verifiers also share a fixed time interval $\Delta$. During the protocol, verifiers send challenges to the prover in every $\Delta$ interval. If the protocol starts around time $t=0$ then the $i$-th round starts at time $t = (i-1)\Delta$. On the $i$-th round, $V_1$ sends $s_i,x_i,k_i$, and $V_2$ sends $y_i$ to $P$. Similar to the sequential compiler, after collecting all the responses for the $\ell$ rounds, the verifiers accepts the location of the claimed prover if it passes all the three checks. Note that, here the verifiers send the challenges in every $\Delta$ time intervals, then the verifiers should also receive the answers in $\Delta$ time interval. Therefore, the total time to run this protocol would be $(\ell-1)\Delta + 2$. For a very small $\Delta$, this is a significant improvement over the sequential protocol that would require $2(\ell -1)$ time to finish. We refer to \Cref{fig:rapid_fire_comp} for the schematic diagram of this protocol. 
\begin{figure}[t]
    \centering
    \includegraphics[scale= 0.45]{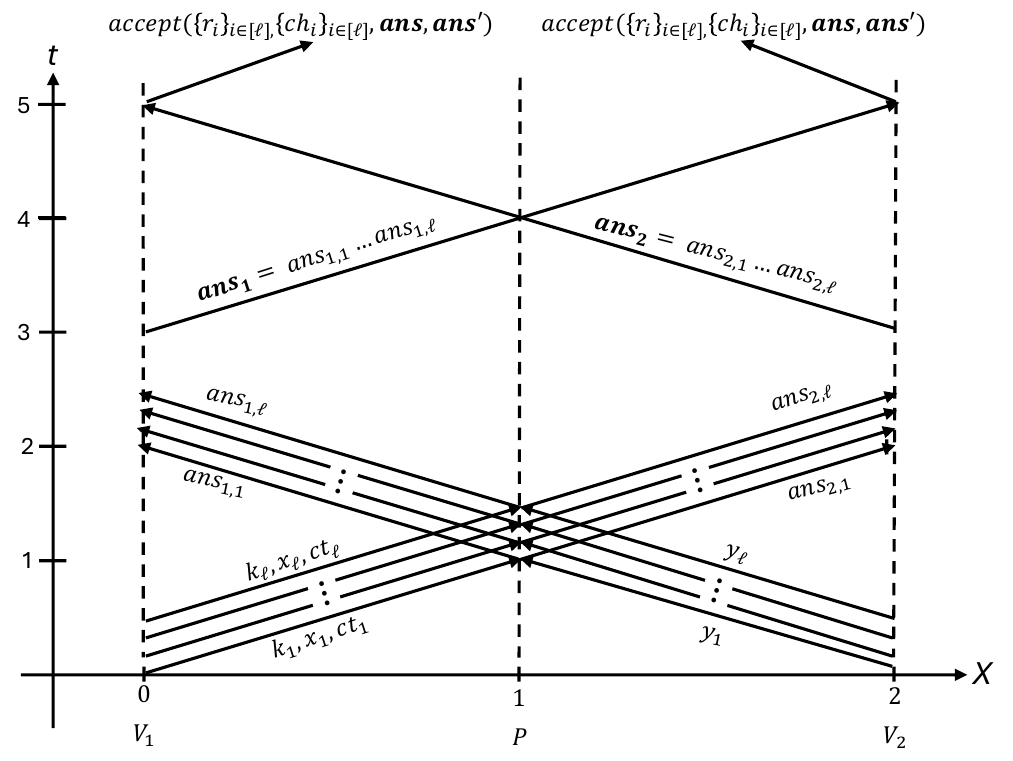}
    \caption{A spacetime diagram describing the interaction between two verifiers at locations $X=0$ and $X=2$ and an honest prover located at the claimed position (for simplicity, located at $X=1$ in the diagram). The verifiers pre-share random values $x_i,y_i,k_i,r_i$ for $i \in [\ell]$ prior to the start of the protocol, and set $s_i = ch_i\oplus G_{k_i}(x_i\oplus y_i)$ for some family of hash functions $\{G_k\}$ and $ch_i = \mathsf{Gen_i}(1^n; r)$ for $\mathsf{Gen_i}$ the verifier functionality of some multi-round certified randomness verifier on round $i$ and where $accept(\cdot)$ determines a boolean function that determines whether the verifiers accept the provers claim based on the transcript received. Verifiers rapidly send each new challenge every $\Delta$ seconds.}
\label{fig:rapid_fire_comp}
\end{figure}
For clarity, we also provide a schematic of the protocol when the verifiers are instead interacting with two malicious provers in \Cref{fig:rapid_fire_dis}.
\begin{figure}[ht]
    \centering
    \includegraphics[scale= 0.45]{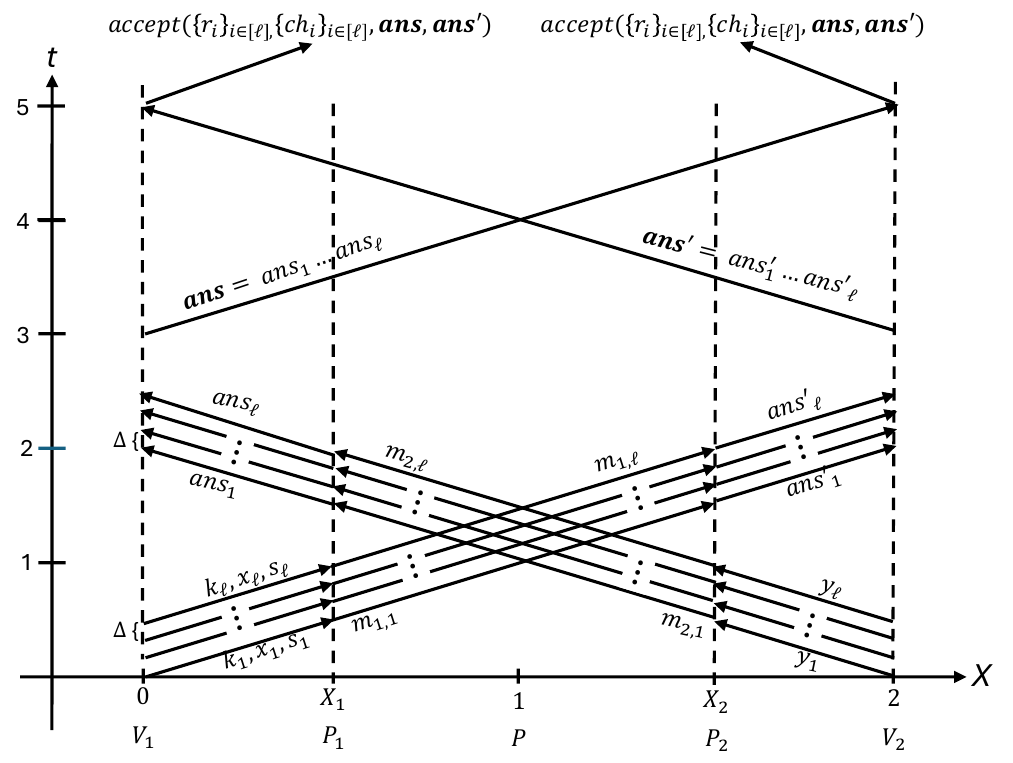}
    \caption{A spacetime diagram describing the interaction between two verifiers at locations $X=0$ and $X=2$ and two dishonest prover located at arbitrary positions between $V_1$ and $V_2$ (for simplicity, located at $X=.5$ and $X=1.5$ in the diagram). The provers are allowed to communicate during the protocol but are non-signaling. Notation and behavior is otherwise analogous to \Cref{fig:rapid_fire_comp}.}
\label{fig:rapid_fire_dis}
\end{figure}

\textbf{Soundness Proof Sketch (\Cref{thm:mult_rapid_comp_inf}).} For the rapid-fire compiler, if we assume that the malicious provers do not communicate the answers to each other during the runtime of the protocol, then we can directly apply the entropy guarantee to upper bound the guessing probability without the sequential decomposition property. However, to satisfy this requirement we need to assume the condition $\Delta l < 2t_{\text{comm}}$ applies to the protocol, where $t_{\text{comm}}$ denotes the communication time between any two malicious provers. 

Note that, due to this condition the rapid-fire compiler can only verify whether the prover is within a range of positions, that satisfy the $2t_{\text{comm}}$ requirement. By reducing $\Delta$, and the number of rounds $l$ one can reduce the $2t_{\text{comm}}$ communication time requirement, but that would introduce additional engineering challenges for the implementation. It can make the protocol less-robust to noise as well. One possible way to increase the robustness of the protocol just by sequential or parallel repetition. In this paper, we have studied the impact of sequential repetition of this compiler, called Sequential Rapid Fire. We refer to \supref{sec:seq_rf} for the details of this construction.

\shortv{
\section*{Acknowledgments}
We thank Dakshita Khurana and Kabir Tomer for pointing out an inaccuracy in our citation of recent work. We thank James Bartusek for valuable discussions regarding the implications of our results.

\section*{Competing Interests}

All authors are co-inventors on a patent application related to this work (no 18/903,860, filed on 1 Oct 2024 by JPMorganChase). The authors declare no other competing interests.

\section*{Disclaimer}
This paper was prepared for informational purposes with contributions from the Global Technology Applied Research center of JPMorgan Chase \& Co. This paper is not a product of the Research Department of JPMorgan Chase \& Co. or its affiliates. Neither JPMorgan Chase \& Co. nor any of its affiliates makes any explicit or implied representation or warranty and none of them accept any liability in connection with this paper, including, without limitation, with respect to the completeness, accuracy, or reliability of the information contained herein and the potential legal, compliance, tax, or accounting effects thereof. This document is not intended as investment research or investment advice, or as a recommendation, offer, or solicitation for the purchase or sale of any security, financial instrument, financial product or service, or to be used in any way for evaluating the merits of participating in any transaction.
}

%% file: Files/prelims.tex
\newcommand{\pv}{\mathscr{V}}

\section{Preliminaries}
\label{sec:prelims}

We now recall a collection of useful definitions and results from the literature. Throughout the paper, we denote the security parameter by $\secparam$. 
\subsection{Quantum Information}

\textbf{Random Oracle Model.}
\noindent In the quantum random oracle model (QROM), all parties $\alice^\hash$ have black-box superposition access to a random function $\hash: \bit^m \to \bit^n$, modeled as a unitary $U^\hash \ket{x}\ket{y} = \ket{x}\ket{y \oplus \hash(x)}$.

\begin{lemma}[\cite{BBBV97}] \label{lem:bbbv}
 Let $\alice$ be an oracle algorithm which makes at most $T$ oracle queries to a function $H: \bit^m \to \bit^n$. Define $\ket{\phi_i}$ as the global state after $\alice$ makes $i$ queries, and $W_y(\ket{\phi_i})$ as the sum of squared amplitudes in $\ket{\phi_i}$ of terms in which $\alice$ queries $H$ on input $y$. Let $\epsilon>0$ and let $F \subseteq \bracC{0,1,\dots,T-1} \times \bit^m$ be a set of time-input pairs such that $\sum_{(i,y) \in F} W_y(\ket{\phi_i}) \le \epsilon^2 / T $.
 \par For $i \in \bracC{0,1,\dots,T-1}$, let $H_i'$ be an oracle obtained by reprogramming $H$ on inputs in $\bracC{y \in \bit^m \; : \; (i,y) \in F}$ to arbitrary outputs. Let $\ket{\phi_T'}$ be the global state after $\alice$ is run with oracle $H_i'$ on the $i$th query (instead of $H$). Then, $\tracedist{\ket{\phi_T}}{\ket{\phi_{T}'}} \le \epsilon / 2$.
\end{lemma}

The following lemma states that it is possible to efficiently simulate a random oracle for a quantum algorithm for which a query-bound is known in advance.
\begin{lemma}[\cite{Zha12}] \label{lem:zhandry_hash}
    A random oracle $H:\cX \to \cY$ is perfectly indistinguishable from a $2q$-wise hash independent hash function $H': \cX \to \cY$ against a quantum algorithm $\alice$ which makes at most $q$ queries.
\end{lemma}

\subsection{Proof of Quantumness (PoQ) and Certified Randomness}

\begin{definition}[PoQ Protocol] \label{def:poq}
    A \emph{proof of quantumness} protocol $\poq = (V,P)$ is an interactive protocol between a classical verifier $V$ and an allegedly quantum prover $P$. Naturally, since $V$ is classical, so is all communication. At the end of the protocol, $V$ either accepts or rejects. $\poq$ is parametrized by a security parameter $\secparam$ and is required to satisfy the following guarantees: \begin{itemize}
        \item {\bf Correctness:} There exists a QPT prover $P$ such that $V$ accepts with overwhelming probability ($1 - \negl(\secparam)$).
        \item {\bf Soundness:} For any PPT prover $P$, the probability that $V$ accepts is at most $\negl(\secparam$).
    \end{itemize}
\end{definition}

\noindent A \emph{single-round} PoQ protocol is a two-message PoQ protocol.

\par \textbf{Certified Randomness.} A property stronger than soundness is certified randomness.
At a high level, certified randomness requires that if a prover $P$ passes the PoQ protocol, then its output cannot be predicted by a guesser $E$ who: \begin{enumerate}
    \item Knows the challenge messages sent by $V$, and
    \item Shares entanglement with $P$, but does not communicate with $P$ during the protocol.
\end{enumerate}
In the most general case, we require this to hold even if $P$ is fully malicious (a.k.a. a \emph{fully general device}). We give the formal single-round definition below. The definitions for the multi-round case are given in \Cref{sec:mr}.

\begin{definition}[Certified Randomness (Single-Round)] \label{def:cr}
    A PoQ protocol $\poq$ is said to have \emph{certified randomness} property if no pair of a QPT prover $P$ and a (possibly entangled) unbounded guesser $Q$ can succeed in the following security game with non-negligible probability: \begin{enumerate}
        \item The verifier $V$ of $\poq$ sends a challenge $\ch$ to both $P$ and $Q$.
        \item $P$ sends back an answer $\ans$ and $Q$ and outputs a guess $\ans'$.
        \item $(P,Q)$ win if $V$ accepts and $\ans = \ans'$.
    \end{enumerate}
\end{definition}

\begin{remark} \label{rem:cr_bell_ineq}
    \Cref{def:cr} captures certified randomness generation using a single device ($P$). The guesser $Q$, who is not allowed to communicate with $P$, appears in the definition simply because it is a game-based definition that is based on the operational meaning of min-entropy. Intuitively, the guarantee is that $P$ generates $\ans$ with high min-entropy conditioned on any quantum side information held by $Q$. This is in contrast with certified randomness generation based on Bell inequality violations \cite{PAMG+10}, which require two devices and the no-communication assumption during \emph{the real execution of the protocol}.
\end{remark}

\subsection{Position Verification with Classical Verifiers}
We focus on position verification in one dimension and give the corresponding definition in the idealized model below. We closely follow the vanilla model of \cite{LLQ22}. Nonetheless, we remark that our work could be generalized to higher dimensions and more robust models.

\begin{definition}[Position Verification with Classical Verifiers] \label{def:cvpv}
A position verification scheme $\pv = (V,X)$ with classical verifiers is an interactive relativistic protocol between a set of classical verifiers $V = (V_i,X_i)_{i \in I}$, where each classical verifier $V_i$ is located at position $X_i \in \R$ on the real line, and $X \in \R$ is the purported location of a quantum prover. During the protocol, we assume that all communication happens at the speed of light (set to be $1$) and all computation is instantaneous. We also assume that $V$ can perform a secure setup before the protocol, and communicate securely during the protocol.
\par We say that $\pv$ is \emph{complete} if there exists an efficient prover $P$ such that if $P$ is located at $X$, then the probability that the verifiers $V$ accept with overwhelming probability.
\par We say that $\pv$ is \emph{sound} if for any collection of efficient, possibly entangled malicious provers $P = (P_j)_{j \in J}$, the probability that $V$ accepts is negligibly small. 
\end{definition}

\begin{remark}[Efficient Verification] \label{rem:effic_ver}
Note that although it is a desirable property, we do not require the verifier(s) to be efficient in \Cref{def:poq,def:cr,def:cvpv}. The reason is that our NISQ-friendly instantiation of CVPV in \Cref{sec:inst} is based on RCS which does not have efficient verification. The rest of the protocols we consider in this work have efficient verification.
\end{remark}

%% file: Files/rcs_rom.tex
\section{CVPV from Single-Round Certified Randomness}
\label{sec:const_rom}

We give a generic construction from a (one-round) PoQ scheme $\poq = (P,V)$ with certified randomness. The basic idea is as follows: the verifiers will send two hash inputs $x,y$ from opposing directions such that they reach the alleged location of the prover at the same time. The challenge $\ch$ of $\poq$ will be computed via evaluating a secure hash function on input $x \oplus y$. This way, the prover needs to receive both $x$ and $y$ before being able to run $P$. A malicious set of provers in our CVPV protocol, intuitively, will be forced into two options: \begin{enumerate}
    \item Try to (at least partially) run $P$ before receiving both $x$ and $y$, and fail the verification of $V$.
    \item Run $P$ with the knowledge of $\ch$ at two locations, and try to get a matching outcome, hence fail due to the certified randomness property.
\end{enumerate}
\noindent Without loss of generality, $V = (\gen, \ver)$ has the following syntax: \begin{enumerate}
    \item It samples random coins $r \from \bit^{\poly(\secparam)}$.
    \item It deterministically generates a challenge $\ch = \gen(1^\secparam;r) \in \bit^n$, and sends it to the prover.
    \item After receiving an answer $\ans$ from the prover, it deterministically verifies by running $\ver(\ch, \ans;r)$.
\end{enumerate}
\begin{construction} \label{constr:cvpv_sr} 
Let $\bracC{\hash_k}_{k \in \bit^\secparam}:\bit^{m} \to \bit^n$ be a cryptographic hash function family, with $m = \omega(\log \secparam)$. We describe the CVPV protocol below:

\begin{enumerate}
    \item At time $t=-\infty$, the verifiers sample random coins $r \from \bit^{\poly(\secparam)}$ for $V$, a hash key $k \from \bit^\secparam$, and random inputs $x,y \from \bit^m$. They publish the hash key $k$, and set $s = \hash_k(x\oplus y) \oplus \ch$.
    \item At $t=0$, $V_0$ sends $(x,s)$ and $V_1$ sends $y$ to the prover simultaneously.
    \item The honest prover, located at position $1$, computes $\ch = \hash_k(x\oplus y) \oplus s$ and $\ans \from P(\ch)$. He immediately sends $\ans$ to both verifiers.
    \item $V_0$ expects $\ans$ at time $t=2$. Similarly, $V_1$ expects $\ans'$ at time $t=2$.
    \item The verifiers accept iff $\ans = \ans'$, and $\ver(\ch, \ans;r)$ accepts. 
\end{enumerate}

\end{construction}

\noindent\textbf{Completeness.} Completeness follows by completeness of $\poq$. \\

\noindent\textbf{Soundness Proof in QROM.}
\begin{theorem}[Single-Round] \label{thm:sr_sec}
    Let $\poq = (V,P)$ be a one-round PoQ scheme that satisfies \Cref{def:cr}. Then, \Cref{constr:cvpv_sr} is a CVPV scheme that is sound in the quantum random oracle model. 
\end{theorem}

\begin{proof}
We will model $\hash_k$ as a classical random oracle $\hash$ with superposition access.
We will create a sequence of hybrids.

\begin{itemize}
    \item {\bf Hybrid 0:} This is the original CVPV soundness experiment.
    \item {\bf Hybrid 1:} In this hybrid, the adversary consists of only two parties: $\alice$ at position $0$ and $\bob$ at position $1$.
    
    \item {\bf Hybrid 2:} In this hybrid, we replace the oracle at times $t < 1$ with the punctured oracle $\hash^\bot$, defined as \begin{align*}
        \hash^\bot(z) = \begin{cases}
            \hash(z), \quad &z \ne x \oplus y \\
            u, \quad &z = x \oplus y
            \end{cases},
    \end{align*}
    where $u \in \bit^n$ is a uniform string.
    \item {\bf Hybrid 3:} In this hybrid, we also replace the oracles accessed by $\alice$ and $\bob$ at time $t=1$ with $\hash^\perp$ defined above. In addition, we give $\ch$ as input to both $\alice$ and $\bob$ at time $t=1$.
    \item {\bf Hybrid 4:} In this hybrid, $\alice$ and $\bob$ each get (only) $\ch$ as input at time $t=0$, but they are not allowed to communicate. Also, they do not get access to the oracle $\hash^\perp$.
\end{itemize}

Let $p_i$ be the optimal success probability of an efficient adversary in {\bf Hybrid $i$}. Let $q = \poly(\secparam)$ be an upper-bound on the total number of oracle queries made by $(\alice, \bob)$. We will show a sequence of claims which suffice for the proof:

\begin{claim}\label{clm:01}
    $p_1 \ge p_0.$
\end{claim}
\begin{proof}
    This step is standard.\footnote{For instance, see \cite{LLQ22}.} One can easily perform a reduction where $\alice$ of {\bf Hybrid 1} can simulate all adversaries in $[0,0.5)$ in {\bf Hybrid 0} and $\bob$ of {\bf Hybrid 1} can simulate all adversaries in $(0.5,1]$ in {\bf Hybrid 0}.
\end{proof}

\begin{claim}\label{clm:12}
    $\abs{p_2-p_1} \le \negl(\secparam)$.
\end{claim}
\begin{proof}
    Suppose the inequality is false for $(\alice,\bob)$, i.e. $\abs{p_2-p_1} < \eps$ for a non-negligible function $\eps(\secparam)$. Then, by \Cref{lem:bbbv}, in {\bf Hybrid 2} the query weight on $\hash(x,y)$ by $\alice$ (the case of $\bob$ being similar) at time $t < 1$ is lower-bounded by $2\eps^2/q$. Consider the following extractor $\alice'(\alice)$: \begin{itemize}
        \item $\alice'$ receives $(x,s)$ from the challenger and the $\reg A$ register of the initial state $\ket{\psi}_\reg{AB}$ for $(\alice, \bob)$. Then $\alice'$ samples $i \from [q]$ and runs $\alice$ on input $(x,s, \reg A)$, measuring the input register of the $i$-th query made by $\alice$ to $\hash$ as $z^*$. She outputs $y^* = z^* \oplus x$.
    \end{itemize}
    Now, the probability that $\alice'$ outputs $y$ is at least $2\eps^2/q^2$ due to no-signalling, which is a contradiction since $\alice'$ has no information about $y$ and $2\eps^2/q^2 > 2^{-m}$.
\end{proof}

\begin{claim}\label{clm:23}
    $p_3 \ge p_2.$
\end{claim}
\begin{proof}
    Follows by a simple reduction $(\alice',\bob')$, which simulates $(\alice,\bob)$ in {\bf Hybrid 2}. $\alice'$ forwards $(x,s)$ and $\bob'$ forwards $y$. Furthermore, they use $(\ch,x,y,s)$ to reprogram the oracle $\hash^\perp$ in order to simulate the oracle $\hash$. 
\end{proof}

\begin{claim}\label{clm:34}
    $p_4 \ge p_3.$
\end{claim}
\begin{proof}
    We give a reduction $(\alice',\bob')$ from {\bf Hybrid 3} to {\bf Hybrid 4}: \begin{itemize}
        \item Let $(\alice, \bob)$ be an adversary for {\bf Hybrid 3} that succeeds with probability $p_3$.
        \item At time $t=-\infty$, $\alice'$ and $\bob'$ prepare the bipartite state $\ket{\psi}_\reg{AB}$ shared between $\alice$ and $\bob$. In addition, they sample a $2q$-wise independent hash function $\hash'$ as well as $(x,y,s) \from \bit^m \times \bit^m \times \bit^n$.
        \item At time $t=0$, $\alice'(t=0)$ runs $\alice$ on input $(x,s,\reg{A})$, using $\hash'$ as the oracle. At time $t=1$, $\alice'$ receives $\ch$ from the verifier and runs $\alice(t=1)$ with $\ch$ as additional input.
        \item At time $t=0$, $\bob'(t=0)$ runs $\bob$ on input $(y,\reg{B})$, using $\hash'$ as the oracle. At time $t=1$, $\bob'$ receives $\ch$ from the verifier and runs $\bob(t=1)$ with $\ch$ as additional input.
    \end{itemize}
    Observe that since the oracle $\hash^\perp$ in {\bf Hybrid 3} is independent of $(x,y,s,\ch)$, and by \Cref{lem:zhandry_hash}, the view of $(\alice,\bob)$ is perfectly simulated by the reduction.
\end{proof}

\begin{claim}\label{clm:4}
    $p_4 \le \negl(\secparam).$
\end{claim}

\begin{proof}
    Suppose $p_4$ is not negligible for some $(\alice, \bob)$. We will break the certified randomness (\Cref{def:cr}) of $\poq$: \begin{itemize}
        \item $P$ holds the $\reg{A}$ register of $\ket{\psi}_\reg{AB}$ prepared by $(\alice, \bob)$ at time $t<0$. After receiving $\ch$ from the verifier, $P$ runs $\alice(t \ge 0)$ and outputs $\ans$ which is sent to the verifier.
        \item The guesser $Q$ holds register $\reg{B}$. She receives $\ch$ and runs $\bob(t\ge 0)$ to output $\ans'$.
    \end{itemize}

    With probability $p_4$, $\ans = \ans'$ and $\ans$ is accepted by the verifier, which violates certified randomness (\Cref{def:cr}).
\end{proof}
\end{proof}

\begin{remark}
    Note that soundness still holds if the adversary $(\alice, \bob)$ get access to $\hash$ at time $t=-\infty$. This means we can heuristically instantiate $\hash$ using an unkeyed public hash function such as SHA-512.
\end{remark}

\section{CVPV from Multi-Round Certified Randomness}
\label{sec:mr}

If the underlying PoQ-based certified randomness protocol uses more than one round, we can naturally generalize our compiler just by composing our single-round compiler sequentially. We first give a natural way to do this in \Cref{sec:mr_seq}. This requires a protocol satisfying the stronger Definition \ref{def:cr_multiround_simult} of certified randomness with sequential decomposition, but we show that any protocol satisfying the definition of certified randomness with single round entropy must satisfy the sequential decomposition property\footnote{In fact, we show that \Cref{def:cr_multiround} is satisfied which is stronger than \Cref{def:cr_multiround_simult}.} in \Cref{sec:seq_crea}, lending its applicability to a wide range of possible instantiations. We then give a more clever way in \Cref{sec:mr_rf} and \Cref{sec:seq_rf} which is superior in the idealized model at the cost of practical robustness.

\subsection{Multi-Round Certified Randomness: Definitions} \label{sec:cr_multiround_def}

We give three definitions of multi-round certified randomness, from strongest to weakest. While all definitions generalize \Cref{def:cr}, they differ in terms of how much communication is allowed between the prover and the guesser in-between rounds: it is unrestricted (\Cref{def:cr_multiround}), restricted (\Cref{def:cr_multiround_simult}), or forbidden (\Cref{def:cr_multiround_nocomm}).

\begin{definition}[Multi-Round Certified Randomness Protocol with Sequential Decomposition Property] \label{def:cr_multiround}
    An $\ell$-round PoQ protocol $\poq$ is said to have \emph{certified randomness with sequential decomposition} property if no pair of a QPT prover $P$ and a (possibly entangled) unbounded guesser $\eve$ can succeed in the following security game with non-negligible probability: \begin{itemize}
        \item For $i \in [\ell]$, the following steps occur in order: \begin{enumerate}
            \item The verifier $V$ of $\poq$ sends a challenge $\ch_i$ to both $P$ and $\eve$.
            \item $P$ sends back an answer $\ans_i$.
            \item $\eve$ outputs a guess $\ans_i'$.
            \item $P$ and $\eve$ can communicate freely and setup again.
        \end{enumerate}
        \item $(P,\eve)$ win the game if $V$ accepts and $\ans_i = \ans_i'$ for all $i \in [\ell]$.
    \end{itemize}
   
\end{definition}
\ \\
\noindent We consider weaker variants of \Cref{def:cr_multiround} by restricting the communication between $P$ and $\eve$ in-between rounds, with the differences being \highlight{highlighted}.

\par Below, we say that two spatially separated parties $(\alice,\bob)$ perform \emph{one round of simultaneous communication} at time $t$ if $\alice$ sends one (classical or quantum) message to $\bob$ at time $t$ and vice versa.

\begin{definition}[Multi-Round Sequential Certified Randomness with Sequential Decomposition and Restricted Communication] \label{def:cr_multiround_simult}
    An $\ell$-round PoQ protocol $\poq$ is said to have \emph{sequential certified randomness with sequential decomposition and restricted communication} property if no pair of a QPT prover $P$ and a (possibly entangled) unbounded guesser $\eve$ can succeed in the following security game with non-negligible probability: \begin{itemize}
        \item For $i \in [\ell]$, the following steps occur in order: \begin{enumerate}
            \item The verifier $V$ of $\poq$ sends a challenge $\ch_i$ to both $P$ and $\eve$.
            \item $P$ sends back an answer $\ans_i$.
            \item $\eve$ outputs a guess $\ans_i'$.
            \item $P$ and $\eve$ can \highlight{perform simultaneous single-round communication}.
        \end{enumerate}
        \item $(P,\eve)$ win the game if $V$ accepts and $\ans_i = \ans_i'$ for all $i \in [\ell]$.
    \end{itemize}
\end{definition}

\begin{definition}[Multi-Round Sequential Certified Randomness with Sequential Decomposition and No Communication] \label{def:cr_multiround_nocomm}
    An $\ell$-round PoQ protocol $\poq$ is said to have \emph{sequential certified randomness with sequential decomposition and no communication} property if no pair of a QPT prover $P$ and a (possibly entangled) unbounded guesser $\eve$ can succeed in the following security game with non-negligible probability: \begin{itemize}
        \item For $i \in [\ell]$, the following steps occur in order: \begin{enumerate}
            \item The verifier $V$ of $\poq$ sends a challenge $\ch_i$ to both $P$ and $\eve$.
            \item $P$ sends back an answer $\ans_i$.
            \item $\eve$ outputs a guess $\ans_i'$.
            \item $P$ and $\eve$ \highlight{cannot communicate}.
        \end{enumerate}
        \item $(P,\eve)$ win the game if $V$ accepts and $\ans_i = \ans_i'$ for all $i \in [\ell]$.
    \end{itemize}
\end{definition}

\begin{remark}
    For single-round $\poq$, \Cref{def:cr,def:crea,def:cr_multiround,def:cr_multiround_simult,def:cr_multiround_nocomm} all coincide.
\end{remark}

\input{Files/sequential}

\input{Files/rom_crea}

\input{Files/rapidfire}

\input{Files/sequential-rapid-firing}

%% file: Files/sequential.tex
\subsection{Sequential Compiler} \label{sec:mr_seq}
 As before, we give a generic construction from a (multi-round) PoQ scheme $\poq = (V,P)$ which has certified randomness with sequential decomposition (\Cref{def:cr_multiround_simult}).

\begin{construction}\label{constr:cvpv_mr_seq}
 Let $\ell = \poly(\secparam)$ be the number of rounds in $\poq$. Without loss of generality, $V = (\gen_1,\dots,\gen_\ell, \ver)$ has the following syntax: \begin{enumerate}
    \item It samples random coins $r \from \bit^{\poly(\secparam)}$. 
    \item For $i=1,\dots,\ell$: \begin{itemize}
        \item It deterministically generates a challenge\\ $\ch_i = \gen_i(1^\secparam, \ans_1,\dots,\ans_{i-1};r) \in \bit^n$.
        \item It receives an answer $\ans_i$ from $P$.
    \end{itemize}
    \item It deterministically verifies by running $\ver(\ch_1, \ans_1, \dots, \ch_\ell, \ans_\ell;r)$.
\end{enumerate} 

\noindent Similarly, $P = (P_1,\dots,P_\ell)$ has the following syntax: For $i=1,\dots,\ell$, after receiving the $i$-th challenge $\ch_i$, $P$ computes $$\ans_i \from P_i(\ch_1,\ans_1,\dots,\ch_{i-1},\ans_{i-1},\ch_i)$$ and responds with $\ans_i$.

\par Let $\bracC{\hash_k}_{k \in \bit^\secparam}:\bit^{m} \to \bit^n$ be a cryptographic hash function family, with $m = \omega(\log \secparam)$. We describe the (multi-round) CVPV protocol below:
\begin{enumerate}
    \item At time $t=-\infty$, the verifiers sample random coins $r \from \bit^{\poly(\secparam)}$ for $V$ and a hash key $k \from \bit^\secparam$. For $i = 1,\dots,\ell$, they sample random inputs $x_i,y_i \from \bit^m$. They publish the hash key $k$.
    \item For $i=1,\dots,\ell$: \begin{itemize}
        \item At time $t=i-1$, $V_0$ computes $\ch_i = \gen_i(1^\secparam,\ans_1,\dots,\ans_{i-1};r)$ and $s_i = \hash_k(x_i \oplus y_i) \oplus \ch_i$. It sends $(x_i,s_i)$ and expects an answer $\ans_i$ at time $t=i$.
        \item Similarly, at time $t=i-1$, $V_1$. It sends $y_i$ and expects an answer $\ans'_i$ at time $t=i$. 
        
        \item At time $t=i-1/2$, the honest prover, located at position $0.5$, computes $\ch_i = \hash_k(x_i\oplus y_i) \oplus s_i$ and $\ans_i \from P_i(\ch_1,\dots,\ch_i,\ans_1,\dots,\ans_{i-1})$. It immediately sends $\ans_i$ to both verifiers.
    \end{itemize}
    
    \item The verifiers accept iff $\ans_i = \ans_i'$ for all $i$, and $\ver(\ch_1, \ans_1, \dots, \ch_\ell, \ans_\ell;r)$ accepts.
\end{enumerate}

\end{construction}

\noindent\textbf{Security Proof in QROM.}
\begin{theorem}[Sequential Compiler] \label{thm:mr_seq_sec}
    Let $\poq = (V,P)$ be a PoQ scheme that satisfies \Cref{def:cr_multiround_simult}. Then, \Cref{constr:cvpv_mr_seq} is a secure CVPV scheme. 
\end{theorem}

\begin{proof}
We will model $\hash_k$, where $k \from \bit^\secparam$, as a random oracle $\hash$.
We give a sequence of hybrid experiments below:
\begin{itemize}
    \item {\bf Hybrid 0:} This is the original CVPV soundness experiment.
    \item {\bf Hybrid 1:} In this hybrid, the adversary consists of only two parties: $\alice$ at position $0$ and $\bob$ at position $1$. W.l.o.g., $(\alice, \bob)$ perform a round of simultaneous communication at times $t=0,1,\dots,\ell-1$.
    \item {\bf Hybrid 2:} In this hybrid, we additionally give $\alice$ and $\bob$ $\ch_i$ at time $t=i$ for $i=1,\dots,\ell$.
    \item {\bf Hybrid 2.1-$\ell$:} We set {\bf Hybrid 2.0} to be {\bf Hybrid 2} and $\hash_0 := \hash$. For $i \in [\ell]$, we define {\bf Hybrid 2.$i$} to be the same as {\bf Hybrid 2.$(i-1)$}, except the oracle $G_{i-1}$ is replaced by the reprogrammed oracle $G_{i}$, where \begin{align*}
        \hash_i(z) = \begin{cases}
            \hash_{i-1}(z), \quad &z \ne x_i \oplus y_i \\
            u_i, \quad &z = x_i \oplus y_i
            \end{cases},
    \end{align*}
    with $u_i \from \bit^n$ being a fresh random string.
    \item {\bf Hybrid 3:} In this hybrid, $\alice$ and $\bob$ only receive $\ch_i$ at time $t=i$, for $i \in [\ell]$, and no other input. They do not get access to the oracle $\hash_\ell$ either.
\end{itemize}

Let $p_i$ be the optimal success probability of an efficient adversary in {\bf Hybrid $i$}. Let $q = \poly(\secparam)$ be an upper-bound on the total number of oracle queries made by $(\alice, \bob)$. We will show a sequence of claims which suffice for the proof:

\begin{claim} \label{clm:mr01}
    $p_1 \ge p_0$.
\end{claim}

\begin{proof}
    Follows by a simple generalization of the corresponding claim in the proof of \Cref{thm:sr_sec}. 
\end{proof}

\begin{claim} \label{clm:mr12}
    $p_2 \ge p_1$.
\end{claim}
\begin{proof}
    Since we give extra information to the adversary, the success probability cannot decrease.
\end{proof}

\begin{claim} \label{clm:mr22}
    Setting $p_{2.0} := p_2$, $p_{2.i} \ge p_{2.(i-1)} - \negl(\secparam)$ for $i \in [\ell]$.
\end{claim}

\begin{proof}
    Let $i \in [\ell]$ and $(\alice, \bob)$ be an adversary that succeeds in {\bf Hybrid 2.$(i-1)$} with probability $p_{2.(i-1)}$. We will give a reduction $(\alice', \bob')$ for {\bf Hybrid 2.$i$}: \begin{itemize}
        \item At times $t<i$, $\alice'$ (resp. $\bob'$) runs $\alice$ (resp. $\bob$) using $G_{i}$ as the oracle.
        \item At time $t=i-1$, $\alice'$ sends $(x_i,s_i)$ and $\bob'$ sends $y_i$ to each other, so that the messages are received at $t=i$.
        \item At times $t \ge i$, $\alice'$ and $\bob'$ can simulate $G_{i-1}$ using $(x_i,y_i,s_i,\ch_i,G_i)$ by reprogramming $G_i$ to output $\ch_i\oplus s_i$ on input $x_i \oplus y_i$.
        \item $\alice'$ (resp. $\bob'$) outputs what $\alice$ (resp. $\bob$) outputs.
    \end{itemize}
    Note that the view of $(\alice,\bob)$ as simulated by $(\alice',\bob')$ differs from {\bf Hybrid 2.$(i-1)$} at times $t<i$, and only on input $x_i \oplus y_i$ to the oracle. Therefore, if the probability that $(\alice',\bob')$ succeeds is upper-bounded by $p_{2.(i-1)} - \eps$ for some non-negligible function $\eps$, then by \Cref{lem:bbbv} the total query weight by $(\alice,\bob)$ on input $x_i \oplus y_i$ at times $t<i$ must be at least $2\eps^2/q$. Suppose the query weight by $\alice$ is at least $\eps^2/q$, for the other case is similar. We give an extractor $(\alicetild, \bobtild)$ in {\bf Hybrid 2.$(i-1)$}: \begin{itemize}
        \item $\alicetild$ samples $j \from [q]$ and simulates $\alice$, stopping the execution at the $j$-th query made by $\alice$ to the oracle $\hash_{i-1}$, measuring the query as $z^*$. She outputs $y^* = z^* \oplus x_i$, where $x_i$ is received at time $t=i-1$ from $V_0$.
        \item $\bobtild$ simulates $\bob$.
    \end{itemize}
    By assumption, $y^*=y_i$ with probability $\eps^2/q^2 > 2^{-m}$, which is a contradiction since $y_i$ is information theoretically hidden from $\alicetild$ at times $t<i$.
\end{proof}

\begin{claim} \label{clm:mr23}
    $p_3 \ge p_{2.\ell}$.
\end{claim}

\begin{proof}
    Let $(\alice,\bob)$ be an adversary for {\bf Hybrid 2.$\ell$} that succeeds with probability $p_{2.\ell}$. We give a reduction $(\alice', \bob')$ that succeeds in {\bf Hybrid 3} with the same probability: \begin{itemize}
        \item At time $t=-\infty$, $\alice'$ and $\bob'$ sample a $2q$-wise independent hash function $\hash'$. In addition, they sample $(x_i,y_i,s_i) \from \bit^m \times \bit^m \times \bit^n$ for $i \in [\ell]$.
        \item $\alice'$ simulates $\alice$ using $\hash'$ as the oracle, the sampled values $(x_i,y_i)$, as well as the values $\ch_i$ received from $V_0$.
        \item $\bob'$ similarly simulates $\bob$ using $\hash'$ as the oracle, the sampled values $y_i$, as well as the values $\ch_i$ received from $V_1$.
    \end{itemize}
    The view of $(\alice,\bob)$ is perfectly simulated since the oracle $G_\ell$ in {\bf Hybrid 2.$\ell$} is independent of the values $(x_i,y_i,s_i)$ for all $i \in [\ell]$. This is because the oracle has been reprogrammed on all inputs $x_i \oplus y_i$ to remove any such dependence. Thus, by \Cref{lem:zhandry_hash}, $\hash'$ simulates an independent random oracle and the proof is complete.
\end{proof}

\begin{claim} \label{clm:mr3}
    $p_3 \le \negl(\secparam).$
\end{claim}

\begin{proof}
    Follows directly from \Cref{def:cr_multiround_simult}. $\alice$ plays the role of the prover and $\bob$ that of the guesser.
\end{proof}
\end{proof}

%% file: Files/rom_crea.tex
\subsection{Sequential Decomposition from Repetition}
\label{sec:mr_rep}
This section uses the same notation as \Cref{sec:seq_crea}. Although other multi-round approaches are possible, the most general sequential compilation approach is desirable due to practical robustness considerations. Section \ref{sec:mr_seq} proves that one can construct CVPV protocols from certified randomness protocols with the sequential decomposition property. However, the usual notion of certified randomness (\Cref{def:cr}) is defined in terms of entropy lower bounds, but we cannot lower bound the CVPV protocol entropy conditioned on the guesser's side information since she can communicate with the prover.
Therefore, it is not immediately obvious whether a certified randomness protocol defined with entropy is a sequential certified randomness protocol (\Cref{def:cr_multiround}) defined by the guessing probability. We show that for a large class of multi-round certified randomness protocols, namely those that rely on repetition of single rounds with entropy bounds `on average', all have the sequential decomposition property.

\begin{definition}[Certified Randomness from Repetition] \label{def:crea}
    A PoQ protocol $\poq$ is said to be \emph{certified randomness from repetition} if
    \begin{itemize}
        \item The verifier $V$ of $\poq$ samples $\ell$ challenges $\{\ch_1,\dots,\ch_\ell\}$ in an i.i.d. manner.
        \item For $i \in [\ell]$, the following steps occur in order: \begin{enumerate}
            \item The verifier $V$ of $\poq$ sends the challenge $\ch_i$ to $P$.
            \item $P$ sends back an answer $\ans_i$. 
            \item $V$ computes a classical output $X_i$ from $\ch_i$ and $\ans_i$.
        \end{enumerate}
    \item $V$ accepts if $X_1^n\equiv X_1\cdots X_n=x_1^n\in\omega'$, where $x_1^n$ is the classical value of register $X_1^n$ and $\omega'$ is a set of acceptable values of $x_1^n$ that satisfy the certified randomness test condition.
    \item For any QPT prover $P$ and verifier $V$ described by the quantum channel $\mathcal{P}_{i}:R_{i-1}\rightarrow X_i A_i C_i R_i, \mathcal{P}_{n}\circ\cdots\circ\mathcal{P}_1$,
    we have
    \begin{equation}
        \inf_{\nu\in\Sigma_i(q)}H(A_i\vert C_i E)_{\nu}\geq f(q), ~~\text{    where}
    \end{equation}
\vspace{-0.1in}
    \begin{equation}
        \Sigma_i(q)=\left\{\nu_{X_i A_i C_i R_i E}=\mathcal{P}_i(\rho)\Big\vert \rho\in S(R_{i-1} E)\wedge\nu_{X_i}=q\right\},
    \end{equation}
    $R$ is the quantum memory, $A$ is the output register for $\ans$, $C$ is the challenge register for $\ch$, $E$ is the quantum side information register, $S(R_{i-1}E)$ is the set of all quantum states on $R_{i-1}E$, $q\in\mathds{P}$, $\mathds{P}$ is the set of density operators corresponding to classical probability distributions on the alphabet $\mathcal{X}$ of $X_i$, and $f$ is an affine function.
    \item For \vspace{-0.2in}
    \begin{align}
    h&=\min_{x_1^n\in\omega'}f\left(\mathsf{freq}(x_1^n)\right)\\
    \mathsf{freq}(x_1^n)(x)&=\frac{\vert \{i\in\{1,\dots,n\}:x_i=x\}\vert}{n},
    \end{align}
    we have $h>0$.
    \end{itemize}

\end{definition}

Indeed, in the situation where we only care about the side information in the challenges or the initial quantum  register $E$ that does not evolve, the Markov chain condition is trivially satisfied since the challenges are generated in an i.i.d. manner. In such a case, the quantum channel of the protocol is illustrated in Fig. \ref{fig:protocol_channels}(a). Therefore, we can apply the entropy accumulation theorem \cite{DFR20} to a certified randomness from repetition protocol to lower bound the smooth min-entropy. Conversely, any protocol that can be proven sound using the entropy accumulation theorem must satisfy Definition \ref{def:crea}. Similarly, for protocols that are secure under a more general adversary model where the environment may be updated each round and use the generalized entropy accumulation theorem of \cite{MFSR22} to prove soundness, they must be secure under the more restricted model where the environment cannot be updated. They must satisfy Definition \ref{def:crea} since generalized entropy accumulation has a stronger single-round entropy requirement.

We note that many existing protocols reuse challenges due to the need for randomness expansion. However, randomness expansion does not concern CVPV, and generating challenges for every round makes the analysis simpler. 

To show that protocols satisfying Definition \ref{def:crea} can be used to construct CVPV protocols (illustrated in Fig. \ref{fig:protocol_channels}(b)) with the sequential decomposition property, we first consider an optimal adversary with the largest overall acceptance probability $\Pr[\Omega]$, where $\Omega$ denotes the event where the answers pass the protocol statistical test and all the guesses are correct. We denote the quantum channel of the $n$-round optimal adversary as $\mathcal{M}_n^{'*}\circ\cdots\circ\mathcal{M}_1^{'*}$. We then construct a modified adversary $\bar{\mathcal{M}}_n^{'}\circ\cdots\circ\bar{\mathcal{M}}_1^{'}$ from this optimal adversary, and show that the modified adversary has the same $\Pr[\Omega]$. Then, we show the modified adversary also satisfies the non-signalling condition required by the generalized entropy accumulation theorem \cite{MFSR22}. This allows us to lower bound the smooth min-entropy conditioned on the guesses and other side information of the modified adversary using the generalized entropy accumulation theorem. Finally, we derive an upper bound on $\Pr[\Omega]$ by using the fact that either the test fails with high probability, or the prover output has high entropy relative to the guesses and the guesser must fail. A rigorous proof is given in \Cref{sec:seq_crea}.

Consider an intermediate quantum-classical state for the optimal adversary:
\begin{align}\nonumber
    \sum_{a_1^i c_1^i g_1^i}p^*(a_1^i c_1^i g_1^i)\vert a_1^i c_1^i g_1^i\rangle\langle a_1^i c_1^i g_1^i\vert_{A_1^i C_1^i G_1^i}\otimes \rho_{R_i R_i'}^{*a_1^i c_1^i g_1^i}\label{eqn:intermediate_memory},
\end{align}
where $R_i,R'_i$ are prover and guesser quantum memory registers after the $i$th round, $A_1^i C_1^i G_1^i$ are classical registers of the answers, challenges, and guessers for the first $i$ rounds, lowercase variables are classical values of the respective registers, $p^*$ is the probability of the classical outcome, and $\rho^*$ is the quantum state on the quantum memory corresponding to the classical outcome. We define the modified adversary quantum channel $\bar{\mathcal{M}}_n'\circ\cdots\circ \bar{\mathcal{M}}_1'$ as
\begin{equation}
    \bar{\mathcal{M}}_{i+1}'=\Gamma_{i+1}\circ\mathrm{tr}_{R_{i+1} R_{i+1}'}\circ\mathcal{M}_{i+1}^{'*}\circ\Gamma_i\circ\mathrm{tr}_{R_i R_i'}\label{eqn:modified_channel_overview},
\end{equation}
where $\Gamma_i:C_1^i G_1^i\rightarrow R_i R'_i C_1^i G_1^i$ is a quantum channel give by
\begin{equation}
    \Gamma_i(\rho)=\sum_{c_1^i g_1^i}\Pi_{C_1^i G_1^i}^{c_1^i g_1^i}~\rho~\Pi_{C_1^i G_1^i}^{c_1^i g_1^i}\otimes \vert c_1^i g_1^i\rangle\langle c_1^i g_1^i\vert\otimes\rho_{R_i R_i'}^{* g_1^i c_1^i g_1^i},\label{eqn:gamma_definition-main}
\end{equation}
where $\Pi$ is the projector onto classical values.

Intuitively, the modified quantum channel first throws the quantum memory away with $\mathrm{tr}_{R_i R_i'}$, and then replaces the memory with another memory state. An optimal adversary may use the quantum memory depends on the classical history $a_1^i c_1^i g_1^i$.
Fortunately, this is not an issue. If the guess is incorrect ($g_1^i\neq a_1^i$), then the protocol has already aborted and it does not matter what quantum memory is supplied to the next round. If the guess is correct ($g_1^i= a_1^i$), then the supplied quantum memory supplied by the modified adversary is the same as that of an optimal adversary ($\rho_{R_i R_i'}^{* g_1^i c_1^i g_1^i}=\rho_{R_i R_i'}^{* a_1^i c_1^i g_1^i}$).
This concludes the proof that the modified adversary defined in E.q. \ref{eqn:modified_channel_overview} has the maximum success probability.

Further, since the prover output $a_i$ is discarded (here, single round classical output register is denoted as $A_i$ and the value is denoted as $a_i$, and the single round challenges and guesses are similarly denoted as $C_i G_i$ and $c_i g_i$), the output side information must be independent on $a_i$. The quantum channel defined in E.q. \ref{eqn:modified_channel_overview} satisfies the non-signalling condition. Namely, for all $i$, there exists $\mathcal{R}_{i+1}:E_i\rightarrow E_{i+1}$ such that
\begin{equation}
\mathrm{tr}_{A_{i+1}R_{i+1}}\circ \bar{\mathcal{M}}_{i+1}'=\mathcal{R}_{i+1}\circ\mathrm{tr}_{R_i},
\end{equation}
where $E_i=C_1^i G_1^i R'_i$. This is apparent from Fig. \ref{fig:schematic}(c).

However, we only have single round entropy for certified randomness shown in Fig. \ref{fig:schematic}(a). Namely, $H(A_i\vert C_i E)\geq f(q)$ for arbitrary input quantum states over $R_{i-1}E$, where $f(q)$ is a function on some test outcome. Instead, in order to use entropy accumulation for the CVPV protocol, for the channel illustrated in Fig. \ref{fig:schematic}(c), we need $H(A_i\vert E_i \Tilde{E})\geq f(q)$ for arbitrary input quantum states over $R_{i-1}E_{i-1}\Tilde{E}$.

Since processing with $\mathcal{G}_i$ cannot decrease entropy, $H(A_i\vert C_i E)\geq f(q)$ for arbitrary states over $R_{i-1}E$ and arbitrary register $E$ implies $H(A_i\vert C_1^{i-1}G_1^{i-1}R'_{i-1})\geq f(q)$ and similarly $H(A_i\vert E_i)\geq f(q)$ before $\mathcal{N}_i$ is applied. Further, since $\mathcal{N}_i$ does not change $C_1^i G_1^i$, the entropy $H(A_i\vert C_1^i G_1^i)$ does not change before or after $\mathcal{N}_i$ and $H(A_i\vert C_1^i G_1^i)\geq H(A_i\vert E_i)\geq f(q)$ at the end of channel $\mathcal{M}_i^{'*}$.

For channel $\bar{\mathcal{M}}_i'$, since $\Gamma_i\circ\mathrm{tr}_{R_i R_i'}$ outputs a state that only depends on $C_1^i G_1^i$, we must have $H(A_i\vert E_i)=H(A_i\vert C_1^i G_1^i)\geq f(q)$. Finally, for conditioning on the arbitrarily entangled register $\Tilde{E}$, we recognize that the input state to $\mathcal{M}_i^{'*}$ after $\Gamma_{i-1}\circ\mathrm{tr}_{R_{i-1} R_{i-1}'}$ is of the form
\begin{equation}
    \sum_{c_1^{i-1} g_1^{i-1}}p(c_1^{i-1} g_1^{i-1})\vert c_1^{i-1} g_1^{i-1}\rangle\langle c_1^{i-1} g_1^{i-1}\vert\otimes\rho_{R_{i-1}R_{i-1}'}^{c_1^{i-1} g_1^{i-1}}\otimes\sigma_{\Tilde{E}}^{c_1^{i-1}g_1^{i-1}},
\end{equation}
where $\rho_{R_{i-1}R_{i-1}'}^{c_1^{i-1} g_1^{i-1}},\sigma_{\Tilde{E}}^{c_1^{i-1}g_1^{i-1}}$ are some density operators over $R_{i-1}R_{i-1}'$ and $\Tilde{E}$ that depend on $c_1^{i-1} g_1^{i-1}$. This is to say that the input state to $\mathcal{M}_i^{'*}$, and therefore the output $A_i$, is only correlated with $\Tilde{E}$ through classical variables $c_1^{i-1} g_1^{i-1}$. This means conditioning on $\Tilde{E}$ cannot reduce entropy when we are already conditioning on $C_1^i G_1^i$. Hence $H(A_i\vert E_i \Tilde{E})\geq f(q)$.

For an $n$-round protocol, entropy accumulation implies $H_{\rm min}^\varepsilon=O(n)$, and we can similarly choose the smoothing parameter $\varepsilon=O(2^{-n})$. Combined, we show that $\Pr[\Omega]=O(2^{-n})$, which proves the soundness of the CVPV protocol. The completeness of CVPV follows from the certified randomness protocol.

\begin{figure}[t]
    \centering
    \includegraphics[width=\textwidth]{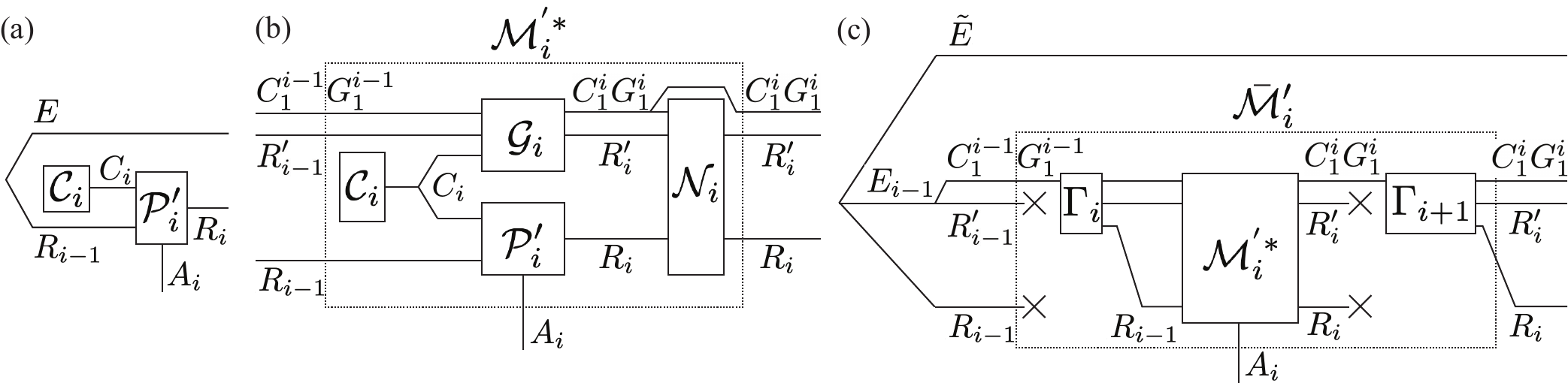}
    \caption{(a) Original certified randomness protocol quantum channel with one prover. (b) CVPV protocol quantum channel with one prover and one guesser. (c) Modified CVPV protocol quantum channel, where crosses represent tracing over the register. For all subpanels, $\mathcal{P}'$ is the prover channel, $\mathcal{G}$ is the guesser channel, $\mathcal{N}$ is the communication channel, and $\mathcal{C}$ is the challenge generation channel. We also use the notation $E_i=C_1^i G_1^i R_i'$.}
    \label{fig:schematic}
\end{figure}

\begin{figure}[t]
    \centering
    \includegraphics[width=0.6\textwidth]{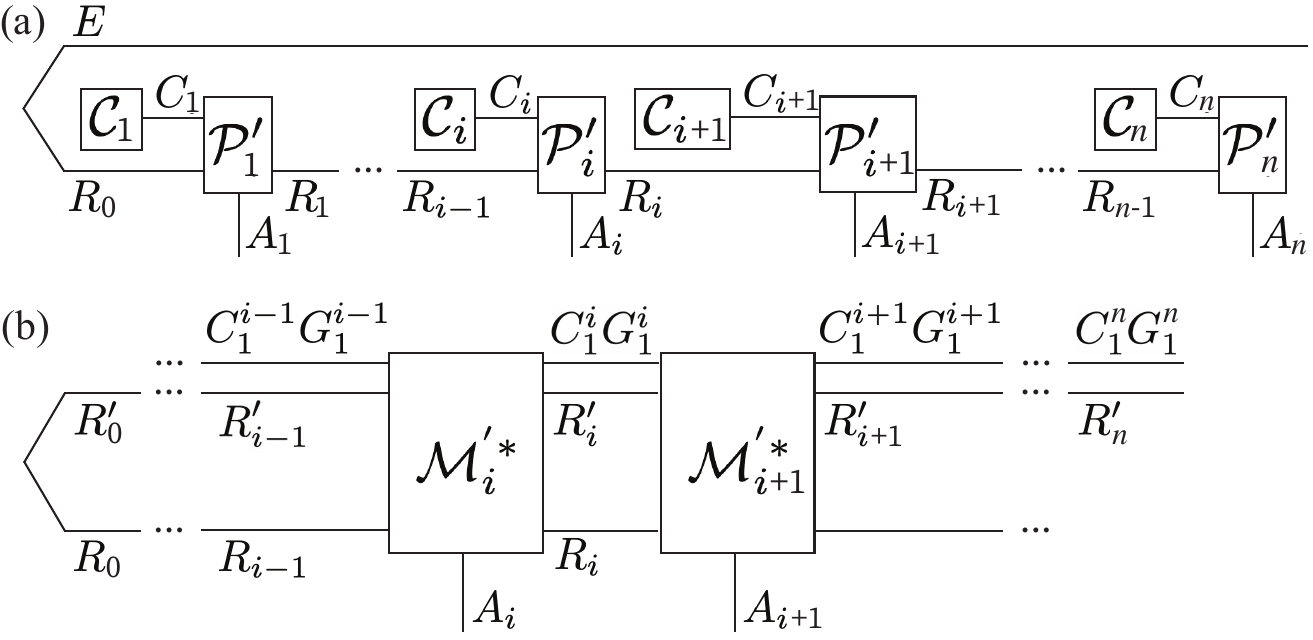}
    \caption{(a) Quantum channel of a certified randomness from repetition protocol. (b) Quantum channel of a CVPV protocol built from (a). Notations are the same as Fig. \ref{fig:schematic}.}
    \label{fig:protocol_channels}
\end{figure}

%% file: Files/rapidfire.tex
\newcommand{\gentild}{\widetilde{\mathsf{Gen}}}

\subsection{Rapid-Firing} \label{sec:mr_rf}

Another way to construct CVPV from multi-round PoQ with CR is \emph{rapid-firing}. The idea is to send messages back-to-back in intervals much smaller than the round-trip communication time. Note that this construction only works if the challenges are chosen non-adaptively in the PoQ scheme because the verifiers need to \emph{fire} the challenges before receiving the answers for previous rounds.
\par We now describe the construction formally. 
\begin{construction}\label{constr:cvpv_mr_rf}
    In this construction, we assume the same syntax for $\poq = (V,P)$ as in \Cref{constr:cvpv_mr_seq}. In addition, we assume that $V$ samples challenges non-adaptively. That is, if $V = (\gen_1, \dots, \gen_\ell, \ver)$, then $\gen_i$ ignores the inputs $\ans_1,\dots,\ans_{i-1}$, that is, there exists $\gentild_i$ such that
    \begin{equation}
        \gen_i(1^\secparam, \ans_1, \dots, \ans_{i-1};r) =: \gentild_i(1^\secparam;r).
    \end{equation}
    \par Let $\bracC{\hash_k}_{k \in \bit^\secparam}:\bit^{m} \to \bit^n$ be a cryptographic hash function family, with $m = \omega(\log \secparam)$. Set\footnote{While this condition is all we need in the idealized model, in practice the optimal value of $\Delta$ is not necessarily the smallest possible value, for a small $\Delta$ will make a faster quantum computer necessary for the prover.} $\Delta \in \brac{0,1/(\ell-1)}$. We describe the (multi-round) CVPV protocol below: 
    \begin{enumerate}
    \item At time $t=-\infty$, the verifiers sample random coins $r \from \bit^{\poly(\secparam)}$ for $V$ and a hash key $k \from \bit^\secparam$. For $i = 1,\dots,\ell$, they sample random inputs $x_i,y_i \from \bit^m$. They publish the hash key $k$.
    \item For $i=1,\dots,\ell$: \begin{itemize}
        \item At time $t=(i-1)\Delta$, $V_0$ computes $\ch_i = \gentild_i(1^\secparam;r)$ and $s_i = \hash_k(x_i \oplus y_i) \oplus \ch_i$. It sends $(x_i,s_i)$ and expects an answer $\ans_i$ at time $t=(i-1)\Delta + 1$.
        \item Similarly, at time $t=(i-1)\Delta$, $V_1$ sends $y_i$ and expects an answer $\ans'_i$ at time $t=(i-1)\Delta + 1$. 
        
        \item \sloppy At time $t=(i-1)\Delta + 1/2$, the honest prover, located at position $0.5$, computes $\ch_i = \hash_k(x_i\oplus y_i) \oplus s_i$ and $\ans_i \from P_i(\ch_1,\dots,\ch_i,\ans_1,\dots,\ans_{i-1})$. It immediately sends $\ans_i$ to both verifiers.
    \end{itemize}
    
    \item The verifiers accept iff $\ans_i = \ans_i'$ for all $i$, and $\ver(\ch_1, \ans_1, \dots, \ch_\ell, \ans_\ell;r)$ accepts.
\end{enumerate}

\end{construction}

\noindent\textbf{Security Proof in QROM.}

\begin{theorem}[Rapid-Firing] \label{thm:mr_rf_sec}
    Let $\poq = (V,P)$ be a PoQ scheme that satisfies \Cref{def:cr_multiround_nocomm}, such that $V$ non-adaptively samples challenges. Then, \Cref{constr:cvpv_mr_rf} is a sound CVPV scheme in the quantum random oracle model. 
\end{theorem}

\begin{proof}
    The proof of \Cref{thm:mr_rf_sec} is nearly identical to the proof of \Cref{thm:mr_seq_sec}. The major difference is that in the final hybrid we reach, the adversary has no time to communicate anymore due to the rapid-fire design, hence \Cref{def:cr} suffices. As in the proof of \Cref{thm:mr_seq_sec}, we will consider hybrid experiments:
    \begin{itemize}
        \item {\bf Hybrid 0:} This is the original CVPV soundness experiment.
        \item {\bf Hybrid 1:} In this hybrid, the adversary consists of only two parties: $\alice$ at position $0$ and $\bob$ at position $1$. W.l.o.g., $(\alice, \bob)$ perform a round of simultaneous communication at times $t=0,\Delta,\dots,(\ell-1)\Delta$.
        \item {\bf Hybrid 2:} In this hybrid, we additionally give $\alice$ and $\bob$ $\ch_i$ at time $t=(i-1)\Delta + 1$ for $i=1,\dots,\ell$.
        \item {\bf Hybrid 2.1-$\ell$:} We set {\bf Hybrid 2.0} to be {\bf Hybrid 2} and $\hash_0 := \hash$. For $i \in [\ell]$, we define {\bf Hybrid 2.$i$} to be the same as {\bf Hybrid 2.$(i-1)$}, except the oracle $G_{i-1}$ is replaced by the reprogrammed oracle $G_{i}$, where \begin{align*}
        \hash_i(z) = \begin{cases}
            \hash_{i-1}(z), \quad &z \ne x_i \oplus y_i \\
            u_i, \quad &z = x_i \oplus y_i
            \end{cases},
    \end{align*}
    with $u_i \from \bit^n$ being a fresh random string.
    \item {\bf Hybrid 3:} In this hybrid, $\alice$ and $\bob$ only receive $\ch_i$ at time $t=(i-1)\Delta + 1$, for $i \in [\ell]$, and no other input. They do not get access to the oracle $\hash_\ell$ either. 
\end{itemize}

Let $p_i$ be the optimal success probability of an efficient adversary in {\bf Hybrid $i$}. Let $q = \poly(\secparam)$ be an upper-bound on the total number of oracle queries made by $(\alice, \bob)$. We will show a sequence of claims which suffice for the proof:

\begin{claim} \label{clm:rf01}
    $p_1 \ge p_0$.
\end{claim}
\begin{proof}
     Follows by a simple generalization of the corresponding claim in the proof of \Cref{thm:sr_sec}. 
\end{proof}

\begin{claim} \label{clm:rf12}
    $p_2 \ge p_1$.
\end{claim}
\begin{proof}
    Since we give extra information to the adversary, the success probability cannot decrease.
\end{proof}

\begin{claim} \label{clm:rf22}
    Setting $p_{2.0} := p_2$, $p_{2.i} \ge p_{2.(i-1)} - \negl(\secparam)$ for $i \in [\ell]$.
\end{claim}

\begin{proof}
    Let $i \in [\ell]$ and $(\alice, \bob)$ be an adversary that succeeds in {\bf Hybrid 2.$(i-1)$} with probability $p_{2.(i-1)}$. We will give a reduction $(\alice', \bob')$ for {\bf Hybrid 2.$i$}: \begin{itemize}
        \item At times $t<(i-1)\Delta + 1$, $\alice'$ (resp. $\bob'$) runs $\alice$ (resp. $\bob$) using $G_{i}$ as the oracle.
        \item At time $t=(i-1)\Delta$, $\alice'$ sends $(x_i,s_i)$ and $\bob'$ sends $y_i$ to each other, so that the messages are received at $t=(i-1)\Delta + 1$.
        \item At times $t \ge (i-1)\Delta + 1$, $\alice'$ and $\bob'$ can simulate $G_{i-1}$ using $(x_i,y_i,s_i,\ch_i,G_i)$ by reprogramming $G_i$ to output $\ch_i\oplus s_i$ on input $x_i \oplus y_i$.
        \item $\alice'$ (resp. $\bob'$) outputs what $\alice$ (resp. $\bob$) outputs.
    \end{itemize}
    Note that the view of $(\alice,\bob)$ as simulated by $(\alice',\bob')$ differs from {\bf Hybrid 2.$(i-1)$} at times $t<(i-1)\Delta + 1$, and only on input $x_i \oplus y_i$ to the oracle. Therefore, if the probability that $(\alice',\bob')$ succeeds is upper-bounded by $p_{2.(i-1)} - \eps$ for some non-negligible function $\eps$, then by \Cref{lem:bbbv} the total query weight by $(\alice,\bob)$ on input $x_i \oplus y_i$ at times $t<(i-1)\Delta + 1$ must be at least $2\eps^2/q$. Suppose the query weight by $\alice$ is at least $\eps^2/q$, for the other case is similar. We give an extractor $(\alicetild, \bobtild)$ in {\bf Hybrid 2.$(i-1)$}: \begin{itemize}
        \item $\alicetild$ samples $j \from [q]$ and simulates $\alice$, stopping the execution at the $j$-th query made by $\alice$ to the oracle $\hash_{i-1}$, measuring the query as $z^*$. She outputs $y^* = z^* \oplus x_i$, where $x_i$ is received at time $t=(i-1)\Delta$ from $V_0$.
        \item $\bobtild$ simulates $\bob$.
    \end{itemize}
    By assumption, $y^*=y_i$ with probability $\eps^2/q^2 > 2^{-m}$, which is a contradiction since $y_i$ is information theoretically hidden from $\alicetild$ at times $t<(i-1)\Delta$.
\end{proof}

\begin{claim} \label{clm:rf23}
    $p_3 \ge p_{2.\ell}$.
\end{claim}
\begin{proof}
    Let $(\alice,\bob)$ be an adversary for {\bf Hybrid 2.$\ell$} that succeeds with probability $p_{2.\ell}$. We give a reduction $(\alice', \bob')$ that succeeds in {\bf Hybrid 3} with the same probability: \begin{itemize}
        \item At time $t=-\infty$, $\alice'$ and $\bob'$ sample a $2q$-wise independent hash function $\hash'$. In addition, they sample $(x_i,y_i,s_i) \from \bit^m \times \bit^m \times \bit^n$ for $i \in [\ell]$.
        \item $\alice'$ simulates $\alice$ using $\hash'$ as the oracle, the sampled values $(x_i,y_i)$, as well as the values $\ch_i$ received from $V_0$.
        \item $\bob'$ similarly simulates $\bob$ using $\hash'$ as the oracle, the sampled values $y_i$, as well as the values $\ch_i$ received from $V_1$.
    \end{itemize}
    The view of $(\alice,\bob)$ is perfectly simulated since the oracle $G_\ell$ in {\bf Hybrid 2.$\ell$} is independent of the values $(x_i,y_i,s_i)$ for all $i \in [\ell]$. This is because the oracle has been reprogrammed on all inputs $x_i \oplus y_i$ to remove any such dependence. Thus, by \Cref{lem:zhandry_hash}, $\hash'$ simulates an independent random oracle and the proof is complete.
\end{proof}

\begin{claim} \label{clm:rf3}
    $p_3 \le \negl(\secparam).$
\end{claim}
\begin{proof}
    Follows from the certified randomness property (\Cref{def:cr_multiround_nocomm}). $\alice$ plays the role of the prover and $\bob$ that of the guesser. Note that the no-communication condition is satisfied because $(\ell-1) \Delta < 1$, which means no message can be sent between $\alice$ and $\bob$ in time from the first challenge until the last.
\end{proof}
\end{proof}

\begin{remark}[Comparison to Sequential Compilation] \label{rem:rf_vs_seq}
    Rapid-firing is theoretically more advantageous than sequential composition, as the latter requires a stronger notion of certified randomness. Nonetheless, this is in the idealized (vanilla) model of CVPV, and in practice the robustness of rapid-firing will be worse; it will require shorter network delays and shorter computation time. 
\end{remark}

%% file: Files/sequential-rapid-firing.tex
\subsection{Sequential Rapid-Firing}\label{sec:seq_rf}

One shortcoming of the rapid-firing CVPV construction is that the number of messages which can be sent depends on the distance between the provers and the time it takes the honest prover to respond. Instead, we consider $m$ sequential rounds of the $\ell$-round rapid-firing protocol. Now, although there may be a practical bound on how large $\ell$ can be, there is no such bound on $m$. In this protocol, we require all rounds pass the consistency check, but only some fraction $\alpha$ of rounds need to pass the certified randomness test.

In the following discussion of sequential compilation of the rapid-firing protocol, we denote each rapid-firing round combined quantum channel of the verifier, prover, and guesser as $\mathcal{M}_i:R_{i-1} E_{i-1}\rightarrow \Omega_i$, where $i\in[m]$, $R_{i-1},E_{i-1}$ are the prover and guesser input quantum memory, and $\Omega_i$ is the event that the $i$th rapid-firing round succeeds (both pass the certified randomness test and the consistency check). We only keep the abort status register $\Omega_i$ for each round and traced out any output registers for the answers, guessers, and any other registers used for certified randomness check. This is because the traced out registers are only intermediate results used to determine $\Omega_i$, and only $\Omega_i$ ultimately determines if the overall sequentially compiled protocol aborts.

To show asymptotic soundness in $m$ while there is no restriction in communication and setup between rapid-fire rounds, we first show the following lemma.

\begin{lemma} \label{lemma:indep}
    There exists independent states $\rho_{R_{0}E_{0}}, \rho_{R_{1}E_{1}}, \dots, \rho_{R_{m-1}E_{m-1}}$ such that \[ \pr{ \Omega_i=1 \ \forall i \in S }_{\nu} \leq \prod_{i \in S} \pr{ \Omega_i=1 }_{\mathcal{M}_{i}\brac{\rho_{R_{i-1}E_{i-1}}}} \]
    for any $S \subset [m]$ and $\nu = \mathcal{M}_{m} \circ \cdots \circ \mathcal{M}_{1}\brac{ \sigma_{R_{0}E_{0}} }$.
\end{lemma}
\begin{proof}
    We one-by-one choose the $\rho_{R_{i-1}E_{i-1}}$ to be optimal in maximizing $\pr{\Omega_i=1}$. Fix $S \subset [m]$. Then,
    \begin{align*}
        &\pr{ \Omega_i=1 \ \forall i \in S }_{\nu} \leq \underset{\rho_{R_0E_0}}{\sup}\pr{ \Omega_i=1 \ \forall i \in S }_{\mathcal{M}_{m} \circ \cdots \circ \mathcal{M}_{1}\brac{ \rho_{R_{0}E_{0}} }}\\
        =&  \underset{\rho_{R_0E_0}}{\sup}\prod_{i\in S}\pr{\Omega_i=1\vert \Omega_j=1\forall j<i\in S}_{\mathcal{M}_{i} \circ \cdots \circ \mathcal{M}_{1}\brac{ \rho_{R_{0}E_{0}} }}\\
        =&  \underset{\rho_{R_0E_0}}{\sup}\prod_{i\in S}\pr{\Omega_i=1}_{\nu_i(\rho_{R_0E_0})\big\vert \Omega_j=1\forall j<i\in S}\\
        \leq &\prod_{i\in S}\underset{\rho'_{R_{i-1}E_{i-1}}}{\sup}\pr{\Omega_i=1}_{\mathcal{M}_i(\rho'_{R_{i-1}E_{i-1}})}
        =\underset{\rho_{R_0E_0}}{\sup}\prod_{i\in S}\pr{\Omega_i=1}_{\nu_i(\rho_{R_0E_0})\big\vert \Omega_j=1\forall j<i\in S}\\
        \leq &\prod_{i\in S}\pr{\Omega_i=1}_{\mathcal{M}_i(\rho_{R_{i-1}E_{i-1}})},
    \end{align*}
    where $\nu_i(\rho_{R_0E_0})\equiv\mathcal{M}_{i} \circ \cdots \circ \mathcal{M}_{1}\brac{ \rho_{R_{0}E_{0}} }$, $\nu_i(\rho_{R_0E_0})\big\vert \Omega_j=1\forall j<i\in S$ is the normalized state of $\nu_i(\rho_{R_0E_0})$ projected to satisfy the condition $\Omega_j=1\forall j<i\in S$, and $\rho_{R_{i-1}E_{i-1}}=\underset{\rho'_{R_{i-1}E_{i-1}}}{\argmax} \pr{\Omega_i=1}_{\mathcal{M}_i}(\rho'_{R_{i-1}E_{i-1}})$.
\end{proof}

This allows us to provide an upper bound on the overall protocol success probability.

\begin{theorem}
    Suppose we have a single-round or rapid-firing protocol such that the probability of success is upper bounded by $p$ over all possible input quantum states. An $m$-round sequential compilation of these protocols must have
    \begin{equation}
        \pr{\sum_i \Omega_i\geq \alpha m}_\nu \leq \brac{ \frac{ep}{\alpha} }^{\floor{ \alpha m }}.
    \end{equation}
\end{theorem}

\begin{proof}
    We first notice that by a union-bound together with \Cref{lemma:indep},
    \begin{align*}
        &\pr{\sum_{i} \Omega_{i} \geq \alpha m}_{\nu}         \leq \pr{ \exists S, \abs{S} = \floor{\alpha m} , \ \Omega_{i} = 1 \ \forall i \in S}_{\nu} \\
                &\leq \sum_{S} \pr{ \Omega_{i} = 1 \ \forall i \in S }_{\nu} \leq \sum_{S} \prod_{i \in S} \pr{ \Omega_{i} = 1 }_{\mathcal{M}_{i}\brac{ \rho_{R_{i-1}E_{i-1}} }} \\
                &\leq \sum_{S} p^{\floor{\alpha m }} = \binom{m}{\floor{\alpha m}} p^{\floor{ \alpha m } } \leq \brac{ \frac{e}{\alpha} }^{\floor{\alpha m} } p^{\floor{\alpha m}}.
    \end{align*}
    This concludes the proof.
   \end{proof}

%% file: Files/computational_entropy.tex
\section{Certified Randomness from CVPV}

Recall that our work is partially motivated by discovering the fundamental property needed to construct CVPV. We have shown how to construct CVPV using certifed randomness, but so far have not discussed the reverse direction. In this section, we complete the picture by showing the reverse direction. To this end, we define a weaker version of certified randomness

\subsection{Weak Certified Randomness: Definitions}

We define a weak version of certified randomness by restricting the guesser to be computationally bounded. 
\par Below, we give two definitions, one for single-round and one for multi-round. For the multi-round case, we use the weakest definition (\Cref{def:comp_cr_mr_nocomm}) which does not allow communication in-between rounds. Since this definition can be used to achieve CVPV, we do not expect   We \highlight{highlight} deviations from the corresponding definitions (\Cref{def:cr,def:cr_multiround_nocomm}).

\begin{definition}[Weak Certified Randomness (Single-Round)] \label{def:comp_cr}
    A proof of quantumness protocol $\poq$ is said to have \emph{\highlight{weak} certified randomness} property if no pair of a QPT prover $P$ and a \highlight{QPT} guesser $Q$ can succeed in the following security game with non-negligible probability: \begin{enumerate}
        \item The verifier $V$ of $\poq$ sends a challenge $\ch$ to both $P$ and $Q$.
        \item $P$ sends back an answer $\ans$ and $Q$ and outputs a guess $\ans'$.
        \item $(P,Q)$ win if $V$ accepts and $\ans = \ans'$.
    \end{enumerate}
\end{definition}

\begin{definition}[Weak Multi-round Sequential Certified Randomness with Sequential Decomposition and No Communication] \label{def:comp_cr_mr_nocomm}
    An $\ell$-round PoQ protocol $\poq$ is said to have \emph{\highlight{weak} sequential certified randomness with sequential decomposition and no communication} property if no pair of a QPT prover $P$ and a \highlight{QPT} guesser $\eve$ can succeed in the following security game with non-negligible probability: \begin{itemize}
        \item For $i \in [\ell]$, the following steps occur in order: \begin{enumerate}
            \item The verifier $V$ of $\poq$ sends a challenge $\ch_i$ to both $P$ and $\eve$.
            \item $P$ sends back an answer $\ans_i$.
            \item $\eve$ outputs a guess $\ans_i'$.
            \item $P$ and $\eve$ cannot communicate.
        \end{enumerate}
        \item $(P,\eve)$ win the game if $V$ accepts and $\ans_i = \ans_i'$ for all $i \in [\ell]$.
    \end{itemize}
\end{definition}

Note that with this relaxation the definition still suffices for our CVPV construction since in the security proof, the guesser corresponds to a spoofer for position-verification, hence is computationally bounded. This holds true in both the single-round and the multi-round setting. We list the formal statements below for completeness:

\begin{theorem}[Single-Round (V2)] \label{thm:sr_from_weak_cr}
    Let $\poq = (V,P)$ be a one-round PoQ scheme that satisfies \Cref{def:comp_cr}. Then, \Cref{constr:cvpv_sr} is a CVPV scheme that is sound in the quantum random oracle model. 
\end{theorem}

\begin{theorem}[Rapid-Firing (V2)] \label{thm:mr_rf_from_weak_cr}
    Let $\poq = (V,P)$ be a PoQ scheme that satisfies \Cref{def:comp_cr_mr_nocomm}, such that $V$ non-adaptively samples challenges. Then, \Cref{constr:cvpv_mr_rf} is a sound CVPV scheme in the quantum random oracle model. 
\end{theorem}

\subsection{Weak Certified Randomness from CVPV}
Following the argument \cite{LLQ22} that shows CVPV implies PoQ, we show that it implies weak certified randomness. \Cref{thm:mr_rf_from_weak_cr} show that weak certified randomness with restricted communication (\Cref{def:comp_cr_mr_nocomm}) is the fundamental property required to construct CVPV\footnote{To be precise, we also require non-adaptive challenge sampling for this \emph{no communication in-between rounds} setting.}.

\begin{theorem} \label{thm:cvpv_to_cr}
    If there exists a CVPV scheme (\Cref{def:cvpv}), then there exists a PoQ protocol with weak certified randomness (\Cref{def:comp_cr_mr_nocomm}).
\end{theorem}

\begin{proof}
    The construction is natural, and given in \cite{LLQ22}. Without loss of generality, let $V_0,V_1$ be the left, right verifiers and let $P^*$ be the honest prover of the CVPV scheme. We construct a PoQ protocol $\poq = (V,P)$ as follows:
    \begin{itemize}
        \item {\bf Verifier:} $V$ runs $V_0,V_1$ together, sending and expecting messages consistent with the timing constraints imposed by the CVPV scheme. 
        \item {\bf Prover:} $P$ simply runs $P^*$, sending all messages to $V$.
    \end{itemize}

In more detail, if $t=t_1$ is the time $P^*$ receives the first message(s) from $(V_0,V_1)$, then $V$ sends that message(s) as the first challenge $\ch_1$. Then, $P$ will send all messages $P^*$ is supposed to send before receiving any new information as $\ans_1$, and so on until the entire CVPV protocol is covered.

\par Onto the proof of security, suppose there exists an adversary $(P, \eve)$ breaking the weak certified randomness with non-negligible probability $\epsilon$. We construct $(\alice, \bob)$ that breaks the CVPV security. $\alice$ will be located at the position of $V_0$ and run $P$, whereas $\bob$ will be located at the position of $V_1$ and run $\eve$. In each round, after receiving the challenges from the nearest verifier, $\alice$ and $\bob$ will forward the challenges to each other. Once they receive both challenges, they will run $P, \eve$ respectively and send the answer to the nearest verifier. With this reduction, it is easy to see that there is a perfect correspondence between the security game of the CVPV and weak certified randomness of $\poq$. Therefore, $(\alice,\bob)$ beat the CVPV security with probability $\epsilon$, which is contradiction.
    
\end{proof}

\begin{remark} \label{rem:cvpv_to_cr}
    Note that if we have an additional guarantee on the space-time structure of the CVPV scheme, it will yield a stronger notion of weak certified randomness. For instance, if the CVPV has a sequential structure so as to only allow simultaneous single-round communication in between challenges, then it would yield the weak version of \Cref{def:cr_multiround_simult}, which can be similarly defined. In this sense, our compilers are optimal in the random oracle model.
\end{remark}

%% file: Files/crea.tex
\newpage
\section{CVPV from Certified Randomness through Repetition}\label{sec:seq_crea}

A summary of notations used in this appendix is presented in \Cref{table:notations}. In this appendix, we prove that certified randomness protocols that repeat single rounds with an average von Neumann entropy lower bound, or more formally those that satisfy \Cref{def:crea}, satisfy the sequential decomposition property of \Cref{def:cr_multiround} and can be used to construct a CVPV protocol secure under the quantum random oracle model by \Cref{thm:mr_seq_sec}.

To proceed, we first provide a formal model of the sequential protocol quantum channels in \Cref{sec:protocol_quantum_channels}. We then consider an optimal adversary with the largest acceptance probability $\Pr[\Omega]$ in \Cref{sec:equal_probability}. We also construct a modified adversary from this optimal adversary, and show that the modified adversary has the same $\Pr[\Omega]$. Then, we show the modified adversary also satisfies the non-signaling condition required by the generalized entropy accumulation theorem in \Cref{sec:crea_entropy}, which allows us to apply entropy accumulation to lower bound the entropy of the prover output conditioned on the guesser side information. Finally, we derive an upper bound on $\Pr[\Omega]$ by using the fact that either the certified randomness protocol test condition $\omega'$ fails with high probability, or the prover output $A_1^n$ has high entropy relative to the guesser side information $E_n$ in \Cref{sec:proof_sequential_decomposition}. This concludes the proof of the sequential decomposition property.

\subsection{Protocol Quantum Channels}\label{sec:protocol_quantum_channels}

We first provide a model of the sequential protocol. We model the action of the verifier $V$, the prover $P$, and the guesser $\eve$ in the $i$-th round as a single  map $\mathcal{M}_i$. We denote the register of the classical challenge as $C_i$, the register of the classical prover output from $P$ as $A_i$, the classical guess of the guesser as $G_i$. Furthermore, $P$ and $\eve$ may share an arbitrary entangled quantum state on $R_{i-1}$ and $R'_{i-1}$ on input. They then engage in arbitrary communication to produce an output state on $R_i$ and $R'_i$ for the next round. We also use the notation $E_i=C_1^i G_1^i R'_i$, where we use the notation $Y_i^j=Y_i\cdots Y_j$ to denote the values of a given set of registers $Y$ between rounds $i$ and $j$.

The verifier $V$ also computes some test result based on the values in $C_i A_i$ 
each round. Namely, $X_i$ are classical systems with common alphabet $\mathcal{X}$, and the test  can be modeled by the map $\mathcal{T}_i:A_i C_i\rightarrow X_i A_i C_i$,

where
\begin{equation}
    \mathcal{T}_i\left(\rho_{A_i C_i}\right)=\sum_{a_i, c_i}\left(\Pi_{A_i}^{a_i}\otimes\Pi_{C_i}^{c_i}\right)\rho_{A_i C_i}\left(\Pi_{A_i}^{a_i}\otimes\Pi_{C_i}^{c_i}\right)\otimes\vert x_i\rangle\langle x_i\vert_{X_i},\label{eqn:test_channel}
\end{equation}
where $\{\Pi_{Y}^y\}$ is the family of projectors on $Y$ to classical values $y$ and $x_i=x(a_i, c_i)$ for some deterministic function $x$. Overall, each round can be modeled by the map 
\begin{equation}
    \mathcal{M}_i:R_{i-1}E_{i-1}\rightarrow X_i A_i R_i E_i=\mathcal{T}_i\circ\mathcal{M}_i',
\end{equation}

\iffullv
\begin{table}[ht]
\centering
\setlength{\extrarowheight}{0.1em}
\begin{tabular}{ | m{1cm} | m{10cm} | } 
  \hline
  $V,P,Q$ & Verifier, prover, guesser \\
  \hline
  $\mathcal{T}_i$ & Quantum channel of the $i$th round prover test result. Explanation for the subscript $i$ will be omitted below\\
  \hline
  $\mathcal{C}_i$ & Quantum channel for challenge generation, part of the verifier $V$\\
  \hline
  $\mathcal{P}_i$ & Quantum channel of the certified randomness protocol of $V,P$ including the test\\
  \hline
  $\mathcal{P}_i'$ & Quantum channel of the certified randomness protocol of $V,P$ excluding the test\\
  \hline
  $\mathcal{G}_i$ & Quantum channel of the guesser $G$\\
  \hline
  $\mathcal{N}_i$ & Quantum communication channel between $P$ and $G$\\
  \hline
  $\mathcal{M}_i$ & Quantum channel of the CVPV protocol of $V,P,Q$ including test \\ 
  \hline
  $\mathcal{M}_i'$ & Quantum channel of the CVPV protocol of $V,P,Q$ excluding test \\ 
  \hline
  $\mathcal{M}_i^{*}$ & Original optimal quantum channel of the CVPV protocol of $V,P,Q$ including the test \\ 
  \hline
 $\mathcal{M}_i^{'*}$ & Original optimal quantum channel of the CVPV protocol of $V,P,Q$ excluding the test \\ 
 \hline
 $\Bar{\mathcal{M}}_i$ & Modified optimal quantum channel of the CVPV protocol of $V,P,Q$ including the test\\ 
  \hline
  $\Bar{\mathcal{M}}_i'$ & Modified optimal quantum channel of the CVPV protocol of $V,P,Q$ excluding the test, defined by \Cref{eqn:modified_channel}\\ 
  \hline
    $C_i$ & Register of the challenge of the $i$th round\\
  \hline
    $A_i$ & Register of the prover's answer for the $i$th round\\
  \hline
    $G_i$ & Register of the guesser's guess for the $i$th round\\
    \hline
    $Y_i^j$ & Collection of $Y_i\dots Y_j$ for any register $Y$\\
    \hline
    $R_{i-1}$ & Quantum memory input of the prover in the $i$th round\\
    \hline
    $R_{i-1}'$ & Quantum memory input of the guesser in the $i$th round\\
    \hline
    $E_i$ & All side information after the $i$th round, i.e. $E_i=C_1^i G_1^i R_i'$\\
    \hline
    $X_i$ & Test result register for the $i$th round\\
    \hline
    $\Omega'$ & Channel to determine if the protocol passes 1. the certified randomness test from $X_i$, and 2. the consistency check\\
    \hline
    $\Omega$ & $\Omega=\Omega'\circ\mathcal{T}_n\circ\dots\circ\mathcal{T}_1$. Also denotes the event of the success\\
    \hline
    $\omega'$ & Function taking the single-round test results $X_1^n$ and determine if the certified randomness protocol accepts\\
    \hline
    $W$ & Binary Register determining if the protocol succeeds\\
    \hline
    $\Sigma_i(q)$ & Set of all input states such that after passing through $\mathcal{P}_i$, the probability distribution on test output $X_i$ is $q$ \\
    \hline
    $\Sigma_i'(q)$ & Same as $\Sigma_i(q)$ except for the original optimal channel $\Bar{\mathcal{M}}_i^{'*}$\\
    \hline
    $\Bar{\Sigma}_i(q)$ & Same as $\Sigma_i(q)$ except for the modified optimal channel $\Bar{\mathcal{M}}_i'$\\
    \hline
    $\Tilde{E}$ & Purification register of $R_{i-1}E_{i-1}$\\
    \hline
\end{tabular}
\caption{Table of variables for this \Cref{sec:seq_crea}.}
\label{table:notations}
\end{table}
\else
\begin{table}[ht]
\centering
\setlength{\extrarowheight}{0.1em}
\begin{tabular}{ | m{1cm} | m{16.5cm} | } 
  \hline
  $V,P,Q$ & Verifier, prover, guesser \\
  \hline
  $\mathcal{T}_i$ & Quantum channel of the $i$th round prover test result. Explanation for the subscript $i$ will be omitted below\\
  \hline
  $\mathcal{C}_i$ & Quantum channel for challenge generation, part of the verifier $V$\\
  \hline
  $\mathcal{P}_i$ & Quantum channel of the certified randomness protocol of $V,P$ including the test\\
  \hline
  $\mathcal{P}_i'$ & Quantum channel of the certified randomness protocol of $V,P$ excluding the test\\
  \hline
  $\mathcal{G}_i$ & Quantum channel of the guesser $G$\\
  \hline
  $\mathcal{N}_i$ & Quantum communication channel between $P$ and $G$\\
  \hline
  $\mathcal{M}_i$ & Quantum channel of the CVPV protocol of $V,P,Q$ including test \\ 
  \hline
  $\mathcal{M}_i'$ & Quantum channel of the CVPV protocol of $V,P,Q$ excluding test \\ 
  \hline
  $\mathcal{M}_i^{*}$ & Original optimal quantum channel of the CVPV protocol of $V,P,Q$ including the test \\ 
  \hline
 $\mathcal{M}_i^{'*}$ & Original optimal quantum channel of the CVPV protocol of $V,P,Q$ excluding the test \\ 
 \hline
 $\Bar{\mathcal{M}}_i$ & Modified optimal quantum channel of the CVPV protocol of $V,P,Q$ including the test\\ 
  \hline
  $\Bar{\mathcal{M}}_i'$ & Modified optimal quantum channel of the CVPV protocol of $V,P,Q$ excluding the test, defined by \Cref{eqn:modified_channel}\\ 
  \hline
    $C_i$ & Register of the challenge of the $i$th round\\
  \hline
    $A_i$ & Register of the prover's answer for the $i$th round\\
  \hline
    $G_i$ & Register of the guesser's guess for the $i$th round\\
    \hline
    $Y_i^j$ & Collection of $Y_i\dots Y_j$ for any register $Y$\\
    \hline
    $R_{i-1}$ & Quantum memory input of the prover in the $i$th round\\
    \hline
    $R_{i-1}'$ & Quantum memory input of the guesser in the $i$th round\\
    \hline
    $E_i$ & All side information after the $i$th round, i.e. $E_i=C_1^i G_1^i R_i'$\\
    \hline
    $X_i$ & Test result register for the $i$th round\\
    \hline
    $\Omega'$ & Channel to determine if the protocol passes 1. the certified randomness test from $X_i$, and 2. the consistency check\\
    \hline
    $\Omega$ & $\Omega=\Omega'\circ\mathcal{T}_n\circ\dots\circ\mathcal{T}_1$. Also denotes the event of the success\\
    \hline
    $\omega'$ & Function taking the single-round test results $X_1^n$ and determine if the certified randomness protocol accepts\\
    \hline
    $W$ & Binary Register determining if the protocol succeeds\\
    \hline
    $\Sigma_i(q)$ & Set of all input states such that after passing through $\mathcal{P}_i$, the probability distribution on test output $X_i$ is $q$ \\
    \hline
    $\Sigma_i'(q)$ & Same as $\Sigma_i(q)$ except for the original optimal channel $\Bar{\mathcal{M}}_i^{'*}$\\
    \hline
    $\Bar{\Sigma}_i(q)$ & Same as $\Sigma_i(q)$ except for the modified optimal channel $\Bar{\mathcal{M}}_i'$\\
    \hline
    $\Tilde{E}$ & Purification register of $R_{i-1}E_{i-1}$\\
    \hline
\end{tabular}
\caption{Table of variables for this \Cref{sec:seq_crea}.}
\label{table:notations}
\end{table}
\fi
\noindent where $\mathcal{M}_i':R_{i-1}E_{i-1}\rightarrow A_i R_i E_i=\mathrm{tr}_{X_i}\mathcal{M}_i$ is the quantum channel without the test computation.

The $n$-round protocol output state is of the form 
\begin{equation}
    \rho_{X_1^n A_1^n R_n E_n}=\left(\mathcal{M}_n\circ\cdots\circ \mathcal{M}_1\right)\rho_{R_0 E_0},
\end{equation}
where $\rho_{R_0 E_0}$ is a density operator on $R_0 E_0$. Similarly,
\begin{equation}
    \rho_{A_1^n R_n E_n}=\mathrm{tr}_{X_1^n}\rho_{X_1^n A_1^n R_n E_n}=\left(\mathcal{M}_n'\circ\cdots\circ \mathcal{M}_1'\right)\rho_{R_0 E_0} \label{eqn:original_state}
\end{equation}
for the state without the test result registers.

The verifier finally determines whether to accept the transcript by checking that the outputs pass some test condition $\Omega$. Specifically, define $\Omega':X_1^n A_1^n G_1^n\rightarrow W X_1^n A_1^n G_1^n$ with action
\begin{equation}
    \Omega'(\rho)=\sum_{x_1^n,a_1^n,g_1^n}\Pi_{X_1^n A_1^n G_1^n}^{x_1^n a_1^n g_1^n}~\rho~\Pi_{X_1^n A_1^n G_1^n}^{x_1^n a_1^n g_1^n}\otimes \vert \omega\rangle\langle \omega\vert_W,
\end{equation}
where
\begin{equation}
    \omega=\omega(x_1^n a_1^n g_1^n)=\begin{cases}
        \omega'(x_1^n)\quad&\text{if}~A_1^n=G_1^n\\
        0\quad &\text{otherwise},
    \end{cases}
\end{equation}
and $\omega'$ is a deterministic function with target $\{0,1\}$. We can define the full test operator $\Omega$ as
\begin{equation}
    \Omega=\Omega'\circ\mathcal{T}_n\circ\cdots\circ \mathcal{T}_1.
\end{equation}

All of this is to say that the verifier only accepts if 1. some test condition $\omega'$ on the prover outputs $A_1^n$ and the challenges $C_1^n$ is satisfied, and 2. the consistency check passes (i.e. $A_1^n=G_1^n$). The probability of accepting is
\begin{equation}
    \Pr[\Omega]_{\rho_{A_1^n R_n E_n}}=\Pi_W^1\left[\mathrm{tr}_{X_1^n A_1^n R_n E_n}\circ\Omega\left(\rho_{A_1^n R_n E_n}\right)\right]\Pi_W^1,
\end{equation}
where $\Pi_W^1$ is a projector on $W$ to the $1$ state.

The state in \cref{eqn:original_state} is built from a sequence of quantum channels, similar to those considered in entropy accumulation. In the context of randomness expansion, the guesser's information about the prover's output can be bounded given certain restrictions on this sequence of channels using entropy accumulation theorems, namely the Markov chain or non-signaling condition. However, in our context of CVPV, we do not have such restrictions on the channels since $P$ and $\eve$ can perform arbitrary communication after $A_i,G_i$ are output each round.

In fact, if arbitrary communication is allowed, the adversary can simply copy all $A_i$ to the guesser's final side information $E_n$, which results in zero entropy. 
However, to be usefully helpful in assisting the guesses, the side information for the $i$th round must be present in the $i$th round, not before or after. This exactly what the Markov chain condition and the non-signaling condition in the entropy accumulation and the generalized entropy accumulation theorem aim to enforce.

\subsection{Proof of Equal Abort Probability}\label{sec:equal_probability}
We first show that modifying any quantum channel by ignoring the prover output does not change the protocol success probability.
\begin{lemma}\label{lem:equal_prob}
Given a quantum channel $\Omega\circ\mathcal{M}_n^{'*}\circ\cdots\circ \mathcal{M}_1^{'*}$ with classical registers $A_1^n C_1^n G_1^n$ and intermediate states 
\begin{align}\nonumber
    \rho_{A_1^i R_i E_i}^*&=\left(\mathcal{M}_i^{'*}\circ\cdots\circ \mathcal{M}_1^{'*}\right)\left(\rho_{R_0 E_0}\right)\\
    &=\sum_{a_1^i c_1^i g_1^i}p^*(a_1^i c_1^i g_1^i)\vert a_1^i c_1^i g_1^i\rangle\langle a_1^i c_1^i g_1^i\vert\otimes \rho_{R_i R_i'}^{*a_1^i c_1^i g_1^i}\label{eqn:intermediate_memory},
\end{align}
a related quantum channel $\Omega\circ\Bar{\mathcal{M}}_n'\circ\cdots\circ \Bar{\mathcal{M}}_1'$ defined by
\begin{equation}
    \Bar{\mathcal{M}}_{i+1}'=\Gamma_{i+1}\circ\mathrm{tr}_{R_{i+1} R_{i+1}'}\circ\mathcal{M}_{i+1}^{'*}\circ\Gamma_i\circ\mathrm{tr}_{R_i R_i'}\label{eqn:modified_channel_equal_prob},
\end{equation}
where $\Gamma_i:C_1^i G_1^i\rightarrow R_i E_i$ is a quantum channel given by
\begin{equation}
    \Gamma_i(\rho)=\sum_{c_1^i g_1^i}\Pi_{C_1^i G_1^i}^{c_1^i g_1^i}~\rho~\Pi_{C_1^i G_1^i}^{c_1^i g_1^i}\otimes \vert c_1^i g_1^i\rangle\langle c_1^i g_1^i\vert\otimes\rho_{R_i R_i'}^{* g_1^i c_1^i g_1^i},\label{eqn:gamma_definition}
\end{equation}
has the same probability of accept as the original quantum channel. Formally,
\begin{align}
    \Pr[\Omega]_{\rho_{A_1^n R_n E_n}^*}=\Pr[\Omega]_{\rho_{A_1^n R_n E_n}}
\end{align}
for $\rho_{A_1^n R_n E_n}=\left(\Bar{\mathcal{M}}_n'\circ\cdots\circ \Bar{\mathcal{M}}_1'\right)\left(\rho_{R_0 E_0}\right)$.
\end{lemma}

\begin{proof}
    In general, the quantum memory state $\rho_{R_i R_i'}^{* a_1^i c_1^i g_1^i}$ conditioned on any particular outcome depends on $a_1^i c_1^i g_1^i$, and therefore the side information $E_{i+1}$ of the next round might depend on $a_1^i$, violating the no-signaling condition. Applying the second half of the optimal adversary yields
\begin{align}\nonumber
    \rho_{A_1^n R_n E_n}^*=&\left(\mathcal{M}_n^{'*}\circ\cdots\circ \mathcal{M}_{i+1}^{'*}\right)\rho_{A_1^i R_i E_i}^*\\
    =&\sum_{a_1^i c_1^i g_1^i}p^*(a_1^i c_1^i g_1^i)\vert a_1^i c_1^i g_1^i\rangle\langle a_1^i c_1^i g_1^i\vert\left(\mathcal{M}_n^{'*}\circ\cdots\circ \mathcal{M}_{i+1}^{'*}\right)\rho_{R_i R_i'}^{* a_1^i c_1^i g_1^i}.
\end{align}
Therefore, we can express the non-abort probability as
\begin{align}\nonumber
    \Pr[\Omega]&=\sum_{a_1^i c_1^i g_1^i}p^*(a_1^i c_1^i g_1^i)~\Pr[\Omega]_{\vert a_1^i c_1^i g_1^i\rangle\langle a_1^i c_1^i g_1^i\vert\left(\mathcal{M}_n^{'*}\circ\cdots\circ \mathcal{M}_{i+1}^{'*}\right)\rho_{R_i R_i'}^{* a_1^i c_1^i g_1^i}}\\
    &=\sum_{c_1^i g_1^i}p^*(g_1^i c_1^i g_1^i)~\Pr[\Omega]_{\vert g_1^i c_1^i g_1^i\rangle\langle g_1^i c_1^i g_1^i\vert\left(\mathcal{M}_n^{'*}\circ\cdots\circ \mathcal{M}_{i+1}^{'*}\right)\rho_{R_i R_i'}^{* g_1^i c_1^i g_1^i}},
\end{align}
where the second equality holds because applying $\Omega$ to states where $a_1^i\neq g_1^i$ gives $\vert\omega=0\rangle_{W}$.
Now, let us consider the effect of inserting $\Gamma_i\circ\mathrm{tr}_{R_i R_i'}$, which leads to an intermediate quantum state
\begin{align}
    \left(\Gamma_i\circ\mathrm{tr}_{R_i R_i'}\mathcal{M}_i^{'*}\circ\cdots\circ \mathcal{M}_1^{'*}\right)\rho_{R_0 E_0}=\sum_{a_1^i c_1^i g_1^i} p^*(a_1^i c_1^i g_1^i)\vert a_1^i c_1^i g_1^i\rangle\langle a_1^i c_1^i g_1^i\vert \rho_{R_i R_i'}^{* g_1^i c_1^i g_1^i}.
\end{align}
Applying the rest of the quantum channel $\left(\mathcal{M}_n^{'*}\circ\cdots\circ \mathcal{M}_{i+1}^{'*}\right)$ yields
\begin{align}
    \rho_{X_1^n A_1^n R_n E_n}=\sum_{a_1^i c_1^i g_1^i}p^*(a_1^i c_1^i g_1^i)\vert a_1^i c_1^i g_1^i\rangle\langle a_1^i c_1^i g_1^i\vert\left(\mathcal{M}_n^{'*}\circ\cdots\circ \mathcal{M}_{i+1}^{'*}\right)\rho_{R_i R_i'}^{* g_1^i c_1^i g_1^i}.
\end{align}
Therefore,
\begin{align}\nonumber
    \Pr[\Omega]&=\sum_{a_1^i c_1^i g_1^i}p^*(a_1^i c_1^i g_1^i)~\Pr[\Omega]_{\vert a_1^i c_1^i g_1^i\rangle\langle a_1^i c_1^i g_1^i\vert\left(\mathcal{M}_n^{'*}\circ\cdots\circ \mathcal{M}_{i+1}^{'*}\right)\rho_{R_i R_i'}^{* g_1^i c_1^i g_1^i}}\\
    &=\sum_{c_1^i g_1^i}p^*(g_1^i c_1^i g_1^i)~\Pr[\Omega]_{\vert a_1^i c_1^i g_1^i\rangle\langle a_1^i c_1^i g_1^i\vert\left(\mathcal{M}_n^{'*}\circ\cdots\circ \mathcal{M}_{i+1}^{'*}\right)\rho_{R_i R_i'}^{* g_1^i c_1^i g_1^i}},
\end{align}
which is identical to the probability of the original optimal adversary. The only difference between the derivation is that we replace $\rho_{R_i R_i'}^{* a_1^i c_1^i g_1^i}$ with $\rho_{R_i R_i'}^{* g_1^i c_1^i g_1^i}$ everywhere instead of just the final result.

Applying the modified adversary defined in E.q. \ref{eqn:modified_channel_equal_prob} for all $i$ is simply inserting $\Gamma_i\circ\mathrm{tr}_{R_i R_i'}$ between all $\mathcal{M}_i^{'*}$ ($\Gamma_i\circ\mathrm{tr}_{R_i R_i'}\circ\Gamma_i\circ\mathrm{tr}_{R_i R_i'}=\Gamma_i\circ\mathrm{tr}_{R_i R_i'}$). Since the final $\Gamma_n\circ\mathrm{tr}_{R_n R_n'}$ does not change $A_1^n C_1^n G_1^n$, the final probability $\pr{\Omega}$ does not change.
\end{proof}

This allows to construct a modified optimal adversary channel that ignores prover outputs.
\begin{lemma}
    For quantum channel $\Bar{\mathcal{M}}_n'\circ\cdots\circ \Bar{\mathcal{M}}_1'$ such that 
    \begin{equation}
        \Bar{\mathcal{M}}_{i+1}'=\Gamma_{i+1}\circ\mathrm{tr}_{R_{i+1} R_{i+1}'}\circ\mathcal{M}_{i+1}^{'*}\circ\Gamma_i\circ\mathrm{tr}_{R_i R_i'},\label{eqn:modified_channel}
    \end{equation}
    where
    \begin{equation}
    \mathcal{M}_n^{'*}\circ\cdots\circ \mathcal{M}_1^{'*},\rho_{R_0 E_0}^*
    =\argmax_{\mathcal{N}_n'\circ\cdots\circ \mathcal{N}_1'\in\mathsf{poly},\rho_{R_0 E_0}\in S(R_0 E_0)}\Pr[\Omega]_{\left(\mathcal{N}_n'\circ\cdots\circ \mathcal{N}_1'\right)\rho_{R_0 E_0}},
\end{equation}
    and $S(R_0 E_0)$ is the space of all density operators on $R_0 E_0$, we have
    \begin{equation}
        \sup_{\mathcal{N}_n'\circ\cdots\circ \mathcal{N}_1'\in\mathsf{poly},\rho_{R_0 E_0}\in S(R_0 E_0)}\Pr[\Omega]_{\left(\mathcal{N}_n'\circ\cdots\circ \mathcal{N}_1'\right)\rho_{R_0 E_0}}=\Pr[\Omega]_{\left(\Bar{\mathcal{M}}_n'\circ\cdots\circ \Bar{\mathcal{M}}_1'\right)\rho_{R_0 E_0}^*}.
    \end{equation}
\end{lemma}
\begin{proof}
    By Lemma \ref{lem:equal_prob}, 
    \begin{equation}
        \Pr[\Omega]_{\left(\Bar{\mathcal{M}}_n'\circ\cdots\circ \Bar{\mathcal{M}}_1'\right)\rho_{R_0 E_0}^*}=\Pr[\Omega]_{\left(\mathcal{M}_n^{'*}\circ\cdots\circ \mathcal{M}_1^{'*}\right)\rho_{R_0 E_0}^*}.
    \end{equation}
    By definition of $\mathcal{M}_n^{'*}\circ\cdots\circ \mathcal{M}_1^{'*},\rho_{R_0 E_0}^*$ and the supremum, this completes the proof.
\end{proof}

\subsection{Proof of Entropy}\label{sec:crea_entropy}

\begin{lemma}\label{lem:no_signaling}
The quantum channel defined in E.q. \ref{eqn:modified_channel} satisfies the non-signalling condition. Namely, for all $i$, there exists $\mathcal{R}_{i+1}:E_i\rightarrow E_{i+1}$ such that
\begin{equation}
\mathrm{tr}_{A_{i+1}R_{i+1}}\Bar{\mathcal{M}}_{i+1}'=\mathcal{R}_{i+1}\circ\mathrm{tr}_{R_i}.\label{eqn:no_signaling}
\end{equation}
\end{lemma}
\begin{proof}
This is satisfied by setting

\begin{equation}
\mathcal{R}_{i+1}=\mathrm{tr}_{A_{i+1}R_{i+1}}\Gamma_{i+1}\circ\mathrm{tr}_{R_{i+1} R_{i+1}'}\circ\mathcal{M}_{i+1}^{'*}\circ\Gamma_i\circ\mathrm{tr}_{R_i'}.
\end{equation}
\end{proof}

Since the modified adversary satisfies the non-signalling condition, we can apply entropy accumulation theorem to bound the unpredictability of $z_i$ from the perspective of the guesser, which is bounded by $H_{\rm min}(A_1^n\vert E_n)$. To do this, we need a lower bound on the single-round von Neumann entropy. Namely, for $q\in\mathds{P}$ and $\mathds{P}$ a set of probability distributions on the alphabet $\mathcal{X}$ of $X_i$, we need
\begin{equation}
    \inf_{\nu\in\Sigma_i(q)}H(A_i\vert E_i\Tilde{E})_\nu\geq f(q),
\end{equation}
where $\Tilde{E}$ is isomorphic to $R_{i-1} E_{i-1}$, and
\begin{equation}
    \Bar{\Sigma}_i(q)=\left\{\nu_{X_i A_i R_i E_i \Tilde{E}}=\mathcal{T}_i\circ\Bar{\mathcal{M}}_i^{'}(\rho)\Big\vert \rho\in S(R_{i-1}E_{i-1}\Tilde{E})\wedge\nu_{X_i}=q\right\}.
\end{equation}
Here, $f(q)$ is called the min-tradeoff function, and it is an affine function.

However, $\Bar{\mathcal{M}}_i'$ are modified quantum channels and we do not have an entropy lower bound for them. Given some entropy lower bound on the original quantum channels $\mathcal{M}_i^{'*}$, we wish to establish a lower bound for $\Bar{\mathcal{M}}_i'$. We proceed with this task by first showing some useful lemmas about conditional entropy.

\begin{lemma}\label{lem:entropy_with_a_classical_register}
For state $\rho$ on quantum system $ABC$ where $C$ is classical and $\rho=\sum_c p(c)\vert c\rangle\langle c\vert_{C}\otimes\sigma^c_{A}\otimes\eta^c_{B}$,
\begin{equation}
    H(ABC)=H(C)+H(A|C)+H(B|C).
\end{equation}
\end{lemma}
\begin{proof}
    First, we note that $\rho$ is block diagonal. Further, each sub-block of $C=c$ is in a product form $\sigma^c_{A}\otimes\eta^c_{B}$. Say $\sigma^c_{A}$ has eigenvalues $\lambda^c_a$ and $\eta^c_{B}$ has eigenvalues $\lambda^c_b$. Note that the eigenvectors are not the same for different $c$ in general. Still, $\sigma^c_{A}\otimes\eta^c_{B}$ has eigenvalues $\lambda^c_a\lambda^c_b$ for all possible $a,b$. Overall, $\rho$ has eigenvalues $p(c)\lambda^c_a\lambda^c_b$ for all $a,b,c$. The entropy of the system is identical to a classical system with probability distribution $p(abc)=p(c)p(a|c)p(b|c)$, and the entropy is $H(ABC)=H(C)+H(A|C)+H(B|C)$.
\end{proof}

\begin{lemma}\label{lem:conditional_entropy_with_classical_register}
    For state $\rho$ on quantum system $ABC$ where $C$ is classical and $\rho=\sum_c p(c)\vert c\rangle\langle c\vert_{C}\otimes\sigma^c_{A}\otimes\eta^c_{B}$, we have $H(A\vert BC)=H(A\vert C)$.
\end{lemma}
\begin{proof}
    Using Lemma \ref{lem:entropy_with_a_classical_register},
    \begin{align}\nonumber
        H(A\vert BC)=&H(ABC)-H(BC)\\\nonumber
        =&\left(H(C)+H(A\vert C)+H(B\vert C)\right)-\left(H(C)+H(B\vert C)\right)=H(A\vert C).
    \end{align}
\end{proof}

This now allows us to reduce the entropy of the modified adversary conditioned on quantum registers (which we need to apply generalized entropy accumulation theorem) to entropy conditioned on classical registers only.

\begin{lemma}\label{lem:conditioned_on_quantum_equals_classical}
For input quantum state $\sigma\in S(R_{i-1} E_{i-1} \Tilde{E})$ and 
\begin{equation}
    \nu_{X_i A_i R_i E_i \Tilde{E}}=\left(\Gamma_i\circ\mathrm{tr}_{R_i R_i'}\circ\mathcal{T}_i\circ\mathcal{M}_i^{'*}\circ\Gamma_{i-1}\circ\mathrm{tr}_{R_{i-1} R_{i-1}'}\right)\sigma,
\end{equation}
we have $H(A_i\vert E_i\Tilde{E})_{\nu_{A_i E_i\Tilde{E}}}=H(A_i\vert C_1^i G_1^i)_{\nu_{A_i C_1^i G_1^i}}$.
\end{lemma}

\begin{proof}
We know that $\mathrm{tr}_{R_i R_i'}\sigma\in S(C_{i-1}G_{i-1}\Tilde{E})$. Further, 
\begin{align}\nonumber
    \nu_{A_i E_i \Tilde{E}_i}&=\mathrm{tr}_{X_i R_i}\nu_{X_i A_i R_i E_i \Tilde{E}}\\\nonumber
    =&\mathrm{tr}_{X_i R_i}\circ\Gamma_i\circ\mathrm{tr}_{R_i R_i'}\circ\mathcal{T}_i\circ\mathcal{M}_i^{'*}\circ\Gamma_{i-1}\circ\mathrm{tr}_{R_{i-1} R_{i-1}'}\sigma\\\nonumber
    =&\mathrm{tr}_{R_i}\circ\Gamma_i\circ\mathrm{tr}_{X_i R_i R_i'}\circ\mathcal{M}_i^{'*}\\\nonumber
    &\sum_{c_1^{i-1}g_1^{i-1}}p(c_1^{i-1}g_1^{i-1})\vert c_1^{i-1}g_1^{i-1}\rangle\langle c_1^{i-1}g_1^{i-1}\vert\otimes\rho_{R_{i-1}R_{i-1}'}^{*g_1^{i-1}c_1^{i-1}g_1^{i-1}}\otimes\sigma_{\Tilde{E}}^{c_1^{i-1}g_1^{i-1}}\\\nonumber
    =&\mathrm{tr}_{R_i}\circ\Gamma_i\circ\sum_{c_1^{i-1}g_1^{i-1}}p(c_1^{i-1}g_1^{i-1})\vert c_1^{i-1}g_1^{i-1}\rangle\langle c_1^{i-1}g_1^{i-1}\vert\otimes\eta^{c_1^{i-1}g_1^{i-1}}_{A_i C_i G_i}\otimes\sigma_{\Tilde{E}}^{c_1^{i-1}g_1^{i-1}}\\\nonumber
    =&\mathrm{tr}_{R_i}\sum_{c_1^i g_1^i}p(c_1^i g_1^i)\vert c_1^i g_1^i\rangle\langle   c_1^i g_1^i\vert\otimes\rho_{R_i R_i'}^{* g_1^i c_1^i g_1^i}\otimes\zeta_{A_i}^{c_1^i g_1^i}\otimes\sigma_{\Tilde{E}}^{c_1^{i-1}g_1^{i-1}}\\
    =&\sum_{c_1^i g_1^i}p(c_1^i g_1^i)\vert c_1^i g_1^i\rangle\langle c_1^i g_1^i\vert\otimes\left(\rho_{R_i'}^{c_1^i g_1^i}\otimes\sigma_{\Tilde{E}}^{c_1^{i-1}g_1^{i-1}}\right)\otimes\zeta_{A_i}^{c_1^i g_1^i},
\end{align}
where
\begin{align}
p(c_1^{i-1}g_1^{i-1})\sigma_{\Tilde{E}}^{c_1^{i-1}g_1^{i-1}}&=\langle c_1^{i-1}g_1^{i-1}\vert\mathrm{tr}_{R_{i-1} R_{i-1}'}\sigma\vert c_1^{i-1}g_1^{i-1}\rangle\in S(\Tilde{E}),\\
\mathrm{tr}_{\Tilde{E}}\sigma_{\Tilde{E}}^{c_1^{i-1}g_1^{i-1}}&=1,
\end{align}
\begin{equation}
\mathrm{tr}_{X_i R_i R_i'}\circ\mathcal{M}_i^{'*}\vert c_1^{i-1}g_1^{i-1}\rangle\langle c_1^{i-1}g_1^{i-1}\vert\rho_{R_{i-1}R_{i-1}'}^{*g_1^{i-1}c_1^{i-1}g_1^{i-1}}=\vert c_1^{i-1}g_1^{i-1}\rangle\langle c_1^{i-1}g_1^{i-1}\vert\eta^{c_1^{i-1}g_1^{i-1}}_{A_i C_i G_i},
\end{equation}
since $\mathrm{tr}_{X_i R_i R_i'}\circ \mathcal{M}_i^{'*}$ cannot change $C_1^i G_1^i$ (this is because $C_1^i G_1^i$ are committed classical variables that cannot be changed), and
\begin{align}
    p(c_1^i g_1^i)\zeta_{A_i}^{c_1^i g_1^i}&=p(c_1^{i-1} g_1^{i-1})\langle c_i g_i\vert\eta^{c_1^{i-1}g_1^{i-1}}_{A_i C_i G_i}\vert c_i g_i\rangle,\\
    \mathrm{tr}_{A_i}\zeta_{A_i}^{c_1^i g_1^i}&=1.
\end{align}
Finally, by Lemma \ref{lem:conditional_entropy_with_classical_register}, we have
\begin{equation}
    H(A_i\vert E_i\Tilde{E})_{\nu_{A_i E_i\Tilde{E}}}=H(A_i\vert C_1^i G_1^i)_{\nu_{A_i C_1^i G_1^i}}.
\end{equation}
\end{proof}

\begin{lemma}\label{lem:entropy_does_not_change}
    For $\Gamma_i$ defined in E.q. \ref{eqn:gamma_definition} and $\nu'=\left(\Gamma_i\circ\mathrm{tr}_{R_i R_i'}\right)\nu$, we have
    \begin{equation}
        H(A_i\vert C_1^i G_1^i)_{\nu}=H(A_i\vert C_1^i G_1^i)_{\nu'}.
    \end{equation}
\end{lemma}
\begin{proof}
    Since $\Gamma_i\circ\mathrm{tr}_{R_i R_i'}$ does not modify $A_i C_1^i G_1^i$, we have
    \begin{equation}
        \nu_{A_i C_1^i G_1^i}=\mathrm{tr}_{X_i R_i R_i'}\nu=\mathrm{tr}_{X_i R_i R_i'}\left(\Gamma_i\circ\mathrm{tr}_{R_i R_i'}\right)\nu=\nu'_{A_i C_1^i G_1^i},
    \end{equation}
    which means the two entropy quantities must be equal.
\end{proof}

We now show that if we have a lower bound on the entropy conditioned on $C_1^i G_1^i$ for the original optimal adversary, then we also have a lower bound on the entropy conditioned on quantum side information for the modified adversary.
\begin{lemma}\label{lem:modified_entropy_bound}
For the quantum channel defined in E.q. \ref{eqn:modified_channel}, if
\begin{equation}
    \inf_{\nu'\in\Sigma'_i(q)}H(A_i\vert C_1^i G_1^i)_{\nu'}\geq f(q),\label{eqn:original_entropy_bound}
\end{equation}
where
\begin{equation}
    \Sigma'_i(q)=\left\{\nu_{X_i A_i R_i E_i}=\mathcal{T}_i\circ\mathcal{M}_i^{'*}(\rho)\Big\vert \rho\in S(R_{i-1} E_{i-1})\wedge\nu_{X_i}=q\right\},
\end{equation}
$q\in\mathds{P}$, and $\mathds{P}$ is the set of probability distributions on the alphabet $\mathcal{X}$ of $X_i$, then
\begin{equation}
    \inf_{\nu\in\Bar{\Sigma}_i(q)}H(A_i\vert E_i\Tilde{E})_\nu\geq f(q),
\end{equation}
where $\Tilde{E}$ is isomorphic to $R_{i-1} E_{i-1}$, and
\begin{equation}
    \Bar{\Sigma}_i(q)=\left\{\nu_{X_i A_i R_i E_i \Tilde{E}}=\mathcal{T}_i\circ\Bar{\mathcal{M}}_i^{'}(\rho)\Big\vert \rho\in S(R_{i-1}E_{i-1}\Tilde{E})\wedge\nu_{X_i}=q\right\}.
\end{equation}
\end{lemma}
\begin{proof}
By Lemma \ref{lem:conditioned_on_quantum_equals_classical},
\begin{align}
    \inf_{\nu\in\Bar{\Sigma}_i(q)}H(A_i\vert E_i\Tilde{E})_\nu&=\inf_{\nu\in\Bar{\Sigma}_i(q)}H(A_i\vert C_1^i G_1^i)_\nu.\label{eqn:inf_conditioned_on_quantum_equals_classical}
\end{align}
Define $\rho^*\in S(R_{i-1}E_{i-1}\Tilde{E})$ and $\nu^*_{X_i A_i R_i E_i \Tilde{E}}=\mathcal{T}_i\circ\Bar{\mathcal{M}}_i^{'}(\rho^*)$ as anything that satisfies the following conditions.
First, $\nu^*_{X_i}=q$, which implies $\nu^*\in\Bar{\Sigma}_i(q)$. Second, for all $\rho\in S(R_{i-1}E_{i-1}\Tilde{E})$ and $\nu_{X_i A_i R_i E_i \Tilde{E}}=\mathcal{T}_i\circ\Bar{\mathcal{M}}_i^{'}(\rho)$ such that $\nu_{X_i}=q$, we have $H(A_i\vert C_1^i G_1^i)_{\nu^*}\leq H(A_i\vert C_1^i G_1^i)_\nu$. By definition, we have
\begin{equation}
    \inf_{\nu\in\Bar{\Sigma}_i(q)}H(A_i\vert C_1^i G_1^i)_\nu=H(A_i\vert C_1^i G_1^i)_{\nu^*},\label{eqn:inf_by_optimal_state}
\end{equation}
Further,
\begin{align}\nonumber
\nu^*_{X_i A_i R_i E_i}&=\mathrm{tr}_{\Tilde{E}}\nu^*_{X_i A_i R_i E_i \Tilde{E}}\\\nonumber
&=\left(\mathrm{tr}_{\Tilde{E}}\circ\Gamma_i\circ\mathrm{tr}_{R_i R_i'}\circ\mathcal{T}_i\circ\mathcal{M}_i^{'*}\circ\Gamma_{i-1}\circ\mathrm{tr}_{R_{i-1} R_{i-1}'}\right)\rho^*\\\nonumber
&=\left(\Gamma_i\circ\mathrm{tr}_{R_i R_i'}\right)\left(\mathcal{T}_i\circ\mathcal{M}_i^{'*}\right)\left[\left(\Gamma_{i-1}\circ\mathrm{tr}_{R_{i-1} R_{i-1}' \Tilde{E}}\right)\rho^*\right]\\
&=\left(\Gamma_i\circ\mathrm{tr}_{R_i R_i'}\right)\left(\mathcal{T}_i\circ\mathcal{M}_i^{'*}\right)\left(\rho'_{R_{i-1} E_{i-1}}\right)\\
&=\left(\Gamma_i\circ\mathrm{tr}_{R_i R_i'}\right)\mu^*_{X_i A_i R_i E_i}.\label{eqn:traced_state}
\end{align}
Therefore, by Lemma \ref{lem:entropy_does_not_change},
\begin{equation}
    H(A_i\vert C_1^i G_1^i)_{\nu^*}=H(A_i\vert C_1^i G_1^i)_{\left(\Gamma_i\circ\mathrm{tr}_{R_i R_i'}\right)\mu^*_{X_i A_i R_i E_i}}=H(A_i\vert C_1^i G_1^i)_{\mu^*}\label{eqn:entropy_no_change_in_proof}
\end{equation}

Since the quantum channel $\Gamma_i\circ\mathrm{tr}_{R_i R_i'}$ does not act on $X_i$ and $\nu^*\in\Bar{\Sigma}_i(q)$, we must have $\mu^*_{X_i}=\nu^*_{X_i}=1$. Further, $\mu^*_{X_i A_i R_i E_i}=\left(\mathcal{T}_i\circ\mathcal{M}_i^{'*}\right)\left(\rho'_{R_{i-1} E_{i-1}}\right)$ for some $\rho'\in S(R_{i-1}E_{i-1})$. Together, these two conditions means that $\mu^*_{X_i A_i R_i E_i}\in\Sigma'_i(q)$ by the definition of $\Sigma'_i(q)$. As a result,
\begin{equation}
    H(A_i\vert C_1^i G_1^i)_{\mu^*}\geq \inf_{\nu\in\Sigma'_i(q)}H(A_i\vert C_1^i G_1^i)_{\nu'}.\label{eqn:bounded_by_inf}
\end{equation}
Combining E.q. \ref{eqn:inf_conditioned_on_quantum_equals_classical}, \ref{eqn:inf_by_optimal_state}, \ref{eqn:entropy_no_change_in_proof}, \ref{eqn:bounded_by_inf}, and \ref{eqn:original_entropy_bound} yields
\begin{align}\nonumber
    \inf_{\nu\in\Bar{\Sigma}_i(q)}H(A_i\vert E_i\Tilde{E})_\nu&=\inf_{\nu\in\Bar{\Sigma}_i(q)}H(A_i\vert C_1^i G_1^i)_\nu=H(A_i\vert C_1^i G_1^i)_{\nu^*}\\
    &=H(A_i\vert C_1^i G_1^i)_{\mu^*}\geq \inf_{\nu\in\Sigma'_i(q)}H(A_i\vert C_1^i G_1^i)_{\nu'}\geq f(q).
\end{align}
\end{proof}

Now that we derived the bound on the required single round von Neumann entropy of the modified adversary, we can apply entropy accumulation to it.

\begin{lemma}\label{lem:sequential_entropy}
For $\Sigma'_i(q)$ defined in Lemma \ref{lem:modified_entropy_bound}, $\Bar{\mathcal{M}}_i'$ defined in E.q. \ref{eqn:modified_channel}, $\Bar{\mathcal{M}}_i=\mathcal{T}_i\circ\Bar{\mathcal{M}}_i'$, $\omega'$ a set of possible outputs on $X_1^n$, if
\begin{equation}
    \inf_{\nu'\in\Sigma'_i(q)}H(A_i\vert C_1^i G_1^i)_{\nu'}\geq f(q),
\end{equation}
then,
\begin{equation}
    H_{\rm min}^{\varepsilon}(A_1^n\vert E_n)_{\Bar{\mathcal{M}}_n\circ\cdots\circ\Bar{\mathcal{M}}_1}(\rho_{R_0 E_0})_{\vert\omega'}\geq nh-c_1\sqrt{n}-c_0,
\end{equation}
where
\begin{align}
    h&=\min_{x_1^n\in\omega'}f\left(\mathsf{freq}(x_1^n)\right)\label{eqn:h}\\
    \mathsf{freq}(x_1^n)(x)&=\frac{\vert \{i\in\{1,\dots,n\}:x_i=x\}\vert}{n},
\end{align}
and $c_0, c_1$ are as defined in Corollary 4.6 in \cite{MFSR22}, correcting the typo where $g(\varepsilon)$ should be
\begin{equation}
    g(\varepsilon)=-\log_2(1-\sqrt{1-\varepsilon^2})
\end{equation}
and replacing $\Pr[\Omega]$ with $\Pr[\omega']$.
\end{lemma}
\begin{proof}
    By Lemma \ref{lem:modified_entropy_bound}, for $\Bar{\Sigma}_i(q)$ defined in Lemma \ref{lem:modified_entropy_bound},
    \begin{equation}
        \inf_{\nu\in\Bar{\Sigma}_i(q)}H(A_i\vert E_i\Tilde{E})_\nu\geq f(q).
    \end{equation}
    Applying Corollary 4.6 of \cite{MFSR22} completes the proof, which is valid because of Lemma \ref{lem:no_signaling}.
\end{proof}

\subsection{Proof of the Sequential Decomposition Property}\label{sec:proof_sequential_decomposition}

Although EAT bounds the $\varepsilon$-min-entropy instead of the min-entropy, we can use the following relation.
\begin{lemma}\label{lem:min_entropy_from_smooth_min_entropy}
    For a classical probability distribution $\rho$ over random variable $X$,
    \begin{equation}
        H_{\rm min}(X)\geq -\log_2\left(\varepsilon+2^{-H_{\rm min}^\varepsilon}\right).
    \end{equation}
\end{lemma}
\begin{proof}
    Consider $x^*_\rho=\mathrm{argmax}_x \Pr[X=x]_\rho$ and $p_{\mathrm{max}, \rho}=\Pr[X=x^*_\rho]_\rho$. By definition of the min-entropy, $p_{\mathrm{max}, \rho}=2^{-H_{\rm min}(X)_{\rho}}$. Consider a classical distribution $\rho'$ in the $\varepsilon$-ball of $\rho$. Maximization of the min-entropy over such a classical state gives the smooth min-entropy due to the definition of the smooth min-entropy in Definition 6.9 and Lemma 6.13 of of \cite{tomamichel2015quantum}. For classical distributions, the TVD distance between $\rho$ and $\rho'$ must be less than $\varepsilon$. Define $x^*_{\rho'}=\mathrm{argmax}_x\Pr[X=x]_{\rho'}$ and $p_{\mathrm{max},\rho'}=\Pr[X=x^*_{\rho'}]_{\rho'}$.
    
    Either $x^*_\rho=x^*_{\rho'}$ or $x^*_\rho\neq x^*_{\rho'}$. If $x^*_\rho=x^*_{\rho'}$, we must have $\varepsilon\geq\mathrm{TVD}[\rho, \rho']\geq\vert p_{\mathrm{max}, \rho}-p_{\mathrm{max}, \rho'}\vert$, and therefore $p_{\mathrm{max},\rho}\leq p_{\mathrm{max}, \rho'}+\varepsilon$.

    If $x^*_\rho\neq x^*_{\rho'}$, we must have $\varepsilon\geq\mathrm{TVD}[\rho, \rho']\geq\vert p_{\mathrm{max}, \rho}-\Pr[X=x^*_{\rho}]_{\rho'}\vert$, and therefore $p_{\mathrm{max},\rho}\leq \Pr[X=x^*_{\rho}]_{\rho'}+\varepsilon$. Note that $\Pr[X=x^*_{\rho}]_{\rho'}$ is the probability of measuring from $\rho'$ the maximum probability string of $\rho$. Since $x^*_\rho\neq x^*_{\rho'}$, we have $\Pr[X=x^*_{\rho}]_{\rho'}\leq p_{\mathrm{max}, \rho'}$ by definition of $p_{\mathrm{max}, \rho'}$, and $p_{\mathrm{max},\rho}\leq p_{\mathrm{max}, \rho'}+\varepsilon$.

    Since in both cases, $p_{\mathrm{max},\rho}\leq p_{\mathrm{max}, \rho'}+\varepsilon$, we have
    \begin{equation}
        H_{\rm min}(X)_{\rho}=-\log_2 p_{\mathrm{max},\rho}\geq -\log_2\left(p_{\mathrm{max}, \rho'}+\varepsilon\right).
    \end{equation}
    Finally, since by definition of the smooth min-entropy in Definition 6.9 of \cite{tomamichel2015quantum}, $H_{\rm min}^\varepsilon(X)_\rho\geq H_{\rm min}(X)_{\rho'}=-\log_2 p_{\mathrm{max}, \rho'}$,
    \begin{equation}
    H_{\rm min}(X)_{\rho}\geq-\log_2\left(2^{-H_{\rm min}^\varepsilon(X)_{\rho}}+\varepsilon\right).
    \end{equation}
\end{proof}

However, this lemma is only for unconditional min-entropy, which does not directly apply. We now use similar arguments to show this for conditional min-entropy.

\begin{lemma}\label{lem:conditional_min_entropy_from_smooth_conditional_min_entropy}
    For a classical probability distribution $\rho$ over random variables $X,Y$,
    \begin{equation}
        H_{\rm min}(X\vert Y)_\rho\geq -\log_2\left(2^{-H_{\rm min}^\varepsilon(X\vert Y)_\rho}+\epsilon\right).
    \end{equation}
\end{lemma}
\begin{proof}
    We closely follow the same arguments presented in the proof of Lemma \ref{lem:min_entropy_from_smooth_min_entropy}. Consider $x^*_{\rho_{X\vert Y=y}}=\mathrm{argmax}_x\Pr[X=x]_{\rho_{X\vert Y=y}}$, where $\rho_{X\vert Y=y}=\langle y\vert \rho \vert y\rangle$. Define $p_{\mathrm{max},\rho_{X\vert Y=y}}=\Pr\left[X=x^*_{\rho_{X\vert Y=y}}\right]_{\rho_{X\vert Y=y}}$. Consider a distribution $\rho'$ in the $\varepsilon$-ball of $\rho$. Maximization of the min-entropy over such a classical state gives the smooth min-entropy due to the definition of the smooth min-entropy in Definition 6.9 and Lemma 6.13 of of \cite{tomamichel2015quantum}. Similarly, we define $x^*_{\rho'_{X\vert Y=y}}$ and $p_{\mathrm{max},\rho'_{X\vert Y=y}}$.
    
    We must have
    \begin{align}
        \varepsilon\geq&\mathrm{TVD}[\rho,\rho']\\
        \geq&\sum_y \bigg\vert \Pr[Y=y]_\rho \cdot p_{\mathrm{max},\rho_{X\vert Y=y}}\\
        &-\Pr[Y=y]_{\rho'}\cdot\Pr\left[X=x^*_{\rho_{X\vert Y=y}}\right]_{\rho'_{X\vert Y=y}}\bigg\vert\\
        \geq& \bigg\vert \sum_y \Pr[Y=y]_\rho \cdot p_{\mathrm{max},\rho_{X\vert Y=y}}\\
        &-\sum_y \Pr[Y=y]_{\rho'}\cdot\Pr\left[X=x^*_{\rho_{X\vert Y=y}}\right]_{\rho'_{X\vert Y=y}}\bigg\vert.
    \end{align}
    Since $\Pr\left[X=x^*_{\rho_{X\vert Y=y}}\right]_{\rho'_{X\vert Y=y}}\leq p_{\mathrm{max}, \rho'_{X\vert Y=y}}$ by definition of $p_{\mathrm{max}, \rho'_{X\vert Y=y}}$, we have
    \begin{align}
        \sum_y \Pr[Y=y]_\rho \cdot p_{\mathrm{max},\rho_{X\vert Y=y}}\leq \sum_y \Pr[Y=y]_{\rho'}\cdot p_{\mathrm{max}, \rho'_{X\vert Y=y}}+\epsilon.
    \end{align}
    Therefore, we have
    \begin{align}
        H_{\rm min}(X\vert Y)_\rho&=-\log_2\left(\sum_y \Pr[Y=y]_\rho \cdot p_{\mathrm{max},\rho_{X\vert Y=y}}\right)\\
        &\geq-\log_2\left(\sum_y \Pr[Y=y]_{\rho'}\cdot p_{\mathrm{max}, \rho'_{X\vert Y=y}}+\epsilon\right),
    \end{align}
    where the definition of conditional min-entropy for classical distributions follow E.q. 6.26 of \cite{tomamichel2015quantum}. Finally, since by definition of the smooth min-entropy in Definition 6.9 of \cite{tomamichel2015quantum},
    \begin{equation}
        H_{\rm min}^\varepsilon(X\vert Y)_\rho\geq H_{\rm min}(X\vert Y)_{\rho'}=-\log_2\left(\sum_y \Pr[Y=y]_{\rho'}\cdot p_{\mathrm{max}, \rho'_{X\vert Y=y}}\right),
    \end{equation}
    we have
    \begin{equation}
        H_{\rm min}(X\vert Y)_\rho\geq -\log_2\left(2^{-H_{\rm min}^\varepsilon(X\vert Y)_\rho}+\epsilon\right).
    \end{equation}
\end{proof}
Further, for a CVPV protocol based on certified randomness with acceptance probability $p$, the CVPV acceptance probability is 
\begin{equation}
    \Pr[\Omega]\leq\min\left(p, 2^{-H_{\rm min}}\right).
\end{equation}
This is because the guessing probability is given by the exponential of the conditional min-entropy for classical variables as shown in E.q. 6.27 of \cite{tomamichel2015quantum}.

To more explicitly show asymptotic soundness, we prove the following theorem.

\begin{lemma}\label{lem:modified_probability}
    For the modified quantum channel defined in E.q. \ref{eqn:modified_channel} and $h$ defined in E.q. \ref{eqn:h}, if $h>0$, then the $n$-round protocol has $\Pr[\Omega]\leq O(2^{-n})$.
\end{lemma}
\begin{proof}
    Lemma \ref{lem:sequential_entropy} shows that to achieve $H_{\rm min}^\varepsilon=O(n)$, we can tolerate $c_0=O(n)$ and $c_1=O(\sqrt{n})$. To achieve this, we can tolerate $g(\varepsilon)=O(n)$ and $\Pr[\omega']=O(n)$ with suitably chosen constants such that $c_0$ and $c_1$ are sufficiently small and $H_{\rm min}^\varepsilon > 0$. To achieve this, we can have $\Pr[\omega']=O(2^{-n})$ and $\varepsilon=O(2^{-n})$. In this case, 
    \begin{align}
    H_{\rm min}&\geq-\log_2\left(\varepsilon + 2^{-H_{\rm min}^\varepsilon}\right)=O(n)\\
    \Pr[\Omega]&\leq\min\left(\Pr[\omega'], 2^{-H_{\rm min}}\right)=\min(O(2^{-n}), O(2^{-n}))=O(2^{-n}).
    \end{align}
\end{proof}

Finally, we have a bound on the probability of the protocol not aborting in the adversarial setting for soundness. According to Definition \ref{def:cr_multiround}, $P$ and $\eve$ can communicate and setup arbitrarily only after $\ans_i, \ans_i'$ are provided. Therefore, we can model the $i$th round CVPV protocol quantum channel as
\begin{equation}\mathcal{M}_i^*=\mathcal{T}_i\circ\mathcal{N}_i\circ\left(\mathcal{P}_i'\circ\mathcal{G}_i\right)\circ\left(\mathcal{C}_i \otimes\mathcal{I}\right),\label{eqn:channel_into_guesser_prover}
\end{equation}
where $\mathcal{I}$ is identity over $R_{i-1}E_{i-1}$, $\mathcal{C}_i:\mathds{C}\rightarrow C_i$ is a channel from complex number to the challenge, $\mathcal{G}_i:C_i E_{i-1}\rightarrow E_i$ where $E_i=C_1^i G_1^i R_i'$, $\mathcal{P}_i':C_i R_{i-1}\rightarrow A_i R_i$, $\mathcal{N}_i:R_i E_i\rightarrow R_i E_i$ is the arbitrary communication and setup channel, and $\mathcal{T}_i:A_i C_i\rightarrow X_i$ is the test channel.

\begin{theorem}\label{thm:original_probability}
    Given a PoQ protocol for $h$ defined in \ref{eqn:h}, if $h>0$ for all single-round QPT channels $\mathcal{M}_i^*$ of the form of E.q. \ref{eqn:channel_into_guesser_prover}, then the protocol is a sequential certified randomness protocol.
\end{theorem}
\begin{proof}
    For the modified quantum channel, if $h>0$ is satisfied, then the $n$-round protocol has $\Pr[\Omega]\leq O(2^{-n})$ due to Lemma \ref{lem:modified_probability}. Finally, by Lemma \ref{lem:equal_prob}, for the original channel, the acceptance probability is identical. 
\end{proof}

Additionally, to show that any certified randomness from repetition satisfying Definition \ref{def:crea} is a sequential certified randomness protocol, we need to show $h>0$ for channel $\mathcal{M}_i^*$ is implied by $h>0$ channel $\mathcal{P}_i$.

\begin{theorem}
    An $\ell$-round PoQ protocol $\poq$ is a sequential certified randomness protocol if it is a certified randomness from repetition, and performs the consistency check and timing check of CVPV. 
\end{theorem}
\begin{proof}
    Consider $\mathcal{M}_i^*$ defined in E.q. \ref{eqn:channel_into_guesser_prover}. Since $\mathcal{N}_i$ does not change $A_i C_1^i G_1^i$, it does not affect $H(A_i\vert C_1^i G_1^i)$. Therefore, we have
    \begin{align}
        H(A_i\vert C_1^i G_1^i)_{\mathcal{M}_i^*(\rho_{R_{i-1}E_{i-1}})}&= H(A_i\vert C_1^i G_1^i)_{\mathcal{T}_i\circ\left(\mathcal{P}_i'\otimes\mathcal{G}_i\right)\circ\left(\mathcal{C}_i \otimes\mathcal{I}\right)(\rho_{R_{i-1}E_{i-1}})}\\
        &\geq H(A_i\vert E_i)_{\mathcal{T}_i\circ\left(\mathcal{P}_i'\otimes\mathcal{G}_i\right)\circ\left(\mathcal{C}_i \otimes\mathcal{I}\right)(\rho_{R_{i-1}E_{i-1}})}.
    \end{align}
    Further, since $\mathcal{G}_i$ does not act on $A_i$,
    \begin{equation}
        H(A_i\vert E_i)_{\mathcal{T}_i\circ\left(\mathcal{P}_i'\otimes\mathcal{G}_i\right)\circ\left(\mathcal{C}_i \otimes\mathcal{I}\right)(\rho_{R_{i-1}E_{i-1}})}\geq H(A_i\vert C_i E_{i-1})_{\mathcal{T}_i\circ\mathcal{P}_i'\circ\left(\mathcal{C}_i \otimes\mathcal{I}\right)(\rho_{R_{i-1}E_{i-1}})}.
    \end{equation}
    Moreover, $\mathcal{T}_i\circ\mathcal{P}_i'\circ\left(\mathcal{C}_i \otimes\mathcal{I}\right)$ is exactly the $i$th round certified randomness from repetition quantum channel $\mathcal{P}_i$. This is because for the challenges to be generated by the verifier independent of any other information, which is required by Definition \ref{def:crea}, the channel must have this form. Therefore, 
    \begin{align}
        H(A_i\vert C_1^i G_1^i)_{\mathcal{M}_i^*(\rho_{R_{i-1}E_{i-1}})}=H(A_i\vert C_i E_{i-1})_{\mathcal{P}_i(\rho_{R_{i-1}E_{i-1}})},
    \end{align}
    and therefore
    \begin{equation}
    \inf_{\nu'\in\Sigma'_i(q)}H(A_i\vert C_1^i G_1^i)_{\nu'}\geq \inf_{\nu\in\Sigma_i(q)}H(A_i\vert C_i E_{i-1})_{\nu}\geq f(q),
    \end{equation}
    where $\Sigma'_i(q)$ is defined in Lemma \ref{lem:modified_entropy_bound} and $\Sigma_i(q),f(q)$ are as in Definition \ref{def:crea}.
\end{proof}

\begin{corollary}
    All results in Section \ref{sec:mr_seq} hold for PoQ scheme $\poq = (V,P)$ that is certified randomness from repetition.
\end{corollary}

%% file: Files/instantiation.tex
\section{Instantiations} \label{sec:inst_all}

We give several instantiations for the certified randomness scheme used by our CVPV compiler. Our main instantiation (\Cref{sec:inst}) is based on random circuit sampling, and utilizes our technical work in \Cref{sec:seq_crea}. 
\par In addition, we give instantiations using error-correcting codes (\Cref{sec:inst_YZ}) and noisy trapdoor claw-free functions (\Cref{sec:inst_lwe}). These two sections do not require extra technical work, and they should mainly serve to demonstrate the generality of our compilers.

\input{Files/instantiation_rcs}
\input{Files/instantiation_YZ}

\input{Files/instantiation_lwe}

%% file: Files/instantiation_rcs.tex
\subsection{Instantiation Using Random circuit sampling} \label{sec:inst}

For a specific instantiation of a CVPV protocol, we consider certified randomness from random circuit sampling (RCS) \cite{AH23}. In particular, it is appealing for near-term implementation due to the fact that RCS is already demonstrated experimentally and classical simulation is believed to be hard. Crucially, the protocol is based on solving the heavy output generation problem.

\begin{definition}[Heavy Output Generation] \label{def:xhog}
    A quantum algorithm $\mathcal{A}$ given $C\sim\distr$, where $\distr$ is some distribution of $\log_2 N$-qubit quantum circuits, is said to solve $b$-$\xhog$ if it outputs a bitstring $z$ such that
    \begin{align*}
    \underset{C\sim\distr}{\mathbb{E}}\left[\underset{z\sim\mathcal{A}^C}{\mathbb{E}}\left[p_C(z)\right]\right]\geq\frac{b}{N},
    \end{align*}
    where $p_C(z)=\vert\langle 0\vert C\vert z\rangle\vert^2$.
\end{definition}

\par \cite{AH23} showed that any quantum algorithm given oracle access to $C$ and passes $b$-$\xhog$ must output samples with conditional von Neumann entropy at least $\Omega(n)$ (Theorem \ref{thm:entropy_oracle}).

\begin{theorem}[Corollary 7.16 of \cite{AH23}]\label{thm:entropy_oracle}
    Consider a $T$-query adversary solving $(1+\delta)$-$\xhog$, $A$ is a length-$n$ bitstring, and $C$ is an $n$-qubit quantum circuit from the Haar measure. For $T=\poly(n), \delta=\Omega(1)$, and $\eta\in(0,1]$, we have
    \begin{equation}
         H(A|CE)_\psi\geq (1-\eta)\delta n-O(\log n),
    \end{equation}
    where $\psi$ is a quantum state $N^{-\Omega(\delta\eta)}$-close to the output of the adversary.
\end{theorem}

However, since the test condition requires passing $\xhog$, exponentially costly classical computation of the probability amplitudes of the received bitstrings is required. As a result, a drawback of this protocol is the asymmetry in the computational power of the adversary and the verifier, namely the adversary is polynomial-time and the verifier can perform exponentially expensive verification.

The polynomial-time constraint on the verifier applies explicitly for the analysis in Section 5 of \cite{AH23} where the Long List Hardness Assumption, a new computational hardness assumption, is used. We do not base our analysis on this model due to reasons we will discuss below. The polynomial-time constraint is more subtly incorporated in the analysis in Section 7 of \cite{AH23}, where the adversary is given oracle access to the challenge circuits. This model aims to capture the situation where the explicit description of the circuit cannot help the adversary in any way other than allowing the direct execution of the circuit, which also precludes the possibility of computing probability amplitudes. This is the adversary model we consider for our instantiation of the CVPV protocol, which is why we reproduce Corollary 7.16 of \cite{AH23} as Theorem \ref{thm:entropy_oracle} here.

The redeeming feature of the protocol is that the verifier only needs to perform spot checking very infrequently. The verifier enjoys significantly more time to perform these exponentially expensive calculations for each circuit. There are also modified protocols and modified adversary models tailored for practical implementation \cite{jpmc_cr}, which we do not discuss here.





\subsubsection{Protocols}

The above single round result allows one to use the entropy accumulation theorem to define a multi-round protocols that outputs certified smooth min-entropy for each case. Moreover, for randomness expansion, fresh randomness is only consumed on logarithmically many rounds to generate fresh challenge circuits.

We see that the random oracle model allows one to obtain a bound on the von Neumann entropy with quantum side information, and the respective RCS-based certified randomness protocols are certified randomness from entropy accumulation protocols satisfying Definition \ref{def:crea}. We have omitted discussions of a general device in Section 5 of \cite{AH23} since does not consider quantum side information, and Definition \ref{def:crea} is not satisfied. We have similarly omitted discussions of a semi-honest device in Section 6 of \cite{AH23} since it only serves as inspiration for a more general adversary in Section 7.

For an instantiation based on Theorem \ref{thm:entropy_oracle}, we use following protocol in Fig. \ref{fig:oracle_protocol}, which uses the construction of the sequential compiler. However, the protocol in Fig. \ref{fig:oracle_protocol} is very different from the protocol in Fig. 4 of \cite{AH23}. Specifically, we do not reuse the circuit and therefore do not perform the test by summing the scores over epochs. The main reason for this choice is that we do not believe the application of entropy accumulation in \cite{AH23} is correct, which we discuss in the supplement Section \ref{sec:protocol_issues}.

From this protocol, the prover is given explicit description of the challenge circuits, not just oracle access, which is inevitable for protocols based on classical communication. Once again, given limitations on the prover's computational power, the oracle model relies on the hope that the prover cannot do better than using the circuits as oracles even with explicit descriptions.



  

\begin{figure}[!ht]
    \hrule\vspace{.5em}

  Input: the qubit count $n$, the number of rounds $\ell=\poly(\lambda)$, the score parameter $\delta\in[0,1]$ and the fraction of test rounds $\gamma=O((\log n)/\ell)$. Additionally, we require a quantum circuit ansatz over $n$ qubits defined as a function $C:\{0,1\}^{\poly(n)}\rightarrow U(n)$, and family of crytographic hash function $\{G_k\}_{k\in\{0,1\}^\lambda}:\{0,1\}^m\rightarrow \{0,1\}^{\poly(n)}$.\\

  Protocol:

\begin{enumerate}
    \item At time $t=-\infty$, the verifiers sample a random hash key $k\leftarrow\{0,1\}^{\lambda}$. For $i\in[\ell]$, they sample random inputs $x_i, y_i\in\{0,1\}^m$, a random challenge $\ch_i\in\{0,1\}^{\poly(n)}$, and $T_i\sim\mathsf{Bernoulli}(\gamma)$. They computes $s_i=G_k(x_i\oplus y_i)\oplus \ch_i$ and $C_i=C(\ch_i)$. They publish the hash key $k$.

    \item For $i\in[\ell]$:
  \begin{enumerate}
    \item At time $t=i-1$, $V_0$ sends $(x_i, s_i)$ and expects an answer $\ans_i$ at time $t=i$.
    \item Similarly, at time $t=i-1$, $V_1$ sends $y_i$ and expects an answer $\ans'_i$ at time $t=i$.
    \item At time $t=i-1/2$, the honest prove, located at position $0.5$, computes $\ch_i=G_k(x_i\oplus y_i)\oplus s_i$ and $C_i=C(\ch_i)$, samples bitstrings $\ans_i=A_i$ from circuit $C_i$, and immediately sends $A_i$ to both verifiers.
  \end{enumerate}

  \item  $V$ accepts iff $\ans_i=\ans'_i$ for all $i\in[\ell]$, and $\frac{1}{t}\sum_{i:T_i=1} p_{C_i}(A_i)\geq(1+\delta)/N$ where $t=\vert\{i:T_i=1\}\vert$.
\end{enumerate}

  \vspace{.5em}
  \hrule
  \caption{CVPV protocol against a $T$-query adversary.}
\label{fig:oracle_protocol}
  \end{figure}

We note that the original entropy accumulation theorem does not apply since Theorem \ref{thm:entropy_oracle} give entropy lower bounds for some scores, and the score is averaged over all rounds. This is in contract to bounding the entropy given some probability distribution. We discuss how to address this issue in the supplement Section \ref{sec:issues_eat}.

\input{Files/issues}

%% file: Files/issues.tex
\subsubsection{Issues with the Protocol in Aaronson and Hung (STOC 23)}\label{sec:protocol_issues}

For the protocol in Fig. 4 of \cite{AH23}, the $i$th round quantum channel $\mathcal{M}_i$ is the joint system of $V$ and $P$. Moreover, $V$ takes the previous round circuit $C_{i-1}$ as one of the inputs and set $C_i=C_{i-1}$ if $T_i=0$ or $C_i\sim\mathrm{Haar}(N)$ otherwise. Similarly, $P$ also takes $C_{i-1}$ as one of the inputs along with a quantum memory $R_{i-1}$. For this channel, the single round entropy $H(A_i\vert C_i E)$ where $P$ takes on input a quantum state over $R_{i-1}E$ is not given by Theorem \ref{thm:entropy_oracle}, since $C_i\sim\mathrm{Haar}(N)$ is required for Theorem \ref{thm:entropy_oracle}. 
\begin{equation}
    H(A_i\vert C_i E)=(1-\gamma)H(A_i\vert C_i ET_i=0) + \gamma H(A_i\vert C_i E T_i=1)
\end{equation}

To see this more explicitly, consider the case where $P$ takes a classical state as input in memory $R_{i-1}$. Let the classical state be an output of an honest prover with input $C_{i-1}$. For a $T$-query $P$, $P$ is allowed to simply output this classical state if $C_i=C_{i-1}$. The conditional entropy $H(A_i\vert R_{i-1})$ in this case is zero. Therefore, there should not be entropy accumulation over any rounds with $T_i=0$.

We now describe mathematically where this breaks:
\begin{align*}
    H(A_{i}|C_{i}T_{i}E)_{\nu} &= \sum_{C} \Pr_{T_{i}}\sbrac{C_{i} = C} H(A_{i}|T_{i}E, C_{i}=C)_{\nu} \\
                             &= \sum_{C} \Pr_{T_{i}}\sbrac{C_{i} = C} H(A_{i}|E, C_{i}=C)_{\nu} \\
                             &= \gamma \sum_{C} h(C) H(A_{i} | E, C_{i} = C)_{\nu} + (1-\gamma) H(A_{i} | E, C_{i} = C_{i-1})_{\nu} \\
                             &= \gamma H(A_{i}|C_{i}E)_{\nu} + (1-\gamma)H(A_{i}|E, C_{i} = C_{i-1})_{\nu},
\end{align*}
where $h(C)$ is the probability of sampling $C$ for a random challenge, and the second equality follows from the fact that once $C$ is fixed, $A_{i}$ and $T_{i}$ are independent. To bound the entropy independent of a $\gamma$ scaling factor, we need to be able to bound single-round entropy (with side information) for other distributions. In our case, we need to bound single-round entropy for point distributions which is clearly impossible.

One may argue that this type of argument could be used against certified randomness based on post-quantum secure trapdoor claw-free functions \cite{BCMVV21}. Indeed, \cite{BCMVV21} explicitly discusses this issue that entropy accumulation theorem requires that single-round entropy bound for \textit{all} possible input states, including those that are computationally inefficient strings. As a result, \cite{BCMVV21} presents significant additional analysis to show that entropy accumulates in the protocol, and nontrivial work was presented in \cite{merkulov2023entropy} to use the entropy accumulation theorem.

It is plausible that similar techniques may be applied to the analysis of \cite{AH23} to the randomness expansion protocol, but we do not consider it here. In our setting, we do not require randomness expansion, and $C_i$ can be sampled each round. This avoids the complications due to circuit reuse, and the entropy accumulation theorem can be directly applied to Theorem \ref{thm:entropy_oracle}.

\subsubsection{Issues with Entropy Accumulation}\label{sec:issues_eat}
The single-round entropy lower bound in Theorem \ref{thm:entropy_oracle} is conditioned on the output achieving some score (e.g. $b$ in $b$-$\xhog$), and the min-tradeoff functions used in \cite{AH23} are defined for continuous scores. As a result, \cite{AH23} developed a modified entropy accumulation theorem, which requires the Markov chain condition. However, for the arguments in Section \ref{sec:seq_crea}, we prove that the modified adversary satisfies the non-signalling condition which allows us to apply generalized entropy accumulation. To apply the entropy accumulation theorem in \cite{AH23} to the modified adversary, we have to show that the modified adversary also satisfies the Markov chain condition.

Consider the modified adversary defined in E.q. \ref{eqn:modified_channel}. For the $i$th round channel output, we relabel $C_1^i G_1^i$ as $I_i$, and we have $\mathcal{M}'_{i+1}:R_iE_i\rightarrow A_{i+1}R_{i+1}R'_{i+1}I_{i+1}$. Further, copy the classical values of $I_i$ into another register $I_i'$ and send $E_i=I_i'R_i'$ as input to the guesser next round. Formally, the new quantum channel becomes
\begin{equation}
    \mathcal{M}_{i+1}'':R_i E_i\rightarrow A_{i+1} R_{i+1} E_{i+1} I_{i+1}=\Lambda_{i+1}\mathcal{M}_{i+1}'\label{eqn:markov_channel},
\end{equation}
where $\Lambda_i:I_i\rightarrow I_i I_i'$ is the classical channel that copies $I_i$ into $I_i'$.

It should be noted that $I_1^n$ has $n-i$ copies of $C_i G_i$. Nevertheless, we have $H_{\rm min}^\varepsilon(A_1^n\vert C_1^n G_1^n)=H_{\rm min}^\varepsilon(A_1^n\vert I_1^n)$, and therefore bounding $H_{\rm min}^\varepsilon(A_1^n\vert I_1^n)$ is sufficient for the protocol soundness.

\begin{lemma}
    For the quantum channel defined in E.q. \ref{eqn:markov_channel}, for some $\mathcal{R}'$, the Markov chain condition $A_i\leftrightarrow I_i\leftrightarrow I_{i+1}$ is satisfied:
    \begin{equation}
        \rho_{A_i I_i I_{i+1}}=\mathcal{I}_{A_i}\otimes \mathcal{R}'_{I_i I_{i+1}\leftarrow I_i}(\rho_{A_i I_i}).
    \end{equation}
\end{lemma}
\begin{proof}
    \begin{align}
        \rho_{A_i I_i I_{i+1}}&=\left(\mathrm{tr}_{A_{i+1}R_{i+1}E_{i+1}}\right)\rho_{A_{i+1}A_i R_{i+1} E_{i+1} I_{i+1} I_i}\\
        &=\left(\mathrm{tr}_{A_{i+1}R_{i+1}R'_{i+1}I'_{i+1}}\Lambda_{i+1}\circ\mathcal{M}'_{i+1}\circ\Lambda_i\right)\rho_{A_i R_i R'_i I_i}\\
        &=\left(\mathrm{tr}_{R'_{i+1}}\left(\mathrm{tr}_{A_{i+1}R_{i+1}}\mathcal{M}'_{i+1}\right)\circ\Lambda_i\right)\rho_{A_i R_i R'_i I_i}\\
        &=\left(\mathrm{tr}_{R'_i-1}\mathcal{R}_{E_{i+1}\leftarrow E_i}\circ \mathrm{tr}_{R_i}\circ\Lambda_i\right)\rho_{A_i R_i R'_i I_i}\\
        &=\left(\mathrm{tr}_{R'_{i+1}}\mathcal{R}_{E_{i+1}\leftarrow I_i'}\circ \mathrm{tr}_{R_i'}\circ\Lambda_i\right)\rho_{A_i R'_i I_i}\\
        &=\left(\mathrm{tr}_{R'_{i+1}}\mathcal{R}_{E_{i+1}\leftarrow I_i'}\circ\Lambda_i\right)\rho_{A_i I_i},
    \end{align}
    where the third equality holds because $\mathrm{tr}_{I'_{i+1}}\Lambda_{i+1}=\mathcal{I}_{I_{i+1}}$, the forth equality holds due to Lemma \ref{lem:no_signaling}, and the fifth equality holds for suitably defined $\mathcal{R}_{E_{i+1}\leftarrow I'_i}$ due to E.q. \ref{eqn:no_signaling}.
\end{proof}

As a result, we can bound $H_{\rm min}^\varepsilon(A_1^n\vert C_1^n G_1^n)$ by bounding $H_{\rm min}^\varepsilon(A_1^n\vert I_1^n)$ instead with the entropy accumulation theorem of \cite{AH23}. We can expand the definition of certified randomness from entropy accumulation to the case where the min-tradeoff function is for a score. To do this, we simply need to change $q$ for $\Sigma_i(q)$ into a score $s$, and define $h$ as $h=\min_s f(s)$. With this expanded definition, the protocols in Fig. \ref{fig:oracle_protocol} are certified randomness from repetition due to Theorem \ref{thm:entropy_oracle}, and they are therefore CVPV protocols.

Finally, we note that the choices of the min-tradeoff function in the application of entropy accumulation theorem of \cite{AH23} are incorrect. Specifically, when using the entropy accumulation theorem in Section 4 of \cite{AH23}, $f(\delta)\rightarrow f(\delta/\gamma)$ and $\delta\rightarrow\gamma\delta$ should be applied for the general adversary without side information due to the fact that the test channel performs the test with probability $\gamma$ only. The analysis of our CVPV protocols are immune from this issue as we effectively have $\gamma=1$ since every round is tested. As for the semi-honest and general $T$-query adversary, similar changes should be adopted, but we leave rigorous analysis regarding this to future work.

The change $f(\delta)\rightarrow f(\delta/\gamma)$ and $\delta\rightarrow\gamma\delta$ has no effect on the linear term of the accumulated entropy, but it makes the correction term more significant due to the change in $\|\nabla f\|_\infty$. Otherwise, the current theorems in \cite{AH23} of accumulated entropy gives entropy completely independent of the test probability $\gamma$. This is implausible since higher $\gamma$ should lead to lower acceptance probability at fixed entropy (fixed adversary), which should increase the entropy if the acceptance probability should be fixed.

We show this more formally and illustrate this for the general adversary without side-information case of Section 5 of \cite{AH23}, and leave the other two cases for future work. The test score state for each round $i$ is given by
\begin{align*} 
    \mathcal{M}_{i}(\sigma_{R_{i-1}})_{W_{i}} = (1-\gamma)\ket{\bot}\bra{\bot} &+ \gamma \Pr_{C, \vec{z} \sim \mathcal{A}\brac{ C } }\sbrac{ \sum_{i} p_{C}(z_{i}) < \frac{bk}{N} } \ket{0}\bra{0} \\ &\qquad + \gamma \Pr_{C, \vec{z} \sim \mathcal{A}\brac{ C } }\sbrac{ \sum_{i} p_{C}(z_{i}) \geq \frac{bk}{N} }\ket{1}\bra{1},
\end{align*}
where $\vert\perp\rangle$ denotes the state on the test outcome register $W_i$ that the round is not a test round, $\vert 0\rangle$ is the test failed to pass the XEB test, and $\vert 1\rangle$ is the test succeeded. From this, we know that we have an entropy bound of
\[ H(A_{i}|C_{i})_{ \nu_{A_{i}C_{i}} } \geq \frac{B}{2} \frac{ b\Pr_{C, \vec{z} \sim \mathcal{A}\brac{ C } }\sbrac{ \sum_{i} p_{C}(z_{i}) \geq \frac{bk}{N} } - \epsilon - 1  }{b-1} = \frac{B}{2} \frac{ b \frac{ \bra{1}\nu_{W_{i}}\ket{1} }{ \gamma } - \epsilon - 1 }{ b-1 }. \]
Notice that what this means is that we have a min-tradeoff function $f_{\min}$ given by
\[ f_{\min}(p) = \frac{B}{2} \frac{ b \frac{p(1)}{\gamma} - \epsilon - 1 }{b-1}, \]
where $p$ is a probability distribution over register $W_i$ and $p(1)\equiv \pr{W_i=1}$. Note that this is the usual notion of min-tradeoff function with probability distributions as the argument, not the version with scores as the argument as required by \cite{AH23}. One can in principle carry out the analysis using the second type as well, but the two types coincide in the analysis for general adversary without side information since the score for each round is a probability.

Now, applying the usual entropy accumulation theorem \cite{DFR20}, we have that
\[ H_{\min}^{\epsilon_{s}}(Z_{1}^{m} | C_{1}^{m}T_{1}^{m} E)_{\rho_{ZCTE|\Omega}} \geq m \frac{B}{2}\frac{bq - \frac{b\delta}{\gamma} - \epsilon - 1 }{b - 1} - V \sqrt{m}\sqrt{ \log\frac{2}{\Pr\sbrac{\Omega}^{2}\epsilon_{s}^{2}} } \]
where $V = \log( 1 + 2^{kn} ) + \lceil \frac{Bb}{2\gamma(b-1)} \rceil$.

%% file: Files/instantiation_YZ.tex
\newcommand{\accept}{E_{\mathsf{ACC}}}
\newcommand{\guess}{E_{\mathsf{GSS}}}
\newcommand{\acceptfixed}{E_{\mathsf{ACC}}^{\hash,\ch}}
\newcommand{\guessfixed}{E_{\mathsf{GSS}}^{\hash,\ch}}
\newcommand{\Z}{\mathbb{Z}}
\newcommand{\hasht}{F}

\subsection{Instantiation using Error-Correcting Codes} \label{sec:inst_YZ}

In this section, we show how to instantiate our single-round compiler (\Cref{thm:sr_sec}) using the seminal work of \cite{YZ24}, which results in a single-round CVPV scheme secure in QROM assuming the Aaronson-Ambainis conjecture (\Cref{thm:cvpv_qrom}).

The following is taken verbatim from \cite{YZ24}, and is commonly known as the \emph{Aaronson-Ambainis Conjecture}, originally due to \cite{AA14}:
\begin{conjecture}[\cite{YZ24}]\label{conj:aa}
    Let $\eps, \delta > 0$. Given any quantum algorithm $\alice$ that makes $Q$ quantum queries to a random oracle $H: \bit^n \to \bit^m$, there exists a deterministic classical algorithm $\bob$ that makes $\poly(Q, m, \eps^{-1}, \delta^{-1})$ classical queries and satisfies \[
        \Pr_H \bracS{ \abs{ \pr{\alice^H() \to 1} - \bob^H()} \le \eps} \ge 1 - \delta.
    \]
\end{conjecture}

\subsubsection{Certified Randomness in QROM - Definitions}
In \cite{YZ24}, single-round certified randomness as we define (\Cref{def:cr}) is referred to as \emph{proof of min-entropy}. 

\par The quantum advantage of \cite{YZ24} can be viewed as a \emph{quantum-query advantage}, meaning they show a soundness gap between a quantum-query adversary and any bounded classical-query adversary. Given this definition, they show that any non-interactive quantum-query advantage gives certified randomness. Intuitively, a deterministic quantum algorithm would not have advantage over a classical-query algorithm, which can simulate a quantum-query algorithm assuming \Cref{conj:aa}.

\begin{definition}[Certified Randomness in QROM] \label{def:cr_qrom}
    Let $\hash$ be a random oracle. A PoQ protocol $\poq^\hash$ is said to have \emph{certified randomness} property in the quantum random oracle model if no pair of a QPT prover $P^\hash$ and an unbounded guesser $Q^\hash$ can succeed in the following security game with non-negligible probability: \begin{enumerate}
        \item The verifier $V^\hash$ of $\poq$ sends a challenge $\ch$ to both $P^\hash$ and $Q^\hash$.
        \item $P^\hash$ sends back an answer $\ans$ and $Q$ and outputs a guess $\ans'$.
        \item $(P^\hash,Q^\hash)$ win if $V^\hash$ accepts and $\ans = \ans'$.
    \end{enumerate}
\end{definition}

We also mention the relevant definition from \cite{YZ24} for completeness, presenting it in the operational (game-based) way we adopted for our other definitions.
\begin{definition}[Proof of Min-Entropy in (auxiliary-input) QROM \cite{YZ24}] \label{def:cr_yz}
    A (single-round) \emph{proof of min-entropy} in the quantum random oracle model is a PoQ protocol $\poq_h^\hash$, addittionally parametrized by a min-entropy threshold $h(\secparam)$, such that for any pair of a QPT prover $P^\hash$ and an unbounded guesser $Q^\hash$ the following is true: Consider the security game defined in \Cref{def:cr_qrom}. For any inverse polynomial $\delta$, with overwhelming probability over $(\hash, \ch)$ we have: \begin{itemize}
        \item either the verifier $V^\hash_h$ accepts with probability at most $\delta$,
        \item or the probability that $\ans = \ans'$ conditioned on $V^\hash_h$ accepting is at most $2^{-h}$.
    \end{itemize}
\end{definition}

\begin{remark}[Comparing Definitions] \label{rem:YZ_definition}
\Cref{def:cr_qrom} can be thought of as the game-based version of the definition of \cite{YZ24} in the auxiliary-input setting because we allow $Q^\hash$ to share an entangled state with $P^\hash$ that depends on $\hash$. Yet, \Cref{def:cr_yz} is stronger than \Cref{def:cr} in several aspects. First, in \Cref{def:cr_yz}, the adversary has negligible success probability for an overwhelming fraction of challenges. Secondly, \Cref{def:cr_yz} guarantees that the guesser $Q^\hash$ succeeds with exponentially small probability conditioned on $V^\hash$ accepting for an overwhelming fraction of $(G,\ch)$, whereas \Cref{def:cr_qrom} gives a guarantee for average $(G, \ch)$. Finally, in \Cref{def:cr_yz} $P_h^\hash$ can output an answer with $h$-bits of min-entropy for any choice of polynomial $h(\secparam)$. We show the formal implication between the two definitions in the proof of \Cref{thm:cvpv_qrom}.
\end{remark}

\subsubsection{Results}

The following result was shown by prior work:

\begin{theorem} \label{thm:cr_YZ}
    If \Cref{conj:aa} is true, then there exists a single-round proof of min-entropy in the quantum random oracle model.
\end{theorem}

This readily implies certified randomness with respect to \Cref{def:cr_qrom}.

\begin{corollary}[Certified Randomness in QROM \cite{YZ24}] \label{thm:YZ_CR}
If \Cref{conj:aa} is true, then there exists a certified randomness protocol in the quantum random oracle model.
\end{corollary}
\begin{proof}
    Let $\poq^\hash$ be a protocol that satisfies Definition 3.5 in \cite{YZ24}. We claim that it satisfies \Cref{def:cr_qrom} if we fix min-entropy as $h = \lambda$. We show a proof by contradiction.
    \par Suppose there exists an adversary $(P^\hash, Q^\hash)$ that wins the security game with non-negligible probability. Let $\accept$ be the event that $V^\hash$ accepts and $\guess$ be the event that $\ans = \ans'$. Similarly define $\acceptfixed$ and $\guessfixed$ as the same events for fixed $\hash, \ch$. Then, for an inverse polynomial $\delta(\secparam)$ and an infinite set $\Lambda \subset \Z^+$ we have $\pr{\accept \land \guess} \ge \delta$ for all security parameters $\secparam \in \Lambda$. By possibly shrinking $\Lambda$ we can assume that $\delta > 2^{-\secparam}$ for $\secparam \in \Lambda$. For the rest of the proof, we restrict our attention to $\secparam \in \Lambda$. 
    \par By union bound, $\pr{\acceptfixed \land \guessfixed} \ge \delta/2$ for at least $\delta/2$ fraction of $(\hash, \ch)$. In particular, $\pr{\acceptfixed} \ge \delta/2$ and $\pr{\guessfixed \vert \acceptfixed} \ge \delta/2 > 2^{-\secparam}$ for such $(\hash, \ch)$. However, this contradicts with the guarantee of $\poq^\hash$ because $\delta/2$ is non-negligible.
\end{proof}

\Cref{thm:YZ_CR} and \Cref{thm:sr_sec} combined yield the following theorem:

\begin{theorem}[Single-Round CVPV in QROM] \label{thm:cvpv_qrom}
    If \Cref{conj:aa} is true, then there exists a single-round CVPV protocol secure in the quantum random oracle model.
\end{theorem}

\subsubsection{Construction}

The reader may be curious to see what our single-round CVPV construction corresponding to \Cref{thm:cvpv_qrom} looks like when instantiated with real-world hash functions. We present it below at a high level for presentation purposes: 

\begin{construction} \label{constr:cvpv_qrom_YZ} 
Let $\hash_k$ and $\hasht_{k'}$ be cryptographic hash function families, and let $C$ be a suitably chosen error-correcting code.

\begin{enumerate}
    \item The verifiers publish the hash key $k$, and set $s = \hash_k(x\oplus y) \oplus k'$, where $x,y$ are random inputs.
    \item At $t=0$, $V_0$ sends $(x,s)$ and $V_1$ sends $y$ to the prover simultaneously.
    \item The honest prover, located at position $1$, computes $k' = \hash_k(x\oplus y) \oplus s$. Then, he computes a codeword $\ans = (x_1,x_2,\dots,x_m) \in C$ such that the $i$th bit of $\hasht_{k'}(x_i)$ equals $1$ for $i \in [m]$. He immediately sends $\ans$ to both verifiers.
    \item $V_0$ expects $\ans$ at time $t=2$. Similarly, $V_1$ expects $\ans'$ at time $t=2$.
    \item The verifiers accept iff $\ans = \ans'$ and $\ans = (x_1,x_2,\dots,x_m) \in C$ such that the $i$th bit of $\hasht_{k'}(x_i)$ equals $1$ for $i \in [m]$.
\end{enumerate}

\end{construction}

%% file: Files/instantiation_lwe.tex
\subsection{Instantiation using NTCFs} \label{sec:inst_lwe}

In this section, we show how to neatly relate the results of \cite{BCMVV21} and \cite{LLQ22} using our work. First, we note that we can more explicitly use the result of \cite{BCMVV21} to get CVPV from LWE. Second, we note that there is an implication in the reverse direction.

\subsubsection{CVPV from LWE in QROM}

The result below is implicit in the work of \cite{BCMVV21}, yet not formally stated before.

\begin{lemma} \label{lem:bcmvv21}
    Assuming LWE, there exists a multi-round certified randomness scheme which satisfies \Cref{def:comp_cr_mr_nocomm}.
\end{lemma}

In fact, the stronger version with regular (not weak) certified randomness follow from the analysis of \cite{BCMVV21}. Now, we could recover the result of \cite{LLQ22} using \Cref{thm:mr_rf_from_weak_cr}. Note that our compiler yields an entirely different (rapid-fire) construction.

\begin{corollary}[CVPV from LWE]
\label{cor:cvpv_from_lwe}
    If \Cref{lem:bcmvv21} is true, then there exists a secure CVPV scheme in the random oracle model.
\end{corollary}

\subsubsection{(Weak) Certified Randomness from LWE using CVPV}

Conversely, we can show \Cref{lem:bcmvv21} is true assuming the existence of a CVPV scheme secure under LWE. This could be of independent interest as it provides a formal argument as to the existence of an NTCF-based certified randomness protocol using our definition.

\begin{proof}[Proof of \Cref{lem:bcmvv21}.]
    Assuming LWE, there exists a multi-round CVPV scheme secure in the random oracle model\footnote{The results in the plain model require additional restrictions on the amount of entanglement shared between the adversaries.} due to \cite{LLQ22}. Using our compiler (\Cref{thm:cvpv_to_cr}), there exists a weak certified randomness protocol. Examining the compiler, we see that the resulting scheme uses the random oracle in a trivial way, which yields a weak certified randomness scheme in the plain model.
\end{proof}